\documentclass[11pt, a4paper]{article} 
\usepackage[english]{babel}
\usepackage[utf8]{inputenc}
\usepackage[a4paper]{geometry}

\usepackage{color}
\usepackage{graphicx}

\usepackage{mathrsfs}

\usepackage{float}

\usepackage{bbm}
\usepackage{enumerate}
\usepackage{tikz}
\usetikzlibrary{matrix}
\usepackage{pgfplots}
\usetikzlibrary{arrows.meta}

\usepackage[colorlinks=false,linkcolor=black,citecolor=black]{hyperref}

\usepackage[symbol]{footmisc}

\usepackage{amsmath}
\usepackage{amssymb}
\usepackage{amsthm}
\usepackage{stmaryrd}

\newcommand{\F}{\mathcal{F}}

\newcommand{\bP}{\mathbf{b}\mathcal{P}}

\newsavebox\dotbox
\sbox{\dotbox}{\(\displaystyle\bigodot\)}

\newcommand{\beao}{\begin{eqnarray*}}
\newcommand{\eeao}{\end{eqnarray*}\noindent}

\newcommand{\beam}{\begin{eqnarray}}
\newcommand{\eeam}{\end{eqnarray}\noindent}

\def\bbr{{\Bbb R}}   
\def\bbn{{\Bbb N}}

\def\bbf{{\Bbb F}}
\def\bbq{{\Bbb Q}}

\def\bbs{{\Bbb S}}

\newcommand{\eps}{{\varepsilon}}
\newcommand{\vp}{{\varphi}}

\newcommand{\Var}{{\rm Var}}

\newcommand{\ov}{\overline}
\newcommand{\un}{\underline}
\newcommand{\wh}{\widehat}
\newcommand{\wt}{\widetilde}

\newcommand{\ph}{\varphi}

\newcommand{\mal}{\stackrel{\mbox{\tiny$\bullet$}}{}}
\newcommand{\auf}{[\![}
\newcommand{\zu}{]\!]}

\DeclareMathOperator*{\uplim}{up-lim}

\theoremstyle{theorem}
\newtheorem{theorem}{Theorem}[section]
\newtheorem{lemma}[theorem]{Lemma}

\newtheorem{proposition}[theorem]{Proposition}

\newtheorem{note}[theorem]{Note}
\newtheorem{remark}[theorem]{Remark}
\newtheorem{example}[theorem]{Example}
\newtheorem{assumption}[theorem]{Assumption}
\theoremstyle{definition}
\newtheorem{definition}[theorem]{Definition}

\numberwithin{equation}{section}

\title{The fundamental theorem of asset pricing with and without transaction costs}
\author{Christoph Kühn\thanks{Institute of Mathematics, Goethe University Frankfurt, D-60054 Frankfurt a.M., Germany, e-mail: ckuehn@math.uni-frankfurt.de\\
{\em Acknowledgments.} I would like to thank Christoph Czichowsky and Alexander Molitor for fruitful discussions and
two anonymous referees for their valuable comments and suggestions.}} 
\date{} 

\begin{document}
	\maketitle
	\begin{abstract}
	We prove a version of the fundamental theorem of asset pricing~(FTAP) in continuous time that is based on the strict no-arbitrage condition and that is applicable to both frictionless markets and markets with proportional transaction costs. We consider a market with a single risky asset whose ask price process is higher than or equal to its bid price process.
    Neither the concatenation property of the set of wealth processes, that is used in the proof of the frictionless FTAP,
    nor some boundedness property of the trading volume of admissible strategies usually argued with in models with a nonvanishing bid-ask spread need to be satisfied in our 
    model. 
    \end{abstract}

\begin{tabbing}
strict no-arbitrage, no unbounded profit with bounded risk\\



 

\end{tabbing}	
	
	\section{Introduction}
In frictionless financial market models with finitely many assets and a single probability measure, the arbitrage theory can be considered to be fully understood in principle. The fundamental theorem of asset pricing by Delbaen and Schachermayer~\cite{delbaen.schachermayer.1998} states 
that a market model satisfies no free lunch with vanishing risk~(NFLVR) iff there exists an equivalent probability measure under which discounted asset prices are  
$\sigma$-martingales. A variant of this theorem by Yan~\cite{yan.1998} provides even true martingales by considering a credit line that is a multiple of the sum of all asset prices in the market. For a detailed discussion of the arbitrage theory in frictionless markets we refer the reader to Delbaen and Schachermayer~\cite{delbaen.schachermayer.2006} and Eberlein and Kallsen~\cite[Subsection~11.7]{eberlein.kallsen}.

The picture for models with proportional 
transactions costs that generalize frictionless markets by allowing for bid-ask spreads is different. 
The picture is clear-cut in finite discrete time and for a finite probability space: 
in a general ``currency model'', Kabanov and Stricker~\cite{kabanov.stricker.2001} show that no-arbitrage~(NA) is equivalent to the existence of a so-called consistent price system~(CPS), which is a multidimensional martingale under the objective probability measure taking values within the dual of the cone of solvent portfolios at each point in time. In the special case of only one risky asset that we consider in the present paper,
a CPS is a pair of an equivalent probability measure and a martingale under this measure that lies between the (discounted) bid and the ask price of the risky asset. 
%
%
For infinite probability spaces, this equivalence
fails; Schachermayer~\cite[Example 3.1]{schachermayer2004fundamental} provides an example for an arbitrage-free market which allows an approximate arbitrage, i.e., a nonzero and nonnegative portfolio which is the
limit in probability of a sequence of portfolios attainable from zero endowment, and
consequently a CPS cannot exist.
 This raises the obvious question under which stronger no-arbitrage conditions the existence of a CPS can be
guaranteed. Schachermayer~\cite{schachermayer2004fundamental} introduces the concept of robust no-arbitrage~(${\rm NA}^r$)  
-- a no-arbitrage condition which is robust with respect to small changes in the bid-ask
spreads. Loosely speaking, if the bid-ask spread (of a pair of assets) does not vanish,
there have to exist more favorable bid–ask prices, leading to a smaller spread, such
that the modified market still satisfies NA. Schachermayer~\cite{schachermayer2004fundamental} shows that ${\rm NA}^r$ implies
that the set of terminal portfolios attainable from zero endowment is closed in probability, and that ${\rm NA}^r$ is equivalent to the existence of a strictly consistent price system~(SCPS), that is, a martingale taking values within the
relative interior of the dual of the cone of solvent portfolios at each point in time.
For general probability spaces but only one risky asset (in addition to a bank account)
it is shown by Grigoriev~\cite{grigoriev.2005} that NA already implies the existence
of a CPS, although the set of attainable portfolios need not be closed in probability;
see also Bayraktar and Zhang~\cite{bayraktar.zhang.2016} for a different proof that holds in the more general framework of model uncertainty. 

Most of the literature on continuous time models is based on the ${\rm NA}^r$ concept. 
Guasoni, R\'asonyi, and Schachermayer~\cite{guasoni.r.s.2010} derive a FTAP for a continuous mid-price process and small deterministic transaction costs. This means that the equivalence between no-arbitrage and the existence of a CPS holds asymptotically for small transaction costs.
Guasoni, L{\'e}pinette, R{\'a}sonyi~\cite{guasoni2012fundamental} prove a FTAP under a continuous time extension of ${\rm NA}^r$ called 
robust no free lunch with vanishing risk~(RNFLVR). 
The condition states that the bid-ask market has to satisfy NFLVR also for a slightly more favorable bid and ask prices. 
The reduction of the spread is uniformly in time but 
not uniformly in the scenario. The closedness required for portfolio optimization is shown under similar conditions, see 
Czichowsky and Schachermayer~\cite{CzichowskySchachermayer2016}.

Introduced by Guasoni~\cite{guasoni.2006}, another popular sufficient condition to obtain an arbitrage-free model is stickiness. It is formulated in the special case that the investor pays deterministic transaction costs when she buys or sells at a stochastic (mid-)price process. Roughly speaking, the condition states that there are positive probabilities that the mid-price stays in neighborhoods of its starting value. 
The appeal of the condition is that it is satisfied for many of the stochastic processes usually considered in stochastic modeling (e.g. for fractional Brownian motion). On the other hand, it is of course far away from being necessary.
  
An alternative to ${\rm NA}^r$ is the strict no-arbitrage~(${\rm NA}^s$) condition introduced by
Kabanov, R\'asonyi, and Stricker~\cite{kabanov.rasonyi.stricker.2002}. Loosely speaking, a market model satisfies ${\rm NA}^s$ if any claim
which is attainable from zero endowment up to some intermediate time~$t$ and can be
liquidated in $t$ for sure can also be attained from zero endowment by trading at time~$t$
only. Property~${\rm NA}^s$ alone does not imply the existence of a CPS, see \cite[Example 3.3]{schachermayer2004fundamental} for the existence of an approximate arbitrage under ${\rm NA}^s$. For a detailed discussion, we refer to the monograph of Kabanov and Safarian~\cite{kabanov.safarian.2009}. More recently, K\"uhn and Molitor~\cite{kuehn.molitor.2019} introduced (in discrete time) a variant of ${\rm NA}^s$
that is called prospective strict no-arbitrage~(${\rm NA}^{ps}$). A market model satisfies ${\rm NA}^{ps}$ if any claim which
is attainable from zero endowment by trading up to some intermediate time~$t$ and can subsequently
be liquidated for sure can also be attained from zero endowment in the subsequent
periods (here, ``subsequent'' is not understood in a strict sense). This means that in
contrast to the ${\rm NA}^s$ criterion, one does not distinguished between a trade that can be
realized at time~$t$ and a trade from which one knows at time~$t$ for sure that it can be
realized in the future. ${\rm NA}^{ps}$ is slightly weaker than ${\rm NA}^r$, but it
guarantees that the set of terminal portfolios attainable from zero endowment is closed in probability
(see \cite[Theorem 2.6]{kuehn.molitor.2019}).

The aim of the present paper is to extend K\"uhn and Molitor~\cite{kuehn.molitor.2019} to continuous time. 
We consider a single risky asset with c\`adl\`ag bid and ask price processes that may or may not coincide depending on the scenario and time.  
An essential preparatory work is the paper by K\"uhn and Molitor~\cite{kuehn.molitor.2022}. Under no unbounded profit with bounded risk~(NUPBR)
for simple strategies it is shown that there exists a semimartingale price system, that is a semimartingale lying between the bid and ask price processes. Then, \cite{kuehn.molitor.2022} show how these semimartingales can be used to construct gains of general trading strategies that are not necessarily of finite variation. 
For the continuous time extension of \cite{kuehn.molitor.2019} that we consider in the present paper it is very natural to merge the 
${\rm NA}^{ps}$ condition with no unbounded profit with bounded risk~(NUPBR). The latter is needed for the frictionless FTAP
(the combination of NA and NUPBR is equivalent to NFLVR). Very loosely speaking, a market model satisfies the new condition that we call
prospective strict no unbounded profit with bounded risk~(${\rm NUPBR}^{ps}$) if the set of {\em cost values} of the portfolios that can be liquidated at maximal loss of $1$ is bounded in probability.
The cost value was introduced by Bayraktar and Yu~\cite{bayraktar.yu.2018} as 
the cost to enter the portfolio position. It is the counterpart of the liquidation value. In the special cases of a discrete time model or a frictionless model, ${\rm NUPBR}^{ps}$ coincide with ${\rm NA}^{ps}$ and NUPBR, respectively. As in almost all FTAPs we prove that the set of attainable terminal portfolios is closed. This property is of independent use.

In plain English, the key idea of our proof is that if a stock position is built up at a time with positive spread (in terms of the ``actual'' bid and ask prices), there is a positive worst-case risk 
taken by the investor that cannot be eliminated by smart trading at a later stage. This restricts the amount of shares that can be hold at that 
time. Of course, the restriction also depends on the current wealth of the investor (whereby we value positions at their cost/purchase price).
Once the number of assets is bounded on an interval, a semimartingale price system can be used to show that trading costs cannot explode.
The same holds for the total variation of strategies as long as the spread is bounded away from zero.
Then, we can apply the stochastic version of Helly's theorem by Campi/Schachermayer, and after passing to forward convex combinations, any sequence of strategies converges pointwise to a finite limit.
By contrast, the worst-case risk described above disappears if the spread is zero since then a stock can be liquidated again immediately at the same price as it has been purchased. But, for frictionless intervals we can apply the (completely different) results used for the proof of the frictionless FTAP by Delbaen and Schachermayer~\cite{delbaen.schachermayer.1998}. Here, one directly analyzes the wealth processes (without having a good control over the size of the strategies) and use that the investor can always switch between strategies without transaction costs.
    
Unfortunately, although the basic intuition described above sounds not too complicated the details are extremely technical.
The main difficulties arises from the transition of frictionless periods and periods with friction. This can occur continuously or
by jumps that cause different mathematical difficulties. 
We nevertheless hope to make the main ideas accessible to a wider readership. The above mentioned worst-case risk is a new approach to the problem that is also rather different to the arguments used under the RNFLVR condition. The latter use that any transaction leads to costs that are added to the gains coming from a more favorable but still arbitrage-free price system. This means that the starting point is to control the trading volume of the strategies.\\  

The rest of the paper is organized as follows.
In Section~\ref{5.6.2023.1}, we introduce the notation and the financial model, discuss the assumptions, and state the main results of the 
paper (Theorems~\ref{25.7.2022.1} and \ref{7.5.2022.2}). 
Section~\ref{6.6.2023.01} is devoted to their proofs. Appendix~\ref{5.6.2023.3} consists of auxiliary statements and their proofs that
are not directly linked to a financial application. 

\section{Definitions and Main Theorems}\label{5.6.2023.1}

Throughout the paper, we fix a terminal time~$T\in\bbr_+$ and a filtered 
probability\linebreak space~$(\Omega,\F, (\F_t)_{t\in[0,T]},P)$ satisfying the usual conditions. 
The predictable $\sigma$-algebra on $\Omega\times[0,T]$ is denoted by $\mathcal{P}$, 
the set of bounded predictable processes starting at zero by $\bP$.
A stopping time~$\tau$ is allowed to take the value~$\infty$, 
but $\auf\tau\zu:=\{(\omega,t)\in\Omega\times[0,T]: t=\tau(\omega)\}$. Especially,
we use the notation $\tau_A$, $A\in\mathcal{F}_\tau$, for the stopping time that
coincides with $\tau$ on $A$ and is infinite otherwise.
(In)equalities between stochastic processes are understood ``up to evanescence'', i.e., up to a global $P$-null set not depending on time.
The term~$\Var_a^b(X)$ denotes the pathwise variation of a process~$X$ on the interval $[a,b]$.
In the case that $P(\Var_a^b(X)<\infty)=1$, $X=X_a+X^\uparrow - X^\downarrow$ is
its Jordan-Hahn decomposition into two nondecreasing processes on $[a,b]$ with $X^\uparrow_a = X^\downarrow_a =0$.
A real-valued process~$X$ is called l\`agl\`ad if and only if all paths possess finite left and right limits (but they can have double jumps).
We set $\Delta^+ X:=X_+-X$ and $\Delta X:=\Delta^- X:=X-X_-$, where $X_{t+}:=\lim_{s\downarrow t} X_s$ and $X_{t-}:=\lim_{s\uparrow t} X_s$
(with the convention $X_{T+}:=X_T$ and $X_{0-}:=X_0$). 
For a random variable~$Y$, we set $Y^+:=\max(Y,0)$ and $Y^-:=\max(-Y,0)$. 
The standard stochastic integral as defined in Jacod and Shiryaev~\cite[Definition~III.6.17]{jacod.shiryaev}
is denoted by $\vp\mal S$ for $\vp\in L(S)$ with the convention $\vp\wt{\vp}\mal S=(\vp\wt{\vp})\mal S$. It does not cause any ambiguity that by $\vp\mal S$ we also denote the integral 
of an almost simple integrand (cf., e.g., \cite[Definition~3.15]{kuehn.molitor.2022}) of the form, e.g., $\vp=1_{\zu \tau_1,\tau_2\auf}$ 
with respect to a l\`agl\`ad process~$S$ (not necessarily a semimartingale and not even right-continuous). The integral reads $\vp\mal S:=(S^{\tau_2-}-S^{\tau_1})1_{\{\tau_2>\tau_1\}}$ and analog definitions are canonical (the integral does not allow to ``invest'' separately 
in $\Delta^+ S$).   
For l\`agl\`ad processes $X$ and $Y$ we define the metric~$d_{up}(X,Y):=E(\sup_{t\in[0,T]}|X_t-Y_t|\wedge 1)$ that metrizes the convergence 
``uniformly in probability''. For semimartingales $X$ and $Y$ (that are by definition c\`adl\`ag) the \'Emery metric is defined by 
$d_{\bbs}(X,Y):=\sup_{H\in\bP, ||H||_\infty\le 1}E(\sup_{t\in[0,T]}|H\mal (X-Y)_t|\wedge 1)$ that metrizes convergence in the semimartingale topology.\\

The financial market consists of one risk-free asset or bank account that does not pay interest and one risky asset with bid price 
$\un{S}$ and ask price $\ov{S}$ expressed in units of the risk-free asset. We assume that $\un{S}$ and $\ov{S}$ are adapted c\`adl\`ag 
processes with 
\beam\label{11.6.2023.1}
0\le \un{S}\le \ov{S}\quad\mbox{and}\quad \ov{S}_T>0.
\eeam
The second condition is made to avoid confusion. Namely, a vanishing ask price would already lead to an arbitrage in the multidimensional sense
(see Definition~\ref{4.4.2023.1} below) but not necessarily in the one-dimensional sense, which is considered when the market is frictionless.
\begin{definition}\label{5.4.2023.1}
The {\em actual bid and ask price processes} are defined as the c\`adl\`ag versions of  	
$\un{X}_t={\rm ess inf}_{\mathcal{F}_t}\sup_{u\in[t,T]}\un{S}_u$ and  $\ov{X}_t={\rm ess sup}_{\mathcal{F}_t}\inf_{u\in[t,T]}\ov{S}_u$	(whose existence is shown in Proposition~\ref{23.7.2022.1} that also provides the precise definition of the conditional essential infimum/supremum). The {\em actual bid-ask spread} is denoted by $X:=\ov{X}-\un{X}$.
\end{definition}
In discrete time transaction costs models, these processes that take certain future trading opportunities into account
have already proven to be very useful. For the case of only one risky asset, which we consider in this paper, Sass and Smaga~\cite[equation before Lemma~4.1]{sass.smaga.2014} introduce them by a backward recursion. We leave it as an easy exercise to the reader to prove (by induction on the number of periods) that in (finite) discrete time the processes from Definition~\ref{5.4.2023.1} coincide with the processes in \cite{sass.smaga.2014}.   
In the general multidimensional Kabanov model, the (discrete time) actual bid and ask price processes correspond to the set-valued processes 
that are constructed in Rokhlin~\cite[page 95]{rokhlin.2008} by a more general backward recursion.

The random variable $\un{X}_t$ is the highest price at which the investor can liquidate the stock at present or in the future for sure -- with the information she has at time~$t$.
One has $\un{X}_t\ge \un{S}_t$. When the inequality is strict, it is silly to liquidate the position right now. The continuous time counterpart of the ${\rm NA}^{ps}$ 
condition from \cite{kuehn.molitor.2019} has to be expressed in terms of $(\un{X},\ov{X})$ since it captures trading opportunities in the future.	

To work with the processes $\un{X}$ and $\ov{X}$ is strongly related to freezing a portfolio position as
it is done in Guasoni, L{\'e}pinette, and R{\'a}sonyi~\cite{guasoni2012fundamental}. They start with almost simple strategies and consider associated portfolios that are frozen after a transaction of the original portfolio if a better liquidation price can be achieved for sure in the future. By introducing $\un{X}$ and $\ov{X}$ we can directly work with general strategies that shortens the proofs. The relation to \cite{guasoni2012fundamental} is discussed further after Definition~\ref{4.4.2023.1}.
%
%

{\em We decided to pass already now to the actual bid and ask price processes} because in the spirit of \cite{guasoni2012fundamental}, the admissibility condition
has anyhow to be expressed in terms of these processes, and also a meaningful definition of a ``frictionless interval'' is in terms of the actual bid and ask prices~$\un{X},\ov{X}$ rather than in terms of $\un{S},\ov{S}$ themselves. 
But, for the motivation of the admissibility condition, we keep the original processes $\un{S}$ and $\ov{S}$ in mind.
We note that the ``actual actual bid and ask prices'' are just the actual bid and ask prices. For the price processes that are commonly considered we have that $\un{X}=\un{S}$ and $\ov{X}=\ov{S}$. 


The arbitrage theory in continuous time frictionless markets is based on the set of general trading strategies, i.e., on the set of predictable processes~$L(S)$ which are integrable against the semimartingale~$S$ modeling the stock price.
This set was generalized to models with transaction costs beyond efficient friction by K\"uhn and Molitor~\cite{kuehn.molitor.2022}.
As in frictionless markets, but in contrast to models with efficient friction, the set contains strategies of infinite variation.
Accordingly, the present paper is based on this set and the corresponding self-financing condition.  We only outline the definitions that are needed to follow the present paper. 
{\bf For the rest of the paper, we assume that there exists a semimartingale price system}, i.e., a semimartingale~$S$ such that 
\beam\label{17.6.2023.3}
\un{X}\le S\le\ov{X}.
\eeam
\begin{remark}
Of course, (\ref{17.6.2023.3}) includes the condition that $\un{X}\le \ov{X}$. This condition is violated in a market with
$\un{S}=\ov{S}$ and $\un{S}$ generated by an unfavorable doubling strategy such that $\un{S}_0=2$ and $\un{S}_T=1$. Then, one has $\un{X}=\un{S}>\ov{X}=1$ on a set that is not evanescent. But, the model is only arbitrage-free with the bank account as num\'eraire. In the present paper, we work with a num\'eraire-free no-arbitrage condition that rules out such price processes and perfectly fits to actual bid and ask price processes.

Under a mild NUPBR-condition for simple strategies in the bid-ask model it is shown in K\"uhn and Molitor~\cite[Theorem~2.7]{kuehn.molitor.2022}  
that a semimartingale price system exists. To motivate (\ref{17.6.2023.3}) we can apply this theorem to $\un{X}$, $\ov{X}$. The assumptions we need 
in the present paper to establish a FTAP are stronger -- even if they were only required for simple strategies. Thus, (\ref{17.6.2023.3}) is no further restriction. 
\end{remark}
For the extension of the self-financing condition to general strategies, also the following assumption is needed.
\begin{assumption}\label{20.8.2022}
	For every $(\omega,t)\in \Omega\times [0,T)$ with $\ov{X}_t(\omega)=\un{X}_t(\omega)$
	there exists an $\eps>0$ such that
	$\ov{X}_s(\omega)=\un{X}_s(\omega)$ for all $s\in(t,(t+\eps)\wedge T)$ or 
	$\ov{X}_s(\omega)>\un{X}_s(\omega)$, $\ov{X}_{s-}(\omega)>\un{X}_{s-}(\omega)$  for all $s\in(t,(t+\eps)\wedge T)$.
\end{assumption}
This means that each zero of the path $t\mapsto \ov{X}_t(\omega)-\un{X}_t(\omega)$ is either an inner point from the right 
of the zero set or a starting point of an excursion away from zero. We note that there starts no excursion of the c\`adl\`ag function~$f(t):=\sum_{n=1}^\infty(2^{-n}-t)1_{(2^{-(n+1)}\le t<2^{-n})}$ at $t=0$ since there are zeros of the left limit. Assumption~\ref{20.8.2022} incorporating left limits is already needed in the proof of \cite[Lemma~5.1]{kuehn.molitor.2022} to ensure that the constructed excursion has a positive length.

Under Assumption~\ref{20.8.2022} it is shown in \cite[Lemma~5.1 and Lemma~5.2]{kuehn.molitor.2022} that 
there exist sequences of stopping times~$(\tau_1^i)_{i\in\bbn}$ and $(\sigma_1^i)_{i\in\bbn}$ that exhaust the set of starting times of the
excursions of the spread away from zero and the set of starting times of the frictionless intervals, respectively. Hence, there is a decomposition
\beam\label{24.3.2023.1}
& & \Omega\times[0,T] = \cup_{i\in\bbn} (\mathcal{I}^c_i\cup\mathcal{I}^{fc}_i)\quad \mbox{up to evanescence, where}\nonumber\\ 
& & \mathcal{I}^c_i:=\zu (\tau_1^i)_{\{X_{\tau_1^i}=0\}},\Gamma(\tau_1^i)\zu\setminus\auf (\Gamma(\tau^i_1))_{\{X_{\Gamma(\tau^i_1)-}=0\}}\zu\in\mathcal{P},\nonumber\\
& & \mathcal{I}^{fc}_i:=\auf (\sigma^i_1)_{\{X_{\sigma^i_1-}=0\}}\zu \cup \zu \sigma_1^i,\Lambda(\sigma_1^i)\zu\nonumber\\
& & \qquad\quad \cup(\zu (\Lambda(\sigma_1^i))_{\{X_{\Lambda(\sigma_1^i)}>0\}},\Gamma(\Lambda(\sigma_1^i))\zu\setminus
\auf(\Gamma(\Lambda(\sigma_1^i)))_{\{X_{\Gamma(\Lambda(\sigma_1^i))-}=0\}}\zu)\in\mathcal{P},\ i\in\bbn,
\eeam
and the stochastic intervals can be chosen such that they are disjoint. Here, $\Gamma(\tau_1^i):=\inf\{t>\tau_1^i : X_t =0\ \mbox{or}\ X_{t-}=0\}$ denotes the end time of the excursion starting in $\tau_1^i$, and $\Lambda(\sigma_1^i):=\inf\{t\ge \sigma_1^i : \exists \eps>0\ \forall s\in(t,(t+\eps)\wedge T)\ X_s>0\}$
%
%
is the starting time of the next excursion after $\sigma_1^i$.  
 The interval~$\mathcal{I}^c_i$ is with ``costs'', and the interval $\mathcal{I}^{fc}_i$ starts with  a ``frictionless'' period that is followed by a period with ``costs'' iff at the end of the frictionless interval the spread is already positive. 
By construction, the investor can rebalance her portfolio at the boundaries of $\mathcal{I}^c_i$ and $\mathcal{I}^{cf}_i$ without costs (since the spread can jump away from zero at an {\em unpredictable} stopping time, this would not be the case if a frictionless period were not sometimes followed by a period with costs). For $X_{\tau_1^i}>0$ there must exist a frictionless forerunner (possibly only consisting of the single point $\tau_1^i$). Then, the excursion is included in some $I^{fc}_j$.
{\em For notational convenience, we fix the sequences in (\ref{24.3.2023.1}) for the rest of the paper, but we stress that the definitions below do not depend on their choice.} 

For a $\vp\in\bP$, specifying the number of risky assets the investor holds in her portfolio, \cite[equation before (3.9)]{kuehn.molitor.2022} constructed the corresponding self-financing position in the bank account by the $[-\infty,\infty)$-valued predictable process
\beam\label{16.6.2023.1}
\Pi(\vp) := \vp^+\mal S -\vp^-\mal S' - C^S(\vp^+) - C^{S'}(-\vp^-) -\vp^+S +\vp^-S'.
\eeam
Here, $S$ and $S'$ are arbitrary semimartingale price systems. 
The nondecreasing processes $C^S$ and $C^{S'}$ model accumulated costs that occur when trades are executed at the less favorable bid and ask prices but positions are evaluated with $S$ and $S'$, respectively. The key discovery was that $\Pi(\vp)$ does not depend on the 
choice of the semimartingale price system (see \cite[Corollary~3.22]{kuehn.molitor.2022}). We consider the construction directly for the actual 
bid and ask prices~$\un{X},\ov{X}$ and not for $\un{S},\ov{S}$ as in \cite{kuehn.molitor.2022}. We apply the cost term~$C$ that is $[0,\infty]$-valued separately to $\vp^+$ and $-\vp^-$ and set $C^{S,S'}(\vp):=C^S(\vp^+)+C^{S'}(-\vp^-)$. This definition makes sense: if one applied $C$ to $\vp\in\bP$ with only one semimartingale, one would obtain $C(\vp^+)+C(-\vp^-)$ (see \cite[Step~3 in the proof of Theorem~4.5]{kuehn.molitor.2022}). The intuition for this is that $\vp^+$ and $-\vp^-$ never trade in the opposite direction. Thus, executing them through different trading accounts
does not yield higher costs. 
Although the choice of the semimartingale is irrelevant for bounded strategies, we provide more flexibility for the extension to unbounded strategies by allowing different semimartingales for long and short positions. Since in \cite{kuehn.molitor.2022} $S'=S$, one has to check that the proof of \cite[Corollary~3.22]{kuehn.molitor.2022} still holds for different semimartingale price systems for the positive and negative part. But this obviously follows by the above mentioned decomposition of the costs into $C(\vp^+)$ and $C(-\vp^-)$. The process
\beao
V^{S,S'}(\vp) :=  \Pi(\vp) + \vp^+ S - \vp^- S' = \vp^+\mal S -\vp^-\mal S' - C^S(\vp^+) - C^{S'}(-\vp^-)
\eeao
is the wealth process if long positions are evaluated by $S$ and short positions by $S'$. 
The wealth process~$V^{S,S'}$ is written as a function of $\vp$ only since the increments of the bank account~$\vp^0$ result from the self-financing condition.
This direct relation exists only for the {\em increments} of $\vp$ and $\vp^0$, which is why the initial values are set to zero. They would have to be modeled separately. With this background information, one can follow the present paper to a large extent without knowing the details of 
\cite{kuehn.molitor.2022}. For a complete understanding the reader is referred to the construction of the cost term 
in \cite[Subsections~3.1 and 3.2]{kuehn.molitor.2022}.

\begin{definition}\label{13.1.2023.3}
Let $S$ be a c\`adl\`ag process. The lower and upper predictable envelopes of $S$, denoted by ${\rm ess inf}_{\mathcal{F}_-}S$ 
and ${\rm ess sup}_{\mathcal{F}_-}S$, are defined as the unique predictable processes such that
$({\rm ess inf}_{\mathcal{F}_-}S)_\tau = {\rm ess inf}_{\mathcal{F}_{\tau-}}S_\tau$ and
$({\rm ess sup}_{\mathcal{F}_-}S)_\tau = {\rm ess sup}_{\mathcal{F}_{\tau-}}S_\tau$ a.s., respectively, for all predictable stopping times~$\tau$
(the definition makes sense by Proposition~\ref{13.1.2023.1}).
\end{definition}
Predictable envelopes are needed for a consistent valuation of portfolio positions {\em after} the portfolio is rebalanced at some time~$t$ but
{\em before} prices move at $t$. Looking at both a semimartingale price system~$S$ and its envelope can be seen as a splitting of time.
In frictionless markets these envelopes are not needed since the wealth only changes due to price movements but not 
due to portfolio rebalancing. 
The example below that is adapted from Larsson~\cite{larsson} shows that $S_-$ generally does not do the job
well enough since $S_t$ can exceed $S_{t-}$ for sure.
%
%
\begin{example}[Example~2.11 in Larsson~\cite{larsson}]\label{15.1.2023.3} 
Let $t_1\in(0,T)$ and $B$ be a standard Brownian motion.
Consider the bid price process $\un{S}_t= (|B_{t_1}|+B_t-B_{t_1}) 1_{[t_1,T]}(t)$ that coincides with its actual bid price process, i.e., 
$\un{X}=\un{S}$. Nevertheless, $\lim_{t\uparrow t_1}\un{X}_t$ is zero and differs from ${\rm ess inf}_{\mathcal{F}_{t_1-}}\un{X}_{t_1}=|B_{t_1}|$.

If $\un{X}$ were a frictionless price process, there would be an arbitrage, but as part of a bid-ask model it can make perfect sense.
Now, consider an investor who buys $\vp_{t_1}$ stocks at the ask price $\ov{X}_{t_1-}$ with the information~$\mathcal{F}_{t_1-}$.
The quantity $\ov{X}_{t_1-}-{\rm ess inf}_{\mathcal{F}_{t_1-}}\un{X}_{t_1}$ is the minimal worst-case loss per share she takes 
by building up this position. Neither $\un{X}_{t_1-}$ nor $\un{X}_{t_1}$ could provide this information 
\end{example} 
The set of unbounded strategies to which (\ref{16.6.2023.1}) is extended is slightly different to \cite{kuehn.molitor.2022}. The main difference is that we do not require that the up-convergence of wealth processes holds ``globally'' over all (countable many) excursions of the spread away from zero. 
%
%
In the special case of a frictionless market, the set of course coincides with the set of integrable processes in the semimartingale sense (see Proposition~\ref{15.1.2023.4}). 
\begin{definition}\label{def:DefinitionL}
	Let $L(\un{X},\ov{X})$ denote the subset of real-valued, predictable processes~$\varphi$ such that there exists a sequence $(\varphi^n)_{n\in\bbn}\subseteq (\bP)^\Pi:=\{\psi\in\bP\ :\ \Pi(\psi)>-\infty\}$ with 
	\begin{itemize}
		\item[(a)] $\vp^n\to\vp$ pointwise on $\Omega\times [0,T]$ and $(\vp^n)^+\leq \vp^+$, $(\vp^n)^-\leq \vp^-$ for all $n\in\bbn$,
		\item[(b)] there exist semimartingales $S$, $S'$ with $\un{X}\leq S\leq \ov{X}$, $\un{X}\leq S'\leq \ov{X}$, an optional l\`agl\`ad wealth process~$V$ satisfying
\beam\label{3.6.2023.2}           
(1_{J_n}\mal V - V) 1_{\{X=0\}},\ (1_{J_n}\mal V_- - V_-) 1_{\{X_-=0\}}\to 0\quad \mbox{in $d_{up}$ as 
$n\to\infty$ for all}\nonumber\\
\mbox{sequences}\ (J_n)_{n\in\bbn}\ \mbox{of finite unions of}\ (I^c_i)_{i\in\bbn}, (I^{fc}_i)_{i\in\bbn}\nonumber\\ 
\mbox{with}\ 1_{J_n}\to 1_{\Omega\times[0,T]}\ \mbox{up to evanescence as\ }n\to\infty,  
\eeam		
and predictable processes ${}^p\!S$ and ${}^p\!S'$ with ${\rm ess inf}_{\mathcal{F}_-}S \le {}^p\!S$, ${}^p\!S'\le {\rm ess sup}_{\mathcal{F}_-}S'$  such that
\beam\label{27.9.2022.1}
1_{\mathcal{I}^c_i}\mal V^{S,S'}(\vp^n) \to 1_{\mathcal{I}^c_i}\mal V,\quad
1_{\mathcal{I}^{fc}_i}\mal V^{S,S'}(\vp^n) \to 1_{\mathcal{I}^{fc}_i}\mal V\quad \mbox{in $d_{up}$ as}\ n\to\infty,
\eeam
\beam\label{27.9.2022.2}
& \mbox{and} & \big(1_{\mathcal{I}^c_i}\mal(\Pi(\vp^n) + (\vp^n)^+\, {}^p\!S -(\vp^n)^-\, {}^p\!S')\big)_{n\in\bbn},\nonumber\\ 
& & \big(1_{\mathcal{I}^{fc}_i}\mal(\Pi(\vp^n) + (\vp^n)^+\, {}^p\!S -(\vp^n)^-\, {}^p\!S')\big)_{n\in\bbn}\ 
\mbox{are $d_{up}$-Cauchy}\ \mbox{for all}\ i\in\bbn.
\eeam
Furthermore, for all competing sequences $(\widetilde{\vp}^n)_{n\in\bbn}\subseteq (\bP)^\Pi$ satisfying (a) and for all $i\in\bbn$, 
there exists a (deterministic) subsequence $(n_k)_{k\in\mathbb{N}}$ such that
		\begin{align}\label{23.8.2020.1}
			\left(1_{\mathcal{I}^c_i}\mal(V^{S,S'}(\widetilde{\vp}^{n_k})-V^{S,S'}(\vp^{n_k}))\right)^+,
			\left(1_{\mathcal{I}^{fc}_i}\mal(V^{S,S'}(\widetilde{\vp}^{n_k})-V^{S,S'}(\vp^{n_k}))\right)^+
				 \to 0,\quad k\to\infty,
		\end{align}
up to evanescence. 
	\end{itemize}
We extend the self-financing operator $\Pi$ to $L(\un{X},\ov{X})$ by setting
\begin{align}\label{25.3.2022}
\Pi(\vp):=V-\vp^+ S +\vp^- S',\quad \vp\in L(\un{X},\ov{X}).
\end{align}
A self-financing strategy is a pair~$(\vp^0,\vp)$ of predictable processes specifying the number of bonds and stocks
in the portfolio such that $\vp-\vp_0\in L(\un{X},\ov{X})$ and $\vp^0=\vp^0_0+\Pi(\vp-\vp_0)$.
When the term ``bounded strategy'' is used, it refers only to the position in the stock, unless otherwise stated.
\end{definition}
\begin{proposition}\label{19.6.2023.3}
$\Pi(\vp)$ is well-defined, i.e., it does not depend on the choice of $S$, $S'$, and $V$. In addition, one has that $(\bP)^\Pi\subseteq L(\un{X},\ov{X})$.
\end{proposition}
\begin{proof}
(i) Let $V^1$ and $V^2$ be wealth processes that result from different approximating strategies with regard to different semimartingales $S^1$, $(S')^1$ and $S^2$, $(S')^2$, respectively. From \cite[Proposition~4.2]{kuehn.molitor.2022} applied to the strategy $\vp 1_{J_n}$ it follows that 
\beao
1_{J_n}\mal V^1=1_{J_n}\mal V^2 + \vp^+(S^1-S^2) - \vp^-((S')^1-(S')^2)\quad\mbox{on}\ J_n\quad\mbox{for all}\ J_n\ \mbox{as in}\ (\ref{3.6.2023.2}).  
\eeao
For this, it is crucial that by construction, at the boundaries of $I^c_i$ and $I^{fc}_i$ the portfolio can be rebalanced without costs
(thus, the choice of the semimartingales does not matter there). 
By (\ref{3.6.2023.2}) and (\ref{27.9.2022.1}) (the latter is only needed for the last excursion) 
it follows that $V^1=V^2 + \vp^+(S^1-S^2) - \vp^-((S')^1-(S')^2)$ up to evanescence.

(ii) For the second assertion we take a $\vp\in(\bP)^\Pi$ and consider it as a constant sequence. The semimartingale price systems are arbitrarily
chosen and $V:=V^{S,S'}(\vp) =  \vp^+\mal S -\vp^-\mal S' - C^S(\vp^+) - C^{S'}(-\vp^-)$. Condition~(\ref{3.6.2023.2})  
holds since $1_{J_n}\mal (\vp^+\mal S)\to \vp^+\mal S$, $1_{J_n}\mal (\vp^-\mal S')\to \vp^-\mal S'$, $1_{J_n}\mal C^S(\vp^+)\to C^S(\vp^+)$, and $1_{J_n}\mal C^{S'}(-\vp^-)\to C^{S'}(-\vp^-)$ converge separately in $d_{up}$ as $n\to\infty$ by the dominated convergence theorem for stochastic integrals (cf., e.g. \cite[Theorem~12.4.10]{cohen.elliott.2015}). For the cost processes that are only l\`agl\`ad, we refer to the notation
and the fact that there are no costs at the boundaries of $J_n$.
Condition~(\ref{23.8.2020.1}) is already shown in \cite[Corollary~3.24]{kuehn.molitor.2022}.
\end{proof}
\begin{remark}
%
%
The combination of (\ref{3.6.2023.2}) and (\ref{27.9.2022.1}) does not imply up-convergence of $(V^{S,S'}(\vp^n))_{n\in\bbn}$ to $V$ on the whole time interval. Namely, we need not have control over the wealth processes during the excursions uniformly in time.  
\end{remark}

\begin{remark}
In frictionless markets, the approximating sequence~$(\vp^n)_{n\in\bbn}\subseteq\bP=(\bP)^\Pi$ from Definition~\ref{def:DefinitionL} can be chosen such that
$\vp^n\mal S\ge -1$ for all $n\in\bbn$ if $\vp\mal S\ge -1$. To see this, one considers for arbitrary $\eps>0$ the stopping times 
$\tau_n:=\inf\{t\ge 0 : \vp^n\mal S_t \le \vp\mal S_t - \eps\}$ and uses that 
$d_{up}(\vp^n\mal S,\vp\mal S)\to 0$, $(\vp^n)^+\le \vp^+$, and $(\vp^n)^-\le \vp^-$. 
Consequently, the approximating bounded (not elementary!) strategies can be chosen to satisfy the admissibility conditions of the 
unbounded strategy. By condition~(\ref{27.9.2022.2}) that allows for a consistent valuation of the portfolio after the trade but before the price movement (cf. Example~\ref{15.1.2023.3}), the same property can be shown in the bid-ask model (see Lemma~\ref{12.6.2022.1}). 
Condition~(\ref{27.9.2022.2}) is weak in the sense that it only requires minimal consistency for the intermediate valuation of long and short positions by ${}^p\!S$ and ${}^p\!S'$, respectively.
\end{remark}

\begin{proposition}\label{15.1.2023.4}
In the special case of a frictionless market, in the sense that $\un{X}=\ov{X}$, condition~(\ref{27.9.2022.1}) implies condition~(\ref{27.9.2022.2}) with
the choice ${}^p\!S:=S_-\vee{\rm ess inf}_{\mathcal{F}_-}S$ and ${}^p\!S':=S'_-\wedge{\rm ess sup}_{\mathcal{F}_-}S'$, and we have $L(\un{X},\un{X})=L(\un{X})$.
\end{proposition}
%
%
\begin{proof}[Proof of Proposition~\ref{15.1.2023.4}]
It is sufficient to show the first assertion. Then, the second one follows along the lines of the proof of \cite[Proposition~4.3]{kuehn.molitor.2022}.

Let $S:=S':=\un{X}$ and $(\vp^n)_{n\in\bbn}\subseteq (\bP)^\Pi$ be a sequence such that 
$\sup_{t\in[0,T]}|(\vp^n-\vp^m)\mal S_t|\to 0$ in probability as $n,m\to\infty$.
This implies that $\sup_{t\in[0,T]}|((\vp^n_t)^+ - (\vp^m_t)^+)\Delta S_t| + \sup_{t\in[0,T]}|((\vp^n_t)^- - (\vp^m_t)^-)\Delta S'_t|\to 0$ in probability. 
%
%
From (\ref{15.1.2023.2}) it follows that $\{{}^p\!S>S_-\}\cap\{{}^p\!S>S\}$ and $\{{}^p\!S'<S'_-\}\cap\{{}^p\!S'<S'\}$ are evanescent, 
and we obtain that
$\sup_{t\in[0,T]}|((\vp^n_t)^+-(\vp^m_t)^+)(\, {}^p\!S_t -S_{t-})| + \sup_{t\in[0,T]}|((\vp^n_t)^-_t-(\vp^m_t)^-)(\, {}^p\!S'_t - S'_{t-})|\to 0$ in probability.
Since $\Pi_t(\vp^n) + \vp^n_t S_{t-} = \vp^n\mal S_{t-}$ and $\sup_{t\in[0,T]}|(\vp^n-\vp^m)\mal S_{t-}|\to 0$ in probability as $n,m\to\infty$, we are done.
\end{proof} 

Our admissibility condition is in the spirit of Guasoni, L{\'e}pinette, and R{\'a}sonyi~\cite{guasoni2012fundamental}.
For a motivation we refer to \cite[Proposition~4.9]{guasoni2012fundamental} and Lemma~\ref{7.5.2022.1} (and the text before the lemma).
\begin{definition}\label{26.7.2022.01}
Let $M\in\bbr_+$. A self-financing strategy~$(\vp^0,\vp)$ is called $M$-admissible iff $\vp^0+M+(\vp+M)^+ \un{X} -(\vp + M)^-\ov{X} \ge 0$. We write $(\vp^0,\vp)\in\mathcal{A}^M$. 
A strategy is admissible iff it lies in $\mathcal{A}:=\cup_{M\in\bbr_+}\mathcal{A}^M$. We denote by $\mathcal{A}^M_0$ and $\mathcal{A}_0$ the corresponding sets of strategies which start at $(0,0)$. The liquidation value process  of $(\vp^0,\vp)$ is defined as
$V^{\rm liq}(\vp) := \vp^0 + \vp^+ \un{X} -\vp^-\ov{X}$.
\end{definition}
The following example shows that in Definition~\ref{26.7.2022.01} the actual bid and ask prices~$(\un{X},\ov{X})$ cannot be replaced by the original bid and ask prices~$(\un{S},\ov{S})$ if the set of terminal (liquidation) values that can be achieved by an admissible strategy should be Fatou-closed (cf., Theorem~\ref{7.5.2022.2} for a definition). This is an insight of \cite{guasoni2012fundamental}, but we did not find an explicit example in the literature. 
%
%
\begin{example}[Admissibility]
Let $T=2$, $B$ be a standard Brownian motion, and $U$ be a random variable that is uniformly distributed on $(0,1)$ and independent of $B$. 
Define $m_t:=\inf_{0\le s\le t}B_s$ and $\tau:=\inf\{ t\ge 0 :  m_t = -U\}\wedge 1$. 
The original ask price is given by $\ov{S}_t:=3+B_{t\wedge\tau}$ and the bid price by $\un{S}_t:=1_{\{t<\tau\}} + (\ov{S}_\tau+2)/2 1_{\{t\ge \tau\}}$.
The filtration is generated by $B^\tau$ and augmented by null sets. We have that $\ov{X}=\ov{S}$ and $\un{X}_t= 2\cdot 1_{\{t<\tau\}} + (\ov{S}_\tau+2)/2 1_{\{t\ge \tau\}}$.
The only non-silly investment strategies are to buy stocks before $\tau$ has occurred and sell them at time $\tau$ (or later). The market satisfies the RNFLVR condition in \cite[Definition~5.2(ii)]{guasoni2012fundamental} (one considers the ask and bid price processes $3/4 \ov{S}_t + 1/4\cdot 2$ and 
$3/2 1_{\{t<\tau\}} + (2/3 \ov{S}_\tau + 1/3\cdot 2) 1_{\{t\ge \tau\}}$ that are uniformly in time strictly more favorable 
than $\ov{S}$, $\un{S}$ by $P(-U>-1)=1$).

We define the function~$f(x):=(1-x)^{-1/2} -1$ for $x\in[0,1)$ and consider the sequence of nondecreasing strategies~$\vp^n_t:=f((-m_{t\wedge \tau})\wedge(1-1/n))1_{\{0<t\le \tau\}}$, $n\in\bbn$. 
The stock position is liquidated before $T$, and the terminal bond position results
\beao
\vp^{0,n}_T=f((-m_\tau)\wedge (1-1/n))\frac{1+B_\tau}2 -\int_0^{(-m_\tau)\wedge (1-1/n)}(1-x)\,df(x) \ge -\int_0^1(1-x)\,df(x). 
\eeao 
Following the strategies~$(\vp^n)_{n\in\bbn}$, stocks are purchased when $\ov{S}$ attains its running minimum (above $3-U$), and the amount of stocks explodes when $\ov{S}$ approaches $2$ before $\tau$ stops.
But, a share purchased at price~$2+\eps$ and liquidated at time~$T$ cannot produce losses larger than $\eps$. 
The example was chosen such that $M:=\int_0^1(1-x)\,df(x)<\infty$. This means that $\vp^{0,n}_T\ge -M\in\bbr$ for all $n\in\bbn$.
We even have that $(\vp^{0,n},\vp^n)$ is $M$-admissible in the sense of  Definition~\ref{26.7.2022.01}.
%
%
The strategies and their terminal wealth converge pointwise to $\vp^\infty_t:=f(-m_{t\wedge \tau})1_{\{0<t\le \tau\}}$ and 
\beao
\vp^{0,\infty}_T=f(-m_\tau)\frac{1+B_\tau}2 -\int_0^{-m_\tau}(1-x)\,df(x) \ge -\int_0^1(1-x)\,df(x),\quad\mbox{respectively}. 
\eeao 
The limiting strategy~$(\vp^{0,\infty},\vp^\infty)$ is $M$-admissible as well ($\vp^\infty$ is not bounded anymore but obviously in $L(\un{X},\ov{X})$).

Now, we turn to admissibility in the sense of Definition~\ref{26.7.2022.01} but with $(\un{X},\ov{X})$ replaced by $(\un{S},\ov{S})$. 
For each $n\in\bbn$, we want to determine the minimal $M_n\in\bbr_+$ such that $(\vp^{0,n},\vp^n)$ is $M_n$-admissible.
The most vulnerable time for the strategy is when $B$ reaches $-1+1/n$ and $U>1-1/n$. If the stocks were sold at that time, the bond position would be $-\int_0^{1-1/n}(1-x)\,df(x) - f(1-1/n)$. 
The difference to above is one unit per share. Since the admissibility condition also allows debts in the stock position, we arrive at
$M_n=(\int_0^{1-1/n}(1-x)\,df(x) + f(1-1/n))/2\in\bbr_+$.
By the choice of $f$, we have that $M_n\to\infty$ as $n\to\infty$. Thus, the sequence is not admissible with regard to a joint $M'\in\bbr_+$,
%
%
and the limiting strategy~$(\vp^{0,\infty}\vp^\infty)$ is not admissible at all. 

This already gives a strong hint that one would not have Fatou-closedness of the set of terminal portfolios which can be achieved
by admissible strategies if $(\un{X},\ov{X})$ were replaced by  $(\un{S},\ov{S})$ in Definition~\ref{26.7.2022.01}.
The reason is that $M_n$ is too large compared to the worst-case risk at maturity.
However, it remains to show that the limiting wealth~$\vp^{0,\infty}_T$ cannot be achieved by an admissible strategy different 
from $(\vp^{0,\infty},\vp^\infty)$. Assume by contradiction that $(\psi^0,\psi)$ is admissible with $(\psi^0_T,\psi_T)=(\vp^{0,\infty}_T,0)$.
We leave it as an exercise for the reader to shows that for each $n\in\bbn$, $\psi$ has to coincide
with $\vp^n$ on $\{U\le 1-1/n\}$. By the minimality of $M_n$, it follows that $(\psi^0,\psi)$ cannot be $M'$-admissible for $M'<M_n$. Since $M_n\to\infty$, $(\psi^0,\psi)$ cannot be admissible at all.    
%
%
%
\end{example}

{\bf For the rest of the paper, we follow the standard convention to assume that} 
\beam\label{27.7.2022.1}
\mathcal{F}_T=\mathcal{F}_{T-},\quad \un{S}_T=\un{S}_{T-},\quad\mbox{and\ }\ov{S}_T=\ov{S}_{T-}.
\eeam
For the actual bid and ask prices this implies that $\un{X}_T=\un{X}_{T-}$ and $\ov{X}_T=\ov{X}_{T-}$ as well.
The assumption allows to identify $(\vp^0_T,\vp_T)$ with the terminal portfolio that cannot be 
rebalanced anymore. It is w.l.o.g. since it just avoids to introduce an additional time after $T$ when the investor can trade at prices $\un{S}_T$, $\ov{S}_T$. 

\begin{definition}\label{4.4.2023.1}
The market model satisfies the num\'eraire-free no-arbitrage condition~(${\rm NA}^{nf}$) iff there is no admissible strategy~$(\vp^0,\vp)$ with $\vp^0_0=\vp_0=0$, $P(\vp^0_T\ge 0,\vp_T\ge 0)=1$, and
$P(\{\vp^0_T>0\}\cup\{\vp_T>0\})>0$.
\end{definition}
The condition is called ``num\'eraire-free'' since bounded short positions in both the bond and the stock are allowed. 
Stating the admissibility condition in terms of $\un{X}$ and $\ov{X}$ is equivalent to freezing a portfolio position as
it is done in Guasoni, L{\'e}pinette, and R{\'a}sonyi~\cite{guasoni2012fundamental}. On the other hand, freezing a short position in a stock 
with a frictionless price modeled by a nonnegative strict local martingale that behaves like a doubling strategy leads to an arbitrage.
Thus, the freezing of positions better fits to the num\'eraire-free arbitrage theory that leads to true martingale price processes.\\

As discussed in the introduction, we want to merge the conditions NUPBR and ${\rm NA}^{ps}$ coming from continuous time frictionless models and discrete time transaction costs models, respectively.
For this purpose, we consider the {\em cost value process} introduced by Bayraktar and Yu~\cite{bayraktar.yu.2018} 
as the cost to enter a portfolio position and defined as
\beao
V^{\rm cost}(\vp) := \vp^0 + \vp^+ \ov{X} -\vp^-\un{X}\quad \mbox{for}\ (\vp^0,\vp)\ \mbox{self-financing with}\ \vp^0_0=\vp_0=0.
\eeao
\begin{definition}\label{14.6.2023.1}
The market model satisfies the {\em prospective strict no unbounded profit with bounded risk}~(${\rm NUPBR}^{ps}$) condition iff
\beam\label{12.6.2023.1}
\left\{\sup_{t\in[0,T]}V^{\rm cost}_t(\vp)\ :\ (\vp^0,\vp)\in\mathcal{A}_0^1\right\}\quad\mbox{is bounded in $L^0$,}
\eeam
and for every sequence $(\vp^{0,n},\vp^n)_{n\in\bbn}\subseteq\mathcal{A}_0^1$ such that $(\vp^n)_{n\in\bbn}\subseteq(\bP)^\Pi$ 
and $(\vp^{0,n}_T,\vp^n_T)\to (C^0,C)$ a.s., where 
$(C^0,C)$ is a maximal element (in the sense that the convergence of $1$-admissible strategies cannot hold for a random vector
that strictly dominates $(C^0,C)$ with respect to the pointwise order), there exist forward convex combinations $(\lambda_{n,k})_{n\in\bbn,\ k=0,\ldots,k_n}$, $k_n\in\bbn$, i.e., $\lambda_{n,k}\in\bbr_+$ and $\sum_{k=0}^{k_n} \lambda_{n,k}=1$, such that
\beam\label{12.6.2023.2}
\sup_{n\in\bbn}\sup_{t\in[0,T]}V^{\rm cost}_t(\sum_{k=0}^{k_n}\lambda_{n,k}\vp^{n+k})<\infty\quad\mbox{a.s.}
\eeam
\end{definition}
 \begin{remark}
 In the special case that $\un{X}=\ov{X}$ the condition ${\rm NUPBR}^{ps}$ coincides with NUPBR (the latter still considered for num\'eraire-free 
$1$-admissible strategies). Namely, when the running supremum of a frictionless wealth process reaches a pre-specified level, the value can be conserved by liquidating the portfolio (for condition (\ref{12.6.2023.2}) in frictionless markets we refer to the proof of the second assertion of Theorem~\ref{25.7.2022.1}). This transfer of the cost value to time~$T$ is not possible in models with friction, and the condition has to be stated
directly in terms of pathwise suprema. Condition~(\ref{12.6.2023.2}) is weaker than assuming $L^0$-boundedness of the convex hull in (\ref{12.6.2023.1}). The latter would be needed to obtain a finite limit of forward convex combinations from arbitrary sequences of $1$-admissible strategies by the $L^0$-version of Koml\'os theorem (see \cite[Lemma~A1.1]{delbaen1994general} and the counterexample in \cite[Remark~4 in the appendix]{delbaen1994general}).
Economically convex combinations of pathwise suprema would not be very meaningful and in general larger than the suprema of mixed strategies considered in (\ref{12.6.2023.2}). 
\end{remark}
\begin{remark}
From the subadditivity of $C^{S,S'}$ it follows that the set $(\bP)^\Pi$ is convex and $V^{\rm cost}(\sum_{k=0}^{k_n}\lambda_{n,k}\vp^{n+k})\ge
\sum_{k=0}^{k_n}\lambda_{n,k}V^{\rm cost}(\vp^{n+k})$ for all $(\vp^n)_{n\in\bbn}\subseteq(\bP)^\Pi$ and $(\lambda_{n,k})_{n\in\bbn,\ k=0,\ldots,k_n}\subseteq\bbr_+$ with $\sum_{k=0}^{k_n} \lambda_{n,k}=1$. The same holds for $V^{\rm liq}$ and $V^{S,S'}$.
On the other hand, in contrast to $L(S)$, the set~$L(\un{X},\ov{X})$ need not be convex. The reason is that the wealth of a mixed strategy would have to be $+\infty$ (which is excluded) if trading costs cancel by the mixing. We note that for the arbitrage theory this is no problem: in the frictionless market
with price process $S_t=t$ the ``arbitrage strategy''~$\vp_t=1/t$ does not lie in $L(S)$, but NUPBR is already ruled out by bounded strategies.\\ 

By convention~(\ref{27.7.2022.1}), we could reformulate the ${\rm NUPBR}^{ps}$ condition by considering only $1$-admissible strategies with $\vp^0_T,\vp_T\ge -1$. This means that at $T$ the investor must actually limit her debts, not just be able to do so. Thus, it makes sense to work with the pointwise order than comparing terminal positions in Definition~\ref{14.6.2023.1}.
\end{remark}
\begin{remark}[Discrete time]
In the case of only one risky asset, that we consider in this paper, the discrete time ${\rm NA}^{ps}$ condition from \cite[Definition~2.3]{kuehn.molitor.2019} reduces to
\beao
\forall\ \vp\ \mbox{predictable},\ t=0,1,\ldots,T\quad (  V^{\rm liq}_t(\vp)\ge 0\ \implies\ V^{\rm cost}_t(\vp)= V^{\rm liq}_t(\vp)=0 )
\eeao
(this follows by the arguments in the proof of Lemma~\ref{7.5.2022.1}). 
In discrete time, we have that ${\rm NUPBR}^{ps}\Leftrightarrow {\rm NA}^{ps}$. The implication ``$\Rightarrow$'' follows from the fact that
$V^{\rm liq}_t(\vp)\ge 0$ implies that $V^{\rm liq}_u(\vp)\ge 0$ for $u=0,1,\ldots,t-1$, where $(\vp^0,\vp)$ is a discrete time strategy not necessarily admissible ex ante (cf. again the proof of Lemma~\ref{7.5.2022.1}). Let us show
``$\Leftarrow$''. On the set~$\{\ov{X}_0 = {\rm ess inf}_{\mathcal{F}_0}\un{X}_1\}\in\mathcal{F}_0$ a purchase at time~$0$ is reversible in the sense of \cite[Definition~3.2]{kuehn.molitor.2019}. Thus, under ${\rm NA}^{ps}$ we have $\ov{X}_1=\un{X}_1=\ov{X}_0$ on $\{\ov{X}_0 = {\rm ess inf}_{\mathcal{F}_0}\un{X}_1\}$, and the purchase can be postponed to time~$1$. 
On the complement~$\{\ov{X}_0 > {\rm ess inf}_{\mathcal{F}_0}\un{X}_1\}$ the number of risky assets of a $1$-admissible strategy is bounded by 
$(1+{\rm ess inf}_{\mathcal{F}_0}\un{X}_1)/(\ov{X}_0 - {\rm ess inf}_{\mathcal{F}_0}\un{X}_1)$. For a short position in the risky asset one gets a similar random bound. Applying this argument inductively for $t=0,1,2,\ldots$ we obtain the following.
For every sequence of $1$-admissible strategies~$(\vp^{0,n},\vp^n)_{n\in\bbn}$ there is a corresponding sequence $(\wt{\vp}^{0,n},\wt{\vp}^n)_{n\in\bbn}$ with the same liquidation and cost value processes such that
$(\wt{\vp}^n_t)_{n\in\bbn}$ is $L^0$-bounded for every $t=1,\ldots,T$ (Namely, when following the strategies we only realize the 
purely nonreversible parts of the orders in the sense of \cite[Lemma~3.3]{kuehn.molitor.2019}. This means that, for example, at time $0$ purchases are only realized on $\{\ov{X}_0 > {\rm ess inf}_{\mathcal{F}_0}\un{X}_1\}$). The $L^0$-boundedness of the adjusted positions yields that ${\rm NUPBR}^{ps}$ is satisfied. 
\end{remark}
\begin{assumption}\label{11.5.2022.1}
Let $\tau$ be  a stopping time such that on $\{\tau<\infty\}$ there starts an excursion of the actual spread~$X:=\ov{X}-\un{X}$ away from zero (cf. (\ref{24.3.2023.1})). Then, there exist a stopping time $\sigma$ with $\sigma\ge \tau$ and
$\sigma>\tau$ on $\{\tau<\infty,\ X_\tau=0\}$ 
and probability measures~$Q^1$ and $Q^2$ equivalent to $P$ 
such that $\ov{X}^{\sigma}-\ov{X}^{\tau}$ is a $Q^1$-supermartingale 
and $\un{X}^{\sigma}-\un{X}^{\tau}$ is a $Q^2$-submartingale.  Let $L^1$ and $L^2$ be the stochastic logarithm of the corresponding density process $Z^{Q^1}$ and $Z^{Q^2}$, 
respectively, i.e., $Z^{Q^i} = 1 + Z^{Q^i}_-\mal L^i$. For each $A\in\mathcal{P}$, we define $Z^A$ by $Z^A=1 +  Z^A_-1_A\mal L^1 + Z^A_- 1_{(\Omega\times [0,T])\setminus A}\mal L^2$ and 
assume that $Z^A$, $A\in\mathcal{P}$, are true martingales defining probability measures~$Q^A$, for which we, in turn, assume that
they are uniformly equivalent to $P$, i.e.,
\beam\label{16.8.2022.2}
\forall\eps>0\ \exists \delta>0\ \forall A\in \mathcal{P}\ \forall B\in\mathcal{F}\quad (P(B)\le \delta\ \implies Q^A(B)\le \eps)
\eeam
\beam\label{16.8.2022.1}
\mbox{and}\quad\forall\eps>0\ \exists \delta>0\ \forall A\in \mathcal{P}\ \forall B\in\mathcal{F}\quad (Q^A(B)\le \delta\ \implies P(B)\le \eps).
\eeam
Let $\Gamma(\tau)$ be the end time of the excursion. Analogously to above, there exist a stopping time $\wt{\sigma}$ with $\wt{\sigma}\le \Gamma(\tau)$ and
$\wt{\sigma}<\Gamma(\tau)$ on $\{\Gamma(\tau)<\infty,\ X_{\Gamma(\tau)-}=0\}$ and probability measures~$\wt{Q}^1$ and $\wt{Q}^2$ equivalent to $P$ such that $\ov{X}^{\Gamma(\tau)}1_{\{X_{\Gamma(\tau)-}>0\}} + \ov{X}^{\Gamma(\tau)-}1_{\{X_{\Gamma(\tau)-}=0\}}
-\ov{X}^{\wt{\sigma}}$ is a 
$\wt{Q}^1$-submartingale and $\un{X}^{\Gamma(\tau)}1_{\{X_{\Gamma(\tau)-}>0\}} + \un{X}^{\Gamma(\tau)-}1_{\{X_{\Gamma(\tau)-}=0\}}-\un{X}^{\wt{\sigma}}$ is a $\wt{Q}^2$-supermartingale. The pair~$(\wt{Q}^1,\wt{Q}^2)$ satisfies the same pasting conditions as $(Q^1,Q^2)$ from above.
\end{assumption}

\begin{remark}
First of all, it should be noted that Assumptions~\ref{20.8.2022} and \ref{11.5.2022.1} are automatically satisfied if the efficient friction condition in the sense of $P(\inf_{t\in[0,T]}(\ov{X}_t-\un{X}_t)>0)=1$ holds. Thus, the assumptions are weaker than
those in the previous literature. In addition, they are automatically satisfied if the spread can only move away from zero or come back by jumps.

But unfortunately, the situation is extremely complicated when this happens continuously.
Already to construct self-financing portfolios, \cite{kuehn.molitor.2022} needs Assumption~\ref{20.8.2022} that rules out Brownian local time behavior, and under which the starting times of excursions of the spread away from zero are stopping times. Assumption~\ref{11.5.2022.1} is 
for beginning and end, and it goes beyond Assumption~\ref{20.8.2022}. It requires that for an arbitrarily short random duration at the beginning of an excursion, the market would be arbitrage-free even if purchases could be liquidated at the ask price and short positions could be closed at the bid price. This can be regarded as a local tightening of the ${\rm NUPBR}^{ps}$ condition at the starting time of an excursion:
The process $V^{\rm cost}$ values a purchased position at the higher ask price as long as it is in the portfolio. In the fictitious frictionless market described above, the position can even be liquidated at the ask price, i.e., the cost value can be realized. With short sells it is the other way round, i.e., the fictitious frictionless market consists of different price processes for long and short positions. Since long and short positions cannot be hold simultanously, there is a non-convex trading constraint, and separation theorems are not applicable.
The condition is mirrored at the end of an excursion. 
%
%
\end{remark}

\begin{remark}
Condition~(\ref{16.8.2022.2}) is equivalent to the condition that $\{Z^A_T : A\in\mathcal{P}\}$ is uniformly integrable.
By the Neyman-Pearson lemma, condition~(\ref{16.8.2022.1}) is equivalent to $\inf\{q_\eps(Z^A_T) : A\in\mathcal{P}\}>0$ for all $\eps>0$, where $q_\eps(Z^A_T)$ denotes the right $\eps$-quantile of the distribution of $Z^A_T$ under $P$. For price processes of the form $\ov{X}_t=\ov{S}_t=B_t+\ov{\mu}t$ and $\un{X}_t=\un{S}_t=B_t+\un{\mu}t$ with a standard Brownian motion $B$ and $\ov{\mu}\ge \un{\mu}$, Assumption~\ref{11.5.2022.1} can easily be verified by using Novikov's condition.
\end{remark} 

\begin{definition}
A consistent price system~(CPS) is a pair $(S,Q)$ such that $Q$ is
a probability measure equivalent to $P$ and $S$ is a $Q$-martingale with $\un{X}\le S\le \ov{X}$
(and thus a fortiori $\un{S}\le S\le \ov{S}$). 
\end{definition}

\begin{theorem}\label{25.7.2022.1}
If the market model satisfies Assumption~\ref{20.8.2022}, Assumption~\ref{11.5.2022.1},
${\rm NA}^{nf}$, and ${\rm NUPBR}^{ps}$, then there exists a CPS.
Conversely, if $(S,Q)$ is a CPS, then the bid-ask model with bid price~$S$ and ask price~$S$ satisfies 
Assumption~\ref{20.8.2022}, Assumption~\ref{11.5.2022.1}, ${\rm NA}^{nf}$, and ${\rm NUPBR}^{ps}$.
\end{theorem}
\begin{theorem}\label{7.5.2022.2}
Under Assumption~\ref{20.8.2022}, Assumption~\ref{11.5.2022.1}, ${\rm NA}^{nf}$, and ${\rm NUPBR}^{ps}$, the cone $\mathcal{C}_0:=\{(\vp^0_T,\vp_T) : (\vp^0,\vp)\in\mathcal{A}\}-L^0(\Omega,\mathcal{F},P;\bbr^2_+)$ is Fatou-closed in the sense that for every $M\in\bbr_+$, every sequence~$(C^{0,n},C^n)_{n\in\bbn}\subseteq \mathcal{C}_0$ with $C^{0,n},C^n\ge -M$  for all $n\in\bbn$, and every\linebreak $(C^0,C)\in L^0(\Omega,\mathcal{F},P;\bbr^2)$ with $(C^{0,n},C^n)\to (C^0,C)$\ a.s. as $n\to\infty$, there exists a $(\vp^0,\vp)\in\mathcal{A}$ with $(\vp^0_T,\vp_T)\ge (C^0,C)$ a.s.    
\end{theorem}

\section{Proof of Theorems~\ref{25.7.2022.1} and \ref{7.5.2022.2}}\label{6.6.2023.01}

The following lemma corresponds to \cite[Proposition~4.9]{guasoni2012fundamental}. Intuitively, it states that at any intermediate time, the
credit line (in terms of bonds and stocks) required for the trading strategy (that is followed so far) is minimized by freezing the portfolio and close the stock position at the best price that can be achieved for sure now or in the future. Put differently, consider an investor who has 
built up a, say, positive stock position at some time~$t$. Then, to minimize her worst-case risk she just has to sell the stocks at price~$\un{X}_t$, and complicated dynamic trading strategies cannot improve the result.

\begin{lemma}\label{7.5.2022.1}
Let $M\in\bbr_+$. Assume that the model satisfies ${\rm NA}^{nf}$.
Let $(\vp^0,\vp)$ be an admissible strategy with $P(\vp^0_T\ge -M, \vp_T\ge -M)=1$. Then, $(\vp^0,\vp)$ is  $M$-admissible.
\end{lemma}
\begin{proof}
It is sufficient to show the following seemingly weaker implication:  
for all $t_0\in(0,T)$ and for all admissible strategies~$(\vp^0,\vp)$ with $\vp^0_{t_0}=0$, $\vp_{t_0}=1$, and $\vp_T=0$, we have ${\rm ess inf}_{\mathcal{F}_{t_0}}\vp^0_T \le \un{X}_{t_0}$
(by symmetry, the arguments for an initial stock position below $-M$ are the same). 

Since we already passed to the actual prices, we prefer not to use the processes $\un{S}$ and $\ov{S}$ anymore (even though this may be more 
appealing at this place). Instead, we observe that $\un{X}_t={\rm ess inf}_{\mathcal{F}_t}\sup_{u\in[t,T]}\un{X}_u$ and  $\ov{X}_t={\rm ess sup}_{\mathcal{F}_t}\inf_{u\in[t,T]}\ov{X}_u$. W.l.o.g. let $\mathcal{F}_{t_0}$ be $P$-trivial. Assume by contradiction that there exists $\eps>0$ such that $P(\vp^0_T\ge \un{X}_{t_0} + \eps)=1$. 
The interesting case is that $\ov{X}_{t_0}\ge \un{X}_{t_0} + \eps$ since otherwise already $(\vp^0-\ov{X}_{t_0},\vp) 1_{\zu t_0,T\zu}$ would be an arbitrage, but we do not have to make a case differentiation. Consider the stopping time
\beam\label{20.4.2023.1}
\tau:=\inf\{t\ge t_0 : (\ov{X}_t\le \un{X}_{t_0} + 2\eps/3\ \mbox{or}\ \vp_t\le 0)\ \mbox{and}\ \un{X}_u\le \un{X}_{t_0} + \eps/3\quad \forall u\in(t_0,t)\}.
\eeam
We have that $P(\tau<\infty)>0$ since $P(\vp_T=0,\ \sup_{u\in[t_0,T]}\un{X}_u\le\un{X}_{t_0} + \eps/3)>0$ by the definition of the conditional essential infimum.
In addition, before time~$\tau(\omega)<\infty$, the ask price is above $\un{X}_{t_0}+2\eps/3$ and the bid price below $\un{X}_{t_0}+\eps/3$. At time~$\tau(\omega)<\infty$, either the investor can buy the stock at a better price than ever before or she short-sells the stock for the first time.
Under the contradiction assumption, we can switch from strategy~$0$ to strategy~$\vp$ at time $\tau$ and obtain an arbitrage.
To formalize this, let us show several inequalities, at first only in the case that $(\vp^0,\vp)$ is almost simple. Here, $S$ and $S'$ are arbitrary semimartingale price systems. The first inequality reads
\beam\label{15.4.2023.1}
\ov{X}_{t_0-} + V^{S,S'}_\tau(\vp 1_{\auf t_0,T\zu}) + \vp_\tau(\ov{X}_\tau - S_\tau) \le \un{X}_{t_0} + 2\eps/3\ \mbox{on\ }\{\tau<\infty, \ov{X}_\tau\le \un{X}_{t_0}+2\eps/3, \vp_\tau>0\}.
\eeam
The LHS are the costs to build up the position~$(\vp^0_\tau,\vp_\tau)$ at time~$\tau$ (note that $V^{S,S'}_0(\cdot)=0$ and $\ov{X}_{t_0-} + V^{S,S'}_{t_0}(\vp 1_{\auf t_0,T\zu})$ is the wealth of $(\vp^0,\vp)$ at time $t_0$). For almost simple strategies, inequality~(\ref{15.4.2023.1}) 
follows from $\ov{X}_t\ge \un{X}_{t_0} + 2\eps/3 \ge \ov{X}_\tau$ and $\un{X}_{t_0} + \eps/3 \ge \un{X}_t$ for all $t\in[t_0,\tau)$, which means that the initial stock position cannot be liquidated at a better price than $\un{X}_{t_0}+\eps/3$ and further purchases reduce the cost value. 
Next, we have again for $(\vp^0,\vp)$ almost simple
\beam\label{17.4.2023.1}
\ov{X}_{t_0-} + V^{S,S'}_\tau(\vp 1_{\auf t_0,T\zu}) + \vp_\tau(\un{X}_\tau - S_\tau) \le \un{X}_{t_0} + \eps/3\ \mbox{on\ }\{\tau<\infty, \ov{X}_\tau > \un{X}_{t_0}+2\eps/3, \vp_\tau>0\}.
\eeam
The estimate is for the case that $\tau$ is triggered by a nonpositive stock position but the infimum is not attained.
Since the long position is liquidated immediately afterwards at the lower bid price, we do not need to have control over $\ov{X}_\tau$, but $\un{X}_\tau \le \un{X}_{t_0}+\eps/3$ holds on $\{\tau<\infty\}$. 
Finally, one has for almost simple strategies
\beam\label{17.4.2023.2}
\ov{X}_{t_0-} + V^{S,S'}_{\tau-}(\vp 1_{\auf t_0,T\zu}) + \vp_{\tau-}(\un{X}_{\tau-} - S_{\tau-})
\le \un{X}_{t_0} + \eps/3\ \mbox{on}\ \{\tau<\infty,\ \vp_\tau\le 0\},
\eeam
where the left limit $\vp_{\tau-}$ also exists for $\vp\in L(\un{X},\ov{X})$ since $\ov{X}-\un{X} \ge \eps/3$ on $[t_0,\tau)$ if $\tau<\infty$, see \cite[Proposition~3.3]{kuehn.molitor.2022} and Proposition~\ref{10.12.2022}(a).

Let us show how (\ref{15.4.2023.1}), (\ref{17.4.2023.1}), and (\ref{17.4.2023.2}) can successively be extended to strategies from $(\bP)^\Pi$ and $L(\un{X},\ov{X})$. 
%
Let $\vp\in(\bP)^\Pi$ and $\tau$ from (\ref{20.4.2023.1}) belongs to this strategy. Since $\ov{X}_-
 - \un{X}_-$ is away from zero up to time $\tau$ on $\{\tau<\infty\}$, $\vp$ can be approximated by almost simple strategies uniformly in probability in the sense of \cite[proof of Proposition~3.17]{kuehn.molitor.2022}. The approximation can be adjusted such that the almost simple strategies take the value zero if $\vp$ or its right limit takes the value zero. We observe that $\tau$ need not coincide with the stopping times in (\ref{20.4.2023.1}) for the almost simple strategies. 
But, (\ref{15.4.2023.1}), (\ref{17.4.2023.1}), and (\ref{17.4.2023.2}) still hold with the different stopping time since we only need that the almost simple strategy does not become negative before the stopping time. Thus, the inequalities carry over to $\vp$ by convergence of $V^{S,S'}$ uniformly in probability and pointwise convergence of the strategies and their left limits at $\tau$ (see 
\cite[Theorem~3.19(ii) and Proposition~3.17]{kuehn.molitor.2022}).
Now, let $\vp\in L(\un{X},\ov{X})$. From the proof of Lemma~\ref{28.12.2022.1} it follows that the approximating 
sequence~$(\vp^n)_{n\in\bbn}\subseteq(\bP)^\Pi$ in Definition~\ref{def:DefinitionL} can be chosen such that $(\vp^n)^+\wedge 1 = \vp^+\wedge 1$ for all $n\in\bbn$. This means that the associated $\tau$ is the same for $\vp^n$ and $\vp$ and the inequalities carry over to $\vp$ by (\ref{27.9.2022.1}) and the fact that $\vp^n_{\tau-}\to \vp_{\tau-}$ on $\{\un{X}_{\tau-} < S_{\tau-}\}$  by  
Proposition~\ref{10.12.2022}(b) after passing to a deterministic subsequence.

Putting inequalities~(\ref{15.4.2023.1}), (\ref{17.4.2023.1}), and (\ref{17.4.2023.2}) together implies that the self-financing strategy~$(\psi^0,\psi)$ with $\psi:=\vp 1_{\auf\tau_{\{\vp_\tau\le 0\}}\zu\cup\zu\tau,T\zu}$
is admissible and satisfies $\psi_T=\vp_T=0$, $\psi^0_T=0$ on $\{\tau=\infty\}$, and
$\psi^0_T\ge \vp^0_T -\un{X}_{t_0} - 2\eps/3$ on $\{\tau<\infty\}$.   
Under the contradiction assumption one has $\vp^0_T -\un{X}_{t_0} - 2\eps/3 \ge \eps/3$. As observed above, we have $P(\tau<\infty)>0$. Thus,
$(\psi^0,\psi)$ must be an arbitrage, which is a contradiction.
\end{proof}
The following proposition describes jumps of general wealth processes at points with positive spread. Here, strategies have to be of finite variation.
\begin{proposition}\label{10.12.2022}
(a) Let $\vp\in L(\un{X},\ov{X})$. On $\{\ov{X}>\un{X}\}$ the paths of $\vp$ are of finite variation in right-hand neighborhoods; consequently, 
the right-hand limit of $\vp$ exists, $\Delta^+\vp$ is finite, and $\Delta^+ V^{S,S'}(\vp) =
(S-\ov{X})(\Delta^+\vp^+)^+ +(\un{X}-S)(\Delta^+\vp^+)^- +(\un{X}-S')(\Delta^+\vp^-)^+ +(S'-\ov{X})(\Delta^+\vp^-)^-$ up to evanescence. Analogously, on $\{\ov{X}_->\un{X}_-\}$ the paths of $\vp$ are of finite variation in left-hand neighborhoods, the left-hand limit of $\vp$ exists, $\Delta^-\vp$ is finite, and
$\Delta^- V^{S,S'}(\vp) = \vp^+\Delta S - \vp^-\Delta S' + (S_- - \ov{X}_-)(\Delta^-\vp^+)^+ +(\un{X}_- - S_-)(\Delta^-\vp^+)^- 
+(\un{X}_- - S'_-)(\Delta^-\vp^-)^+ +(S'_- - \ov{X}_-)(\Delta^-\vp^-)^-$ up to evanescence.\\
(b) Let $(\vp^n)_{n\in\bbn}\subseteq(\bP)^\Pi$ be an ``optimal'' approximating sequence of $\vp$ in the sense of Definition~\ref{def:DefinitionL}. Then, one has that
$(\Delta^+(\vp^{n_k})^+)^+\to (\Delta^+\vp^+)^+$ on $\{\ov{X}>S\}\cup\{\ov{X}>\un{X},\ (\Delta^+\vp^+)^+=0\}$
and $(\Delta^-(\vp^{n_k})^+)^+\to (\Delta^-\vp^+)^+$ on $\{\ov{X}_->S_-\}\cup\{\ov{X}_->\un{X}_-,\ (\Delta^-\vp^+)^+=0\}$ up to evanescence as $k\to\infty$ for some (deterministic) subsequence~$(n_k)_{k\in\bbn}$. The analogous statements for all combinations of negative/positive parts of the strategies and their jumps hold as well (cf. (a)).
\end{proposition}
\begin{proof}
Ad (a). Let $(\omega,t_1)\in\Omega\times[0,T)$ with $\ov{X}_{t_1}(\omega)>\un{X}_{t_1}(\omega)$ and, omitting $\omega$ in the notation, 
$t_2:=\inf\{t>t_1 : \ov{X}_t \le  
2/3\ov{X}_{t_1} + 1/3\un{X}_{t_1}\ \mbox{or}\ \un{X}_t \ge 1/3\ov{X}_{t_1} + 2/3\un{X}_{t_1}\}\wedge T$. 
By the right-continuity of $\ov{X}$ and $\un{X}$, we have that $t_2>t_1$.
For an almost simple strategy~$\vp$, elementary calculations show the estimate
\beam\label{4.10.2022.1}
V^{S,S'}_{t_2}(\vp) - V^{S,S'}_{t_1}(\vp) \le -(\ov{X}_{t_1} - \un{X}_{t_1})/6 {\rm  Var}_{t_1}^{t_2}(\vp) + 
(|\vp_{t_1}|+|\vp_{t_2}|)(\sup_{t\in[t_1,t_2]}\ov{X}_t-\inf_{t\in[t_1,t_2]}\un{X}_t).
\eeam 
(Having the ``last in first out principle'' in mind, we argue as follows: A stock position which is both built up and liquidated between $t_1$ and $t_2$ cause a loss higher than $(\ov{X}_{t_2} - \un{X}_{t_1})/3$. Gains from shares which are trades at most once can be estimated very roughly
since their trading volume is bounded by $|\vp_{t_1}|+|\vp_{t_2}|$). 

It remains to extend (\ref{4.10.2022.1}) successively to $(\bP)^\Pi$ and to $L(\un{X},\ov{X})$. 
%
%
This follows by convergence up to evanescence along subsequences (cf. \cite[Theorem~3.19(ii)]{kuehn.molitor.2022} and Definition~\ref{def:DefinitionL})
and by the fact that the variation of the pointwise limiting strategy is dominated by the 
$\liminf$ of the variation processes.
%
%
The arguments for the left-hand neighborhood are the same.  
For a strategy path of finite variation the equations in (a) easily follow from the definition of the cost term 
(cf. \cite[Definition~3.2 and the proof of (A.3)]{kuehn.molitor.2022}).\\

Ad (b). Convergence in probability implies almost sure convergence along a subsequence. Thus, by a diagonalization argument, we can find a
(deterministic) subsequence $(n_k)_{k\in\bbn}$ (not depending on $i$) such that for each $i\in\bbn$, the convergence in (\ref{27.9.2022.1}) holds uniformly in time for almost all $\omega$. By part~(a), we know that on the set~$\{\ov{X}>\un{X}\}$ the right-hand limits of the strategies 
exist and $(S-\ov{X})(\Delta^+(\vp^{n_k})^+)^+ +(\un{X}-S)(\Delta^+(\vp^{n_k})^+)^- +(\un{X}_-S')(\Delta^+(\vp^{n_k})^-)^+ +(S'-\ov{X})(\Delta^+(\vp^{n_k})^-)^-
\to (S-\ov{X})(\Delta^+\vp^+)^+ +(\un{X}-S)(\Delta^+\vp^+)^- +(\un{X}-S')(\Delta^+\vp^-)^+ +(S'-\ov{X})(\Delta^+\vp^-)^-$ up to evanescence
by the above mentioned uniform convergence of the wealth processes on a single excursion.
There remains the problem that $\Delta^+\vp^{n_k}$ can have the opposite sign of $\Delta^+\vp$ although the jumps of the cost terms are similar. 

We fix an $\omega\in \Omega$ outside the $P$-null sets from above and omit it in the notation. Let $t_1\in[0,T)$. To save space, we only write down 
the case that $\vp_{t_1}\ge 0$, $\ov{X}_{t_1}>S_{t_1}$, and $\Delta^+\vp_{t_1}>0$. The other cases follow analogously.
Now, assume by contradiction that, in addition, there exist an $\eps>0$ and a (random) subsequence~$(k_l)_{l\in\bbn}$ with
\beam\label{10.12.2022.1}
|(\Delta^+ (\vp^{n_{k_l}}_{t_1})^+)^+ - (\Delta^+\vp^+_{t_1})^+|\ge \eps\quad \mbox{for all\ }l\in\bbn.
\eeam
Define 
\beao
t_2 & := & \inf\{t>t_1 : \vp_t \le \vp_{t_1+} -\Delta^+\vp_{t_1}/2\quad \mbox{or}\quad V^{S,S'}_t(\vp) \le V^{S,S'}_{t_1+}(\vp) - \Delta^+\vp_{t_1}(\ov{X}_{t_1} - S_{t_1})/7\\ 
& & \qquad \mbox{or}\quad \ov{X}_t \le 2/3 \ov{X}_{t_1} +1/3 S_{t_1} \quad \mbox{or}\quad S_t \ge 1/3 \ov{X}_{t_1} + 2/3 S_{t_1}\}\wedge T.
\eeao
By $\vp^{n_k}_{t_1}\to\vp_{t_1}$, $\Delta^+ V^{S,S'}_{t_1}(\vp^{n_k}) \to \Delta^+ V^{S,S'}_{t_1}(\vp)$, and $\Delta^+V^{S,S'}_{t_1}(\vp)<0$,
we must have that $\Delta^+\vp^{n_{k_l}}_{t_1}<0$ for all $l$ large enough (depending on $\omega$). Indeed, if $\Delta^+\vp^{n_{k_l}}_{t_1}$ 
had the same sign as $\Delta^+\vp_{t_1}$, divergent absolute values would be contrary to convergent jumps of the cost term.
By $\vp^{n_k}_{(t_1+t_2)/2}\to\vp_{(t_1+t_2)/2}$ and the construction of $t_2$, the loss of wealth that
$(\vp^{n_{k_l}})_{l\in\bbn}$ necessarily produces on $(t_1,(t_1+t_2)/2]$ to get closer to $\vp$ again can be estimated from below by $1/6 (\ov{X}_{t_1}-S_{t_1})\Delta^+\vp_{t_1}$ for $l\to\infty$. Since this dominates the possible losses of $\vp$ {\em after} $t_1+$, we arrive at a contradiction to $V^{S,S'}_{t_1}(\vp^{n_k})\to V^{S,S'}_{t_1}(\vp)$ and $V^{S,S'}_{(t_1+t_2)/2}(\vp^{n_k})\to V^{S,S'}_{(t_1+t_2)/2}(\vp)$\ as $k\to\infty$.
This means that if $(\vp^{n_{k_l}})_{l\in\bbn}$ traded in the opposite direction of $\vp$, it would have to compensate for this promptly, which would lead to {\em additional} transaction costs. The assertion for the left-hand jump follows analogously.
\end{proof}

By definition of $\vp\in L(\un{X},\ov{X})$, its wealth process can be approximated by wealth 
processes~$V^{S,S'}(\vp^n)$ with bounded strategies~$\vp^n$, $n\in\bbn$, satisfying $(\vp^n)^+ \le \vp^+$, $(\vp^n)^- \le \vp^-$, and 
$\vp^n \to \vp$. 
However, $\vp^n$ need not be $M$-admissible if $\vp$ is $M$-admissible. 
Lemma~\ref{12.6.2022.1} shows that we can {\em choose} the approximating bounded strategies such that 
they are $M$-admissible. The following two lemmas prepare Lemma~\ref{12.6.2022.1}. Since the investor has a credit line of $M$ stocks,
the ``safest short-term strategy'' is to hold $-M$ stocks. This means that for $M>0$, the inequalities $(\vp^n)^+ \le \vp^+$ and $(\vp^n)^- \le \vp^-$ do not imply that $\vp^n$ is ``safer than $\vp$ in the short term''. The following lemma overcomes this problem. 
\begin{lemma}\label{28.12.2022.1}
Let $\vp\in L(\un{X},\ov{X})$ and $M\in\bbr_+$. Then, the approximating sequence~$(\varphi^n)_{n\in\bbn}\subseteq (\bP)^\Pi$ in Definition~\ref{def:DefinitionL} can be chosen such that 
\beam\label{16.6.2022.2}
(\vp^n)^-\wedge M=\vp^-\wedge M
\eeam 
and thus $(\vp^n+M)^+\le (\vp+M)^+$ for all $n\in\bbn$.
\end{lemma}
We note that, by contrast, the inequality $(\vp^n+M)^-\le (\vp+M)^-$ already follows from $(\vp^n)^-\le \vp$.
\begin{proof}[Proof of Lemma~\ref{28.12.2022.1}]
{\em Step 1:} Let $\vp\in L(\un{X},\ov{X})$ and $M\in\bbr_+$. By definition, there exists a sequence~$(\vp^n)_{n\in\bbn}\subseteq (\bP)^\Pi$ with $(\vp^n)^+\le \vp^+$, $(\vp^n)^-\le \vp^-$, $\vp^n\to \vp$ and semimartingale price systems~$S$, $S'$ 
such that (\ref{3.6.2023.2}), (\ref{27.9.2022.1}), (\ref{27.9.2022.2}) hold 
and (\ref{23.8.2020.1}) is satisfied  for all competing sequences.
For any $\psi\in(\bP)^\Pi$, we have that 
\beam\label{27.12.2022.1}
C^{S'}(-\psi^-) = C^{S'}(-(\psi^-\wedge M)) + C^{S'}(-(\psi + M)^-)
\eeam
since the
strategies $-(\psi^-\wedge M)$ and $-(\psi + M)^-$ never trade in the opposite direction (formally, one checks it for almost simple
strategies and apply the approximation in \cite[proof of Theorem~3.19]{kuehn.molitor.2022}).
Consequently, we can decompose: $V^{S,S'}(\vp^n) = V^{S,S'}((\vp^n)^+) + V^{S,S'}(-((\vp^n)^-\wedge M)) + V^{S,S'}(-(\vp^n+M)^-)$ and consider the alternative pointwise approximation 
$\wt{\vp}^n:= \vp^n1_{\{\vp\ge 0\}} - (\vp^-\wedge M) - (\vp^n+M)^- = \vp^n\wedge ((-M)\vee \vp)$,\ $n\in\bbn$.
Again by (\ref{27.12.2022.1}), we have that $V^{S,S'}(\wt{\vp}^n) = V^{S,S'}((\vp^n)^+) + V^{S,S'}(-(\vp^-\wedge M)) + V^{S,S'}(-(\vp^n+M)^-)$ and thus
$V^{S,S'}(\wt{\vp}^n) - V^{S,S'}(\vp^n)=V^{S,S'}(-(\vp^-\wedge M)) - V^{S,S'}(-((\vp^n)^-\wedge M))$. 
By \cite[Corollary~3.24]{kuehn.molitor.2022} using a Fatou-type estimate, there exists a (deterministic) subsequence~$(n_k)_{k\in\bbn}$ such that 
\beam\label{25.6.2023.1}
(V^{S,S'}(-(\vp^-\wedge M)) - V^{S,S'}(-((\vp^{n_k})^-\wedge M)))^-\to 0\quad\mbox{up to evanescence as}\ k\to\infty.
\eeam 
Now, we turn to a single interval $\mathcal{I}^c_i$ (the same for $\mathcal{I}^{fc}_i$)
and define $Y^n:=1_{\mathcal{I}^c_i}\mal (V^{S,S'}(-(\vp^-\wedge M)) - V^{S,S'}(-((\vp^n)^-\wedge M)))$.
(\ref{25.6.2023.1}) already implies that, with $(\vp^n)_{n\in\bbn}$, also $(\wt{\vp}^n)_{n\in\bbn}$ is better than all competing sequences in the sense of (\ref{23.8.2020.1}). By (\ref{23.8.2020.1}), there exists a further subsequence~$(k_l)_{l\in\bbn}$ such that
$(1_{\mathcal{I}^c_i}\mal(V^{S,S'}(\wt{\vp}^{n_{k_l}}) - V^{S,S'}(\vp^{n_{k_l}})))^+ \to 0$, and together with (\ref{25.6.2023.1}) we arrive at 
\beam\label{16.6.2022.1}
1_{\mathcal{I}^c_i}\mal Y^{n_{k_l}}\to 0\quad\mbox{up to evanescence as}\ l\to\infty.
\eeam
We continue by observing that
\beao
& & Y^n\\ 
& & = 1_{\mathcal{I}^c_i}\mal ((\vp^n)^-\wedge M  - (\vp^-\wedge M))\mal S' - 1_{\mathcal{I}^c_i}\mal C^{S'}(-(\vp^-\wedge M)) 
  + 1_{\mathcal{I}^c_i}\mal C^{S'}(-((\vp^n)^-\wedge M))
\eeao
and define $\wt{Y}^n:= 1_{\mathcal{I}^c_i}\mal C^{S'}(-((\vp^n)^-\wedge M)) - 1_{\mathcal{I}^c_i}\mal C^{S'}(-(\vp^-\wedge M))$.
There exists a further (deterministic) subsequence~$(l_j)_{j\in\bbn}$ such that
\beam\label{28.12.2022.2}
\sup_{t\in[0,T]}|((\vp^{n_{k_{l_j}}})^-\wedge M - (\vp^-\wedge M))\mal S'_t|\to 0,\quad \mbox{a.s.}
\eeam
(cf. \cite[Theorem~I.4.31(iii)]{jacod.shiryaev}) and Proposition~\ref{10.12.2022}(b) holds true for $(n_{k_{l_j}})_{j\in\bbn}$. Putting (\ref{16.6.2022.1}) and (\ref{28.12.2022.2}) together,
we have
\beam\label{27.9.2022.3}
\wt{Y}^{n_{k_{l_j}}}\to 0\ \mbox{up to evanescence as}\ j\to\infty.
\eeam
{\em Step 2:} It remains to show that $(\wt{\vp}^{n_{k_{l_j}}})_{j\in\bbn}$ satisfies (\ref{27.9.2022.1}) and (\ref{27.9.2022.2}), where the limiting wealth process is of course the same as for the approximating strategies~$(\vp^n)_{n\in\bbn}$ we started with. 
By (\ref{28.12.2022.2}), it is sufficient to show that $\sup_{t\in[0,T]}|\wt{Y}^{n_{k_{l_j}}}_t|\to 0$ a.s. as $j\to\infty$. W.l.o.g. $n_{k_{l_j}}=j$.
We assume by contradiction that there exists $A\in\mathcal{F}_T$ with $P(A)>0$ such that on $A$ the sequence~$(\sup_{t\in[0,T]}|\wt{Y}^n_t|)_{n\in\bbn}$ does not tend to $0$, but the convergence in (\ref{27.9.2022.3}) holds. 
We fix an $\omega\in A$ that is omitted in the following notation. There exists $\eps>0$ such that $\sup_{t\in[0,T]}|\wt{Y}^n_t|>\eps$ for infinitely many $n$ (depending on $\omega$).  
Define $\sigma_n:=\inf\{t\ge 0 : |\wt{Y}^n_t|>\eps\}$. The sequence~$(\sigma_n)_{n\in\bbn}$ has an accumulation point in $[0,T]$
denoted by $\sigma$. 
By (\ref{27.9.2022.3}), we have that 
\beam\label{28.9.2022.1}
\wt{Y}^n_\sigma\to 0\quad\mbox{as\ }n\to\infty.
\eeam 
At first, we consider the behaviour of $\wt{Y}^n$ in a right-hand neighborhood of $\sigma$. By $\vp^n_\sigma\to\vp_\sigma$ and Proposition~\ref{10.12.2022}(b), considering the cases $\Delta^+ \vp^-_\sigma<0$, $\Delta^+ \vp^-_\sigma>0$, and 
$\Delta^+ \vp^-_\sigma=0$, we obtain that $(S'_\sigma-\ov{X}_\sigma)((\vp^n_{\sigma+})^-\wedge M -((\vp^n_\sigma)^-\wedge M))^- +(\un{X}_\sigma-S'_\sigma)((\vp^n_{\sigma+})^-\wedge M -((\vp^n_\sigma)^-\wedge M))^+$ converges to
$(S'_\sigma-\ov{X}_\sigma)(\vp^-_{\sigma+}\wedge M - (\vp^-_\sigma\wedge M))^- +(\un{X}_\sigma-S'_\sigma)(\vp^-_{\sigma+}\wedge M -(\vp^-_\sigma\wedge M))^+$. This means that $\Delta^+V^{S,S'}_\sigma(-((\vp^n)^-\wedge M))\to\Delta^+V^{S,S'}_\sigma(-(\vp^-\wedge M))$. 
By the right-continuity of stochastic integrals and by (\ref{28.9.2022.1}), we have that
$C^{S'}_{\sigma+}(-((\vp^n)^-\wedge M))\to C^{S'}_{\sigma+}(-(\vp^-\wedge M))$. In addition,  
there exists $\sigma'>\sigma$ such that the limit cost term $C^{S'}_{\sigma'}(-(\vp^-\wedge M))$ is bounded from above by
$C^{S'}_{\sigma+}(-(\vp^-\wedge M))+\eps/2$. By $C^{S'}_{\sigma'}(-((\vp^n)^-\wedge M))\to C^{S'}_{\sigma'}(-(\vp^-\wedge M))$ as $n\to\infty$,
(\ref{28.9.2022.1}), and by the monotonicity of the cost terms, we obtain
$\sup_{n\ge n_0,\ t\in[\sigma,\sigma']}|C^{S'}_t(-((\vp^n)^-\wedge M))-C^{S'}_t(-(\vp^-\wedge M))|\le \eps$ for $n_0$ large enough. 
The left-hand neighborhood of $\sigma$ can be handled in the same way. 
We arrive at a contradiction to $\sigma_n\to\sigma$. This shows that $\sup_{t\in[0,T]}|\wt{Y}^n_t|\to 0$ a.s. and thus (\ref{27.9.2022.1}).
Since the differences of the cost terms and the semimartingale gains converge separately in $d_{up}$, we can argue as in Proposition~\ref{15.1.2023.4} to derive (\ref{27.9.2022.2}) with the same process ${}^p\!S'$ that does the job for $(\vp^n)_{n\in\bbn}$.
\end{proof}
%
%

\begin{lemma}\label{29.1.2023.2}
For $\psi\in L(\un{X},\ov{X})$ we define the predictable process
\beao
L(\psi):=\Pi(\psi) + M + (\psi + M)^+ {\rm ess inf}_{\mathcal{F}_-}\un{X} - (\psi + M)^- {\rm ess sup}_{\mathcal{F}_-}\ov{X}
\eeao
(that models the liquidation value of $(\Pi(\psi)+M,\psi+M)$ after the portfolio is rebalanced but before the prices jump at some time~$t$).
Then, 
\beao
(\Pi(\psi),\psi)\mbox{\ is $M$-admissible}\quad\implies\quad L(\psi)\ge 0
\eeao
\end{lemma}
\begin{proof}
We suppose otherwise. Then, by a section theorem for predictable sets (see, e.g., \cite[Theorem~4.8]{he.wang.yan.1992}), there would exist a predictable stopping time~$\tau$ with
$P(\tau<\infty)>0$ and $L_\tau(\psi)<0$ on $\{\tau<\infty\}$. This implies $P(\Pi_\tau(\psi) + M + (\psi_\tau + M)^+ \un{X}_\tau - 
(\psi_\tau + M)^- \ov{X}_\tau<0)>0$, a contradiction to the $M$-admissibility of $(\Pi(\psi),\psi)$. 
\end{proof}
\begin{lemma}\label{12.6.2022.1}
Let $\vp\in L(\un{X},\ov{X})$ such that $(\Pi(\vp),\vp)$ is $M$-admissible for some $M\in\bbr_+$. Then, the approximating sequence~$(\varphi^n)_{n\in\bbn}\subseteq (\bP)^\Pi$ in Definition~\ref{def:DefinitionL} can be chosen such that $(\Pi(\vp^n),\vp^n)$ are also $M$-admissible for all $n\in\bbn$.
%
%
%
\end{lemma}
\begin{proof}[Proof of Lemma~\ref{12.6.2022.1}]
Let $\vp\in L(\un{X},\ov{X})$ such that $(\Pi(\vp),\vp)$ is $M$-admissible. By Lemma~\ref{28.12.2022.1}, we can and do choose an
approximating sequence~$(\varphi^n)_{n\in\bbn}\subseteq (\bP)^\Pi$ such that (\ref{16.6.2022.2}) holds.
The approximation holds with regard to the semimartingale price systems $S$ and $S'$ and the associated predictable processes by ${}^p\!S$ and ${}^p\!S'$. 
We write $\vp^{0,n} := \Pi(\vp^n)$, $\vp^0:=\Pi(\vp)$, $V^{S,S'}(\vp):=V$ and introduce the predictable processes 
${}^pV^n:= \vp^{0,n} + (\vp^n)^+\, {}^p\!S -(\vp^n)^-\, {}^p\!S'$,
${}^pV:= \vp^0 + \vp^+\, {}^p\!S -\vp^-\, {}^p\!S'$, $L^n:=L(\vp^n)$, and
$L:=L(\vp)$. 

{\em Step 1:} In the first step, we prove the lemma for the special case that $\vp$ invests only during finitely many intervals~$\mathcal{I}^c_i$ 
and $\mathcal{I}^{fc}_i$.  
Let $\delta>0$. In this special case, (\ref{27.9.2022.1}) and (\ref{27.9.2022.2}) imply that 
$P(\sup_{t\in[0,T]}(|V^{S,S'}_t(\vp^n) - V^{S,S'}_t(\vp)| +  |{}^pV^n_t-{}^pV_t|) > \delta)\le \delta$ for $n$ large enough. 
We define $\tau:=\inf\{t\ge 0 : V^{S,S'}_t(\vp^n) < V^{S,S'}_t(\vp) -\delta\ \mbox{or}\ {}^pV^n_t < {}^pV_t -\delta\}$ and 
$A:=\{{}^pV^n_\tau \ge {}^pV_\tau - \delta\}\in\mathcal{F}_{\tau-}$. 
Since $\tau$ is the debut of a progressive set, it is a stopping time (see e.g. \cite[Theorem~7.3.4]{cohen.elliott.2015}).
Consequently, $\auf 0,\tau\auf\cup \auf \tau_A\zu = \auf 0,\tau\zu \cap \{{}^pV^n\ge {}^pV - \delta\}\in\mathcal{P}$, which
allows us to consider the strategy
\beao
\wt{\vp}^n:=\vp^n1_{\auf 0,\tau\auf\cup \auf \tau_A\zu}.
\eeao
At first, we compare liquidation values strictly before $\tau$. By (\ref{16.6.2022.2}), it is easy to check that 
\beam\label{9.2.2023.1}
& & \vp^{0,n} + M + (\vp^n + M)^+ \un{X} - (\vp^n + M)^- \ov{X} - \left(\vp^0 + M + (\vp + M)^+ \un{X} - (\vp + M)^- \ov{X}\right)\nonumber\\
& & = V^{S,S'}(\vp^n) - V^{S,S'}(\vp) + (\vp^+ - (\vp^n)^+)(S-\un{X}) +  (\vp^- - (\vp^n)^-)(\ov{X}-S')\nonumber\\
& & \ge V^{S,S'}(\vp^n) - V^{S,S'}(\vp) \ge -\delta\qquad\quad\mbox{on}\ \auf 0, \tau\auf
\eeam
and thus, by the $M$-admissibility of $\vp$, 
\beam\label{10.2.2023.1}
\vp^{0,n} + M + (\vp^n + M)^+ \un{X} - (\vp^n + M)^- \ov{X} \ge -\delta\qquad\quad\mbox{on}\ \auf 0,\tau\auf.
\eeam
We proceed by analyzing the liquidation value of the portfolio $(\Pi(\wt{\vp}^n)+M,\wt{\vp}^n+M)$ at $\tau$ if the event $\Omega\setminus A$ occurs.
This means that long and short positions in the stock are liquidated at the prices $\un{X}_{\tau-}$ and $\ov{X}_{\tau-}$, respectively.
The paths of the processes $\vp^{0,n}$ and $\vp^n$ must be of finite variation in a left-hand neighborhood 
of $\tau$ if $\ov{X}_{\tau-}>\un{X}_{\tau-}$ (see Proposition~\ref{10.12.2022}(a)). Consequently, on $(\Omega\setminus A)\cap\{\ov{X}_{\tau-}>\un{X}_{\tau-}\}$ one receives $(\vp^n_{\tau-} + M)^+ \un{X}_{\tau-} - (\vp^n_{\tau-} + M)^- \ov{X}_{\tau-} \ge -(\vp^{0,n}_{\tau-}+M+\delta)$ by (\ref{10.2.2023.1}).
On $(\Omega\setminus A)\cap\{\ov{X}_{\tau-}=\un{X}_{\tau-}\}$ we use that $V^{S,S'}_t(\vp^n+M)+M\ge -\delta$ for all $t<\tau$ by (\ref{10.2.2023.1}),  
$V^{S,S'}_t(\vp^n+M)\to V^{S,S'}_{\tau-}(\vp^n+M)$ as $t\uparrow \tau$, and 
$V^{S,S'}_{\tau-}(\vp^n+M)+M=V^{S,S'}_{\tau-}(\vp^n+M)+M+\Delta^-C_\tau(\vp^n+M)=\vp^{0,n}_\tau+M+(\vp^n_\tau + M)^+ \un{X}_{\tau-} - (\vp^n_\tau + M)^- \ov{X}_{\tau-}$. This yields that 
\beam\label{11.2.2023.1}
\wt{\vp}^{0,n}_\tau + M = \vp^{0,n}_\tau+M+(\vp^n_\tau + M)^+ \un{X}_{\tau-} - (\vp^n_\tau + M)^- \ov{X}_{\tau-} 
\ge -\delta\quad\mbox{on}\ \Omega\setminus A.
\eeam
Finally, we analyze the liquidation on the event~$A$, i.e., we have to show that 
$\vp^{0,n} + M + (\vp^n + M)^+ \un{X} - (\vp^n + M)^- \ov{X}\ge 0$ on the set $\auf\tau_A\zu$.
Analogously to (\ref{9.2.2023.1}), again using (\ref{16.6.2022.2}), we estimate   
\beam\label{18.9.2022.1}
L^n - L & = & {}^pV^n - {}^pV + (\vp^+ - (\vp^n)^+)({}^p\!S -{\rm ess inf}_{\mathcal{F}_-}\un{X}) +  (\vp^- - (\vp^n)^-)({\rm ess sup}_{\mathcal{F}_-}\ov{X}-{}^p\!S')\nonumber\\
& \ge & {}^pV^n - {}^pV\quad \mbox{on}\ \auf\tau_A\zu.
\eeam
By Lemma~\ref{29.1.2023.2}, we have that $L\ge 0$ up to evanescence. 
Together with (\ref{18.9.2022.1}), we obtain that 
\beam\label{10.2.2023.2}
L^n \ge L -\delta \ge -\delta\quad\mbox{on}\ \auf \tau_A\zu.
\eeam
In order to replace the predictable lower bound~$L^n$ of the liquidation process by the process itself, we pointwise distinguish the two cases below. Note that $\tau$ is in general not predictable and thus, $({\rm ess inf}_{\mathcal{F}_-}\un{X})_\tau$ need not coincide with 
${\rm ess inf}_{\mathcal{F}_{\tau-}}\un{X}_\tau$.

{\em Case~1:}\ $\un{X}_\tau \ge ({\rm ess inf}_{\mathcal{F}_-}\un{X})_\tau$ if $\vp_\tau + M \ge 0$ 
(and $\ov{X}_\tau\le ({\rm ess sup}_{\mathcal{F}_-}\ov{X})_\tau$ if $\vp_\tau + M < 0$). 
This means that the liquidation value becomes higher. Inequation~(\ref{10.2.2023.2}) yields
\beam\label{7.1.2022.1}
& & \vp^{0,n} + M + (\vp^n + M)^+ \un{X} - (\vp^n + M)^- \ov{X}\\ 
& & \ge L^n \ge -\delta\quad\mbox{on}\ \auf \tau_A\zu\cap(\{\vp+M\ge 0,\ \un{X}\ge {\rm ess inf}_{\mathcal{F}_-}\un{X}\}\cup\{\vp+M<0,\ \ov{X}
\le {\rm ess sup}_{\mathcal{F}_-}\ov{X}\})\nonumber.
\eeam

{\em Case~2:}\ $\un{X}_\tau<({\rm ess inf}_{\mathcal{F}_-}\un{X})_\tau$ if $\vp_\tau + M \ge 0$ 
(and $\ov{X}_\tau>({\rm ess sup}_{\mathcal{F}_-}\ov{X})_\tau$ if $\vp_\tau + M < 0$). 
This case can only occur if $\tau$ is an unpredictable stopping time with $\un{X}_\tau<\un{X}_{\tau-}$ (or $\ov{X}_\tau>\ov{X}_{\tau-}$).
The liquidation value becomes smaller than its predictable lower bound, but the decrease is smaller than that of the limiting strategy:
\beam\label{7.1.2022.2}
& & \vp^{0,n} + M + (\vp^n + M)^+ \un{X} - (\vp^n + M)^- \ov{X}\\
& & = \vp^0 + M + (\vp + M)^+ \un{X} - (\vp + M)^- \ov{X} + L^n - L\nonumber\\
& & \quad + ((\vp^n)^+ -\vp^+)(\un{X}-{\rm ess inf}_{\mathcal{F}_-}\un{X}) + (\vp^- -(\vp^n)^-)(\ov{X}-{\rm ess sup}_{\mathcal{F}_-}\ov{X})\nonumber\\
& & \ge \vp^0 + M + (\vp + M)^+ \un{X} - (\vp + M)^- \ov{X} + L^n - L\nonumber\\
& & \ge -\delta\quad\mbox{on}\ \auf \tau_A\zu\cap(\{\vp+M\ge 0,\ \un{X}<{\rm ess inf}_{\mathcal{F}_-}\un{X}\}\cup\{\vp+M<0,\ \ov{X} > {\rm ess sup}_{\mathcal{F}_-}\ov{X}\}),
\nonumber
\eeam
where the equality holds by (\ref{16.6.2022.2}) and the second inequality holds because $(\vp^0,\vp)$ is $M$-admissible. 

Putting together, the strategy~$(\Pi(\wt{\vp}^n),\wt{\vp}^n)$ is $(M+\delta)$-admissible.
By compression and since $\delta>0$ is arbitrary, we find an approximating sequence (with a suitable null sequence~$(\delta_n)_{n\in\bbn}$)
that is also $M$-admissible.\\ 

{\em Step 2:} We now proceed to the general case. Again, let $\delta>0$. By (\ref{3.6.2023.2}) one can find a finite union $J$ of intervals $(I^c_i)_{i\in\bbn}$, $(I^{fc}_i)_{i\in\bbn}$ such that $P(\sup_{t\in[0,T]} |1_J\mal V_t - V_t |1_{\{X_t=0\}}\vee |1_J\mal V_{t-} - V_{t-}|1_{\{X_{t-}=0\}} >\delta)\le \delta$.
Given $\delta>0$, we modify the approximating strategies such that they satisfy $\vp^n:=-M$ on 
$(\Omega\times[0,T])\setminus J$. By construction of $I^c_i$, $I^{cf}_i$, there are no trading costs at the boundaries of the intervals. In addition,
by the semimartingale property of $S'$, $J$ can be chosen close enough to $\Omega\times[0,T]$ such that 
$d_{\bbs}(-M 1_{(\Omega\times[0,T])\setminus J}\mal S',0)\le \delta$.
With these considerations in mind, the proof is analog to Step~1. We note that on $(\Omega\times[0,T])\setminus J$ the strategies $\vp^n$ 
need not be liquidated since they just hold $-M$ stocks.
\end{proof}

\begin{lemma}\label{28.2.2023.3}
Let $\vp^n,\vp\in L(\un{X},\ov{X})$ for all $n\in\bbn$ and let $S,S'$ be semimartingale price systems. If $\vp^n$ is $M$-admissible for all $n\in\bbn$, $V^{S,S'}(\vp^n)\to V^{S,S'}(\vp)$ uniformly in probability, and $\vp^n\to\vp$ up to evanescence on $\{\ov{X}>\un{X}\}$, then $\vp$ is $M$-admissible as well. 
\end{lemma}
\begin{proof}
%
For any $\psi\in L(\un{X},\ov{X})$ we define the process~$A(\psi):=\Pi(\psi) + M + (\psi + M)^+ \un{X} -(\psi+M)^-\ov{X} - V^{S,S'}(\psi)$ and rewrite it to
%
%
\beao
A(\psi) & = & M + 1_{\{\psi\ge 0\}}(\psi(\un{X}-S)+M\un{X}) 
+ 1_{\{-M<\psi<0\}}(\psi(\un{X}-S')+M\un{X})\nonumber\\
& & + 1_{\{\psi\le -M\}}(\psi(\ov{X}-S')+M\ov{X}).  
\eeao
Consequently, $\vp^n\to\vp$ implies that $A(\vp^n)\to A(\vp)$ pointwise.
Since, $\Pi(\vp^n) + M + (\vp^n + M)^+ \un{X} -(\vp^n+M)^-\ov{X}\ge 0$ for all $n\in\bbn$ and $V^{S,S'}(\vp^n)\to V^{S,S'}(\vp)$ uniformly in probability, the assertion follows.
\end{proof}
 
For the rest of the section, we fix an $M\in\bbr_+$ and a sequence 
\beam\label{25.4.2023.1}
(\vp^{0,n},\vp^n)_{n\in\bbn}\subseteq \mathcal{A}\ \mbox{such that}\ \vp^{0,n}_T,\vp^n_T\ge -M\ \mbox{for all}\  n\in\bbn\ \mbox{and}\ 
(\vp^{0,n}_T,\vp^n_T)\to (C^0,C)\ \mbox{a.s.},
\eeam
where $(C^0,C)\in L^0(\Omega,\mathcal{F},P;\bbr^2)$ is a {\em maximal element} in the sense that (\ref{25.4.2023.1}) cannot hold for a random vector that strictly dominates $(C^0,C)$ in the pointwise order. E.g., $(2,-1)$ does not dominate $(0,0)$ even if $\ov{S}_T=1$. By (\ref{12.6.2023.1}) maximal elements exist following the arguments in \cite[Lemma 4.3]{delbaen1994general}. 

To proof Theorem~\ref{7.5.2022.2},
we have to show that there exists an $M$-admissible strategy with terminal value $(C^0,C)$, but there is still a long way to go. 
We refer to \cite[Remark 4.4]{delbaen1994general} for an in-depth discussion for the need of maximal elements and also for the argument why it is sufficient to consider sequences as in (\ref{25.4.2023.1}) to prove Theorem~\ref{7.5.2022.2}. We note that for these considerations it does not make any difference that we consider a two-dimensional framework here. Under convention~(\ref{27.7.2022.1}), we can interpret $(\vp^{0,n}_T,\vp^n_T)$ as the position after the market has closed. Thus, positions in different ``currencies'' are analog to wealth in different scenarios.\\

By Lemma~\ref{7.5.2022.1}, the sequence in (\ref{25.4.2023.1}) has to be $M$-admissible. In addition, for each $n\in\bbn$, there is a sequence of bounded processes~$(\vp^{n,m})_{m\in\bbn}$ that approximate $\vp^n$ in the sense 
of Definition~\ref{def:DefinitionL}. By Lemma~\ref{12.6.2022.1}, the self-financing strategies~$(\vp^{0,n,m},\vp^{n,m})$ with $\vp^{0,n,m}:=\Pi(\vp^{n,m})$ can be chosen to be $M$-admissible as well. Conditions~(\ref{3.6.2023.2}) and (\ref{27.9.2022.1}) in Definition~\ref{def:DefinitionL}  imply that $V^{S,S'}_T(\vp^{n,m})\to V_T$ in probability as $m\to\infty$, which means that $\vp^{0,n,m}_T + (\vp^{n,m}_T)^+ S_T - (\vp^{n,m}_T)^- S'_T \to \vp^{0,n}_T + (\vp^n_T)^+ S_T - (\vp^n_T)^- S'_T$ in probability. By $\vp^{n,m}_T\to \vp^n_T$ pointwise as $m\to\infty$, we get $\vp^{0,n,m}_T\to \vp^{0,n}_T$ in probability as $m\to\infty$ for all $n\in\bbn$. 
Since the convergence in probability is metrizable, there exists a subsequence~$(m_n)_{n\in\bbn}$ with $(\vp^{0,n,m_n}_T,\vp^{n,m_n}_T)\to (C^0,C)$ in probability as $n\to\infty$ (one can choose $m_n$ such that $P(|\vp^{0,n,m_n}_T-\vp^{0,n}_T|>1/n\ \mbox{or}\ |\vp^{n,m_n}_T-\vp^n_T|>1/n)\le 1/n$). By passing to a subsequence, we can also get almost sure convergence.   
This allows us to assume w.l.o.g. that there exists a sequence~$(a_n)_{n\in\bbn}\subseteq\bbr_+$ such that already the sequence in (\ref{25.4.2023.1}) satisfies
\beam\label{31.12.2022.1}
|\ph^n|\le a_n\quad\mbox{and}\quad (\vp^{0,n},\vp^n)_{n\in\bbn}\subseteq\mathcal{A}^M,\quad n\in\bbn.
\eeam
Of course, $a_n$ may explode as $n\to\infty$ but for the analysis at the boundaries of the intervals with friction we want to argue with bounded strategies. 
\begin{note}\label{5.5.2023.1}
Let $\wt{S}$ be a semimartingale price system.
Then, there exist semimartingale price systems~$S$ and $S'$ such that for every starting time of an excursion~$\tau$, there exist stopping times $\sigma$ and 
$\wt{\sigma}$ satisfying the conditions from Assumption~\ref{11.5.2022.1} with
\beam\label{13.8.2022.3}
(S,S')1_{\auf\tau,\Gamma(\tau)\zu}=(\ov{X},\un{X})1_{\auf \tau,\sigma\auf} + (\wt{S},\wt{S}) 1_{\auf\sigma,\wt{\sigma}\auf} + (\un{X},\ov{X})1_{\auf\wt{\sigma},\Gamma(\tau)\zu},
\eeam
\end{note}
\begin{proof}
There are at most countably many excursions. Since the processes $\ov{X},\un{X}$, and $\wt{S}$ are c\`adl\`ag, one can choose $\sigma$ and $\wt{\sigma}$ close enough to $\tau$ and $\Gamma(\tau)$, respectively, such that the correction terms that occurs by replacing the semimartingale by $Q^A$-super- and submartingales
at the boundaries are arbitrarily small in terms of the metric~$d_{\bbs}$.
\end{proof} 
For the rest of the section, we fix the semimartingale price systems~$S$ and $S'$ to evaluate long and short positions in the risky asset, respectively, and construct a limiting strategy 
of the sequence in (\ref{25.4.2023.1})/(\ref{31.12.2022.1}) using these semimartingales in Definition~\ref{def:DefinitionL}(b). This implies that if the spread cannot move away from zero continuously or return to zero continuously, then one can take any semimartingale price system to construct the limiting strategy.   

\subsection{Fictitious dormant market}

In the proof of the frictionless FTAP the concatenation property of the set of wealth processes plays a crucial role.
This property is not available in models with transaction costs since on $\{X_t>0\}$ one cannot switch between two strategies without additional costs. On the other hand, frictionless markets are included in our model. 
To prove Theorem~\ref{7.5.2022.2}, we proceed in two steps. In the {\em first step}, we consider a fictitious market that behaves similar to a frictionless market and satisfies the concatenation property. When the bid-ask spread is positive, the market is dormant, but at other times one can switch between the strategies. More precisely, the gains that are obtained during an excursion of the spread away from zero enter in the wealth processes but one cannot switch from one strategy to another strategy during this time. 
 Then, in the {\em second step}, the limiting strategy is constructed separately for each excursion, and properties (\ref{27.9.2022.1}), (\ref{27.9.2022.2}), and (\ref{23.8.2020.1}) are verified accordingly. We stress that the second step is the main new step. However, the first step that is developed in this subsection is not only needed to capture the frictionless intervals but also to paste the limiting strategies of the countably many excursions together and obtain a suitable limiting strategy along the whole interval.

Intuitively, the dormant model coincides  with the original 
model if $X_-=0$, and it is dormant during excursions of the spread away from zero. When an excursion ends and a new frictionless interval begins, the model is restarted with the current frictionless wealth (including gains from trades during the excursions). 
At the starting time of an excursion two different cases can occur: Either the spread is still zero. Then, we investor can still switch to another strategy using the information that an excursion has just started. Or, the spread jumps away from zero. Then, the investor is already in the midst of an excursion, and in the fictitious model the gains at the jump time and the gains during the rest of the excursion cannot be separated from each other. The easiest way to model this is to allow {\em in the first case} for double jumps of the wealth process at the starting time of an excursion: the 
left jump models movements triggered by a possible synchronous jump of the semimartingale price systems and the right jump models the anticipated gains during the excursion.
By contrast, if the spread jumps away from zero, the fictitious wealth already takes the value after the excursion and remains constant up to the end  of the excursion. Formally, we introduce the generalized time change~$(\tau_t,D_t)_{t\in[0,T]}$ by
\begin{align*}
\tau_t:=\begin{cases}
t & \mathrm{if}\ X_t=0\\
\Gamma(\tau_1^i)\wedge T &  \mathrm{if}
\ X_t>0\ \mathrm{and\ for}\ i\in\bbn\ \mbox{with}\ t\in[\tau_1^i,\Gamma(\tau_1^i))
\end{cases}
\end{align*}
and $D_t:=\{X_t>0\}\cap(\cup_{i\in\bbn} \{\tau_1^i\le t<\Gamma(\tau_1^i)\le T,\ X_{\Gamma(\tau_1^i)-}=0\})$. We observe that $(\tau_t)_{t\in[0,T]}$ is a nondecreasing family of stopping times (not necessarily right-continuous). It is shown in \cite[proof of Lemma~5.2]{kuehn.molitor.2022} that $\Gamma(\tau_1^i)_{\{X_{\Gamma(\tau_1^i)-}=0\}}$ is a predictable stopping time. Together with $\{X_t>0\}=\{\tau_t>t\}\in\mathcal{F}_{\tau_t-}$ for $t\in[0,T)$, this implies that $D_t\in\mathcal{F}_{\tau_t-}$ for all $t\in[0,T]$, using that $D_T=\emptyset$.
We introduce the not necessarily right-continuous filtration $\wt{\bbf}=(\wt{\mathcal{F}}_t)_{t\in[0,T]}$ by
\beao
\wt{\mathcal{F}}_t  & := & \sigma(\mathcal{F}_{\tau_t-} \cup ((\Omega\setminus D_t)\cap\mathcal{F}_{\tau_t})),\quad t\in[0,T].
\eeao
For any $\psi\in (\bP)^\Pi$ with wealth process~$V^{S,S'}(\psi) := \psi^+\mal S -\psi^-\mal S' - C^S(\psi^+) - C^{S'}(-\psi^-)$, the {\em dormant wealth process}~$\wt{V}^\psi$ is defined as
\beam\label{11.8.2022.1}
\wt{V}^\psi_t:= V^{S,S'}_{\tau_t-}(\psi) 1_{D_t} + V^{S,S'}_{\tau_t}(\psi) 1_{\Omega\setminus D_t},\quad t\in[0,T].
\eeam
For the bounded strategies above, we write $\wt{V}^n:=\wt{V}^{\vp^n}$, $n\in\bbn$. We note that the definition of $\wt{V}^\psi_t$ on $\{\tau_t=T\}$ can depend on the choice of $(S,S')$, but these semimartingale price systems are already fixed.
The paths of the process~$\wt{V}^\psi$ are l\`agl\`ad
but can have double jumps. The integration theory for l\`agl\`ad integrators is comparatively little developed (see the monograph by Abdelghani and Melnikov~\cite{abdelghani.melnikov.2022} for a survey).
An integration theory that is tailor-made for trading models with l\`agl\`ad price processes (not trading strategies!) is introduced in 
K\"uhn and Stroh~\cite{kuehn.stroh.2009}. The key idea is that one can separately invest 
in the part of the right jumps of the asset price process that can be overlapped by countably many stopping times, i.e., that is ``accessible''.
\begin{remark}\label{29.6.2023.1}
It is quite natural to model the dormant market with double jumps since this allows to keep the time domain~$[0,T]$. 
However, a reader who prefers to follow the arguments of this subsection within the standard framework of c\`adl\`ag integrators
may introduce at time $(\tau_1^i)_{\{X_{\tau_1^i}=0\}}$ an additional time shift of $2^{-i}$: the movement in the original model is paused briefly, and the right jump turns into a left jump taking place at a later point in time. Formally, this corresponds to the generalized time change~$((\tau\circ\tau')_{t+},D'_t)_{t\in[0,T]}$ with $\tau'_t := \inf\{u\ge 0 : u + \sum_{i=1}^\infty 2^{-i} 1_{\{\tau_1^i<u,\ X_{\tau_1^i}=0\}}\ge t\}$ and $D'_t:=\{X_{\tau'_t}>0\}\cap\{\tau'_t<T\}\cap\{X_{((\tau\circ\tau')_{t+})-}=0\}$ for $t\in[0,T+\sum_{i=1}^\infty 2^{-i} 1_{\{\tau_1^i<T,\ X_{\tau_1^i}=0\}}]$.
\end{remark}

\begin{lemma}[cf. Korollar~10.12 in Jacod~\cite{jacod.1979}]\label{3.6.2024.01}
For every $\psi\in (\bP)^\Pi$, the process~$\wt{V}^\psi$ is an optional semimartingale (see, e.g., 
\cite[Definition 2.5]{kuehn.stroh.2009}) under the filtration~$\wt{\bbf}$.
\end{lemma}
\begin{proof}
We have to start with some preparatory work on stopping times regarding the two filtrations. Let $\sigma$ be an $\wt{\bbf}$-stopping time.

{\em Step 1:} Let us show that $\tau_{\sigma}$ is an $(\mathcal{F}_t)_{t\in[0,T]}$-stopping time. By construction of $(\tau_t)_{t\in[0,T]}$, one has
\beam\label{12.7.2024.1}
\{\tau_{\sigma}\le t\} = \left( \{X_t=0\}\cap \{\sigma\le t\}\right)\cup\cup_{i\in\bbn}\left(\{\tau_1^i\le t\}\cap \{\sigma<\tau_1^i\}\right),\quad t\in[0,T].
\eeam
In addition, $\{q<\tau_1^i\}\in \mathcal{F}_{\tau_1^i}$ and $\{q<\tau_1^i\}\cap \mathcal{F}_{\tau_q}\subseteq \mathcal{F}_{\tau_1^i}$ imply that
\beam\label{12.7.2024.2}
\{\sigma<\tau_1^i\} = \cup_{q\in\bbq\cap[0,T]} \left(\{\sigma<q\}\cap\{q<\tau_1^i\}\right)\in\mathcal{F}_{\tau_1^i}.
\eeam
By (\ref{12.7.2024.1}), $\{X_t=0\}\cap\wt{\mathcal{F}}_t\subseteq \mathcal{F}_t$, and (\ref{12.7.2024.2}), one obtains that $\{\tau_\sigma\le t\}\in\mathcal{F}_t$.

{\em Step 2:} Let us show that $(\tau_{\sigma})_{D_\sigma}$ is an $(\mathcal{F}_t)_{t\in[0,T]}$-predictable stopping time, where $D_\sigma := \{\omega\in\Omega : \omega\in D_{\sigma(\omega)}\}$. Since $\Gamma(\tau_1^i)_{\{X_{\Gamma(\tau_1^i)-}=0\}}$ is an $(\mathcal{F}_t)_{t\in[0,T]}$-  predictable stopping time (see \cite[proof of Lemma~5.2]{kuehn.molitor.2022}), we have that $\{\tau_1^i\le\sigma<\Gamma(\tau_1^i)\wedge T\}\cap\{X_{\Gamma(\tau_1^i)-}=0\}\in\mathcal{F}_{\Gamma(\tau_1^i)-}$ and 
$\auf \Gamma(\tau_1^i)_{\{\tau_1^i\le\sigma<\Gamma(\tau_1^i)\wedge T\}\cap\{X_{\Gamma(\tau_1^i)-}=0\}},T\zu\in\mathcal{P}$
(cf., e.g., \cite[Proposition~I.2.10]{jacod.shiryaev}). By $\auf (\tau_{\sigma})_{D_\sigma}, T\zu = \cup_{i\in\bbn} \auf \Gamma(\tau_1^i)_{\{\tau_1^i\le\sigma<\Gamma(\tau_1^i)\wedge T\}\cap\{X_{\Gamma(\tau_1^i)-}=0\}}, T\zu$ we are done.

{\em Step 3:} Now, we come to the main part of the proof. Since the cost terms~$C^S$, $C^{S'}$ are nondecreasing adapted processes, $V^{S,S'}(\psi)$ is an optional semimartingale under $(\mathcal{F}_t)_{t\in[0,T]}$ which even possesses a local martingale part~$M$ with c\`adl\`ag paths. We can adapt the arguments of Jacod~\cite[Theorem~10.10]{jacod.1979}. It is sufficient to analyze the time-changed process~$\wt{M}_t:=M_{\tau_t-}1_{D_t} + M_{\tau_t}1_{\Omega\setminus D_t}$, $t\in[0,T]$, whereas the time-changed finite variation process is obviously of finite variation. 

Let $(T_k)_{k\in\bbn}$ be a localizing sequence of $(\mathcal{F}_t)_{t\in[0,T]}$-stopping times such that $M^{T_k}$, $k\in\bbn$, are martingales under $(\mathcal{F}_t)_{t\in[0,T]}$. 
Let us show that the time-changed stopped process~$\wt{M}^k_t:=M^{T_k}_{\tau_t-}1_{D_t} + M^{T_k}_{\tau_t}1_{\Omega\setminus D_t}$, $t\in[0,T]$, 
is an optional martingale under $\wt{\bbf}$. For this, let $\sigma$ be an arbitrary $\wt{\bbf}$-stopping time. 
By Steps~1 and 2, we have that $\auf (\tau_\sigma)_{D_\sigma}\zu\cup\zu \tau_\sigma, T\zu\in\mathcal{P}$ and thus
$E(\wt{M}^k_T-\wt{M}^k_\sigma)=E(1_{\auf (\tau_\sigma)_{D_\sigma}\zu\cup\zu \tau_\sigma, T\zu}\mal M^{T_k}_T)=0$, using that $\tau_T=T$ and $D_T=\emptyset$.

Now, consider $S_k := \tau_{T_k}$, $k\in \bbn$, which is a localizing sequence of $\wt{\bbf}$-stopping times by $\mathcal{F}_t\subseteq \wt{\mathcal{F}}_t$ and Step~1. 
Since $\wt{M}=\wt{M}^k$ on $\auf 0, \sup\{s\ge 0 : \tau_s<T_k\}\auf$ and $\wt{M}$ is constant on $\zu \sup\{s\ge 0 : \tau_s<T_k\},\tau_{T_k}\auf$, the process $\wt{M}$ is an optional semimartingale under $\wt{\bbf}$, and the assertion of the lemma follows (we note that one has no control over the left and right jump at $\sup\{s\ge 0 : \tau_s<T_k\}$, which is in general not even an $\wt{\bbf}$-stopping time, thus one cannot conclude that $\wt{M}$ is an optional local martingale). 
\end{proof}
\begin{lemma}\label{24.4.2023.1}[cf. Lemma~4.5 in Delbaen and Schachermayer~\cite{delbaen1994general}]
Let ${\rm NA}^{nf}$ and ${\rm NUPBR}^{ps}$ be satisfied. Then,
\beao
\sup_{t\in[0,T]}|\wt{V}^n_t-\wt{V}^m_t|\to 0\quad\mbox{in probability as\ }n,m\to\infty.
\eeao
\end{lemma}
This means that at the times the bid-ask spread vanishes the wealth processes have to be close together. 
\begin{proof}
The proof is analogous to that in the frictionless model by \cite[Lemma~4.5]{delbaen1994general}, but a bit more technical since we have to deal with processes that have double jumps.
By (\ref{25.4.2023.1}), we already know that $(\wt{V}^n_T)_{n\in\bbn}$ is a Cauchy sequence in probability.
Now, assume by contradiction that there exists $\eps>0$ and subsequences $(n_k,m_k)_{k\in\bbn}$ with $P(\sup_{t\in[0,T), \tau_t<T}(\wt{V}^{n_k}_t-\wt{V}^{m_k}_t)> \eps)\ge\eps$ for all $k\in\bbn$ (note convention~(\ref{27.7.2022.1})).
Let $\wt{T}_k:=\inf\{t\ge 0 : \tau_t<T,\ \wt{V}^{n_k}_t-\wt{V}^{m_k}_t\ge \eps\}$. Define $T_k:=\wt{T}_k$ on $\{\wt{T}_k<\infty\}\cap\{X_{\wt{T}_k}=0\}\cap\{\wt{V}^{n_k}_{\wt{T}_k}-\wt{V}^{m_k}_{\wt{T}_k}\ge \eps\}$, $T_k:=\Gamma(\wt{T}_k)$ elsewhere on $\{\wt{T}_k<\infty\}$, and $T_k:=\infty$ on $\{\wt{T}_k=\infty\}$. Consider 
$\psi^k:=\vp^{n_k}1_{\auf 0,T_k\zu\setminus\auf (T_k)_A\zu} + \vp^{m_k}1_{\auf (T_k)_A\zu\cup\zu T_k,T\zu}$ with
$A:=\{\wt{T}_k<\infty\}\cap(\{X_{\wt{T}_k}>0\}\cup\{\wt{V}^{n_k}_{\wt{T}_k}-\wt{V}^{m_k}_{\wt{T}_k}<\eps\})\cap\{X_{\Gamma(\wt{T}_k)-}=0\}$. The process~$\psi^k$ is predictable 
in the original model and thus generates a strategy.
If $\wt{T}_k$ coincides with some $\tau_1^i$, there are two cases. In the case that $X_{\wt{T}_k}=0$ and $\wt{V}^{n_k}_{\wt{T}_k}-\wt{V}^{m_k}_{\wt{T}_k}\ge \eps$ one 
switches from $\vp^{n_k}$ to $\vp^{m_k}$ before the excursion starts (using the information $\mathcal{F}_{\tau_1^i}$). 
Otherwise, one still follows the strategy~$\vp^{n_k}$ during the excursion and switches to $\vp^{m_k}$ 
at its end. In the fictitious model the information about the excursion is already available at $\tau_1^i$ and the model is 
dormant afterwards.
We arrive at $\psi^k_T = \vp^{m_k}_T$, $\Pi_T(\psi^k) \ge \Pi_T(\vp^{m_k}) + \eps$ on $\{\wt{T}_k < \infty\}$ and 
$\psi^k_T = \vp^{n_k}_T$, $\Pi_T(\psi^k) = \Pi_T(\vp^{n_k})$ on $\{\wt{T}_k=\infty\}$ with $P(\wt{T}_k<\infty)\ge\eps$. Now, we pass to forward convex combination of $(\psi^k)_{k\in\bbn}$ (cf. \cite[Lemma~A1.1]{delbaen1994general}). Since $(\Pi(\psi^k),\psi^k)\in\mathcal{A}^M$ for all $k\in\bbn$ and the self-financing operator $\Pi$ is concave on $(\bP)^\Pi$, we arrive at a contradiction to the maximality of $(C^0,C)$ in the above sense.  
\end{proof}
\begin{lemma}\label{14.7.2024.1}
Let $Y$ be an $\wt{\bbf}$-optional semimartingale whose right jumps only take place at the stopping 
times~$(\tau_1^i)_{\{X_{\tau_1^i}=0\}}$, $i\in\bbn$. 
Then, the sequence~$\left(\sum_{i=1}^k \Delta^+ Y_{\tau_1^i} 1_{\{X_{\tau_1^i}=0\}} 1_{\zu \tau_1^i,T\zu}\right)_{k\in\bbn}$ is $d_{up}$-Cauchy, 
and the limiting process~$Y^g:=\sum_{i=1}^\infty \Delta^+ Y_{\tau_1^i} 1_{\{X_{\tau_1^i}=0\}} 1_{\zu \tau_1^i,T\zu}$ is a left-continuous $\wt{\bbf}$-optional semimartingale. The process~$Y^r:=Y-Y^g$ is a c\`adl\`ag $\wt{\bbf}$-semimartingale.
\end{lemma}
\begin{proof}
By the continuity of the integral in \cite[Theorem~3.5]{kuehn.stroh.2009}, applied to the optional semimartingale~$Y$ and the sequence of optional integrands~$H^{2,k}:=\sum_{i=1}^k 1_{\auf (\tau_1^i)_{\{X_{\tau_1^i}=0\}}\zu}$, $k\in\bbn$, the above sequence is $d_{up}$-Cauchy, and its  limit~$Y^g$ is the integral of the pointwise limit of $(H^{2,k})_{k\in\bbn}$ with respect to the integrator~$Y$. Then, the integral~$Y^g$ is an optional semimartingale by Galtchouk~\cite[Theorem~2.3]{galtchouk.1985} and the arguments at the end of Step~2 in the proof of \cite[Theorem~3.5]{kuehn.stroh.2009} (the former is only for the case that the integrator is a locally square integrable martingale). By construction, the process~$Y^r$ has no right jumps.
\end{proof}
Next, we proceed towards convergence with respect to the semimartingale topology. 
Since the processes are not c\`adl\`ag but can have right jumps, the \'Emery metric has to be adjusted. To avoid the introduction of too much notation, we write down $\wt{d}_{\bbs}$ only for $\wt{\bbf}$-optional semimartingales~$Y^1$, $Y^2$ whose right jumps only take place at the stopping 
times~$(\tau_1^i)_{\{X_{\tau_1^i}=0\}}$, $i\in\bbn$, using the decomposition in Lemma~\ref{14.7.2024.1}:
\beam\label{19.6.2023.02}
\wt{d}_{\bbs}(Y^1,Y^2) & := & \sup_{H\ \mbox{\small $\wt{\bbf}$-predictable with}\ |H|\le 1,\ G_i\in L^0(\mathcal{F}_{\tau_1^i}),\ |G_i|\le 1} E(\sup_{t\in[0,T]}| H\mal ((Y^1)^r-(Y^2)^r)_t\nonumber\\
& & \qquad\qquad +\sum_{i=1}^\infty G_i 1_{\{\tau_1^i<t,\ X_{\tau_1^i}=0\}}(\Delta^+ Y^1_{\tau_1^i}-\Delta^+ Y^2_{\tau_1^i})|\wedge 1).
\eeam
The integral in (\ref{19.6.2023.02}) can be seen as a standard stochastic integral under $\wt{\bbf}$ or $\wt{\bbf}^+$ (the filtrations generate the same predictable sets). The sum in (\ref{19.6.2023.02}) converges uniformly in probability by the same arguments as in the proof of Lemma~\ref{14.7.2024.1}.
\begin{remark}\label{20.7.2024.1}
With the integral in \cite[Theorem~3.5]{kuehn.stroh.2009}, the definition of $\wt{d}_{\bbs}$ for arbitrary optional semimartingales is canonical. 
For the purposes of the present paper, however, it is sufficient to consider optional semimartingales whose right jumps only take place at the stopping 
times~$(\tau_1^i)_{\{X_{\tau_1^i}=0\}}$, $i\in\bbn$. By Lemma~\ref{3.6.2024.01}, the dormant wealth processes~$\wt{V}^\psi$ satisfy this property. This special situation is much simpler since right jumps can then be treated as left jumps that take place at separate times (cf. Remark~\ref{29.6.2023.1}).
\end{remark}
\begin{lemma}\label{2.6.2023.2}[cf. Lemmas~4.7 to 4.11 in Delbaen and Schachermayer~\cite{delbaen1994general}] 
Let ${\rm NA}^{nf}$ and ${\rm NUPBR}^{ps}$ be satisfied. Then, there exist forward convex combinations
$(\lambda_{n,k})_{n\in\bbn,\ k=0,\ldots,k_n}$, $k_n\in\bbn$, i.e., $\lambda_{n,k}\in\bbr_+$ and $\sum_{k=0}^{k_n} \lambda_{n,k}=1$, such that
for $\psi^n := \sum_{k=0}^{k_n}\lambda_{n,k}\vp^{n+k}$, $(\wt{V}^{\psi^n})_{n\in\bbn}$ is a $\wt{d}_{\bbs}$-Cauchy sequence.
In addition, there exists an $\wt{\bbf}$-optional semimartingale~$\wt{V}$ whose right jumps only take place at the stopping times~$(\tau_1^i)_{\{X_{\tau_1^i}=0\}}$, $i\in\bbn$, such that $\wt{d}_{\bbs}(\wt{V}^{\psi^n},\wt{V})\to 0$ as $n\to\infty$.
\end{lemma}
As soon as the lemma is proven, we pass to these forward convex combinations and use that $\wt{d}_{\bbs}(\wt{V}^n,\wt{V})\to 0$ as $n\to\infty$.
By concavity of $\Pi$, the properties of $(\vp^{0,n},\vp^n)_{n\in\bbn}$ remain valid.
\begin{proof}[Proof of Lemma~\ref{2.6.2023.2}]
By the time change~$\tau'$ from Remark~\ref{29.6.2023.1}, that pauses the movement at the stopping times~$(\tau_1^i)_{\{X_{\tau_1^i}=0\}}$, $i\in\bbn$, $\wt{\bbf}$-optional semimartingales~$Y^1,Y^2$ whose right jumps only take place at the times~$(\tau_1^i)_{\{X_{\tau_1^i}=0\}}$, $i\in\bbn$, can be transformed into the processes~$(Y^j)'_t:=Y_{\tau'_{t+}}$, $t\in[0,T]$, $j\in\{0,1\}$, that are c\`adl\`ag semimartingales under the right-continuous filtration~$\mathcal{F}'_t:=\bigcap_{s>t}\wt{\mathcal{F}}_{\tau'_s}$ for $t\in[0,T)$ and $\mathcal{F}'_T:=\wt{\mathcal{F}}_{\tau'_T}$. By construction of $\tau'$, it follows that $d_{\bbs}((Y^1)',(Y^2)')=\wt{d}_{\bbs}(Y^1,Y^2)$  (cf., e.g., the proof of \cite[Theorem~5.55]{he.wang.yan.1992}). This makes the results below for c\`adl\`ag semimartingales directly applicable to our setting.\\

Based on the analogon of Lemma~\ref{24.4.2023.1}, it is shown in Delbaen and Schachermayer~\cite[proofs of Lemmas~4.7 to 4.11]{delbaen1994general} that after passing to forward convex combinations, the convergence holds in the stronger \'Emery topology (when reading one must abstract from the fact that the semimartingales are stochastic integrals). The proofs are reformulated in Kabanov~\cite{kabanov.1997} in a more abstract setting. 
It seems to be easier to adapt the arguments in \cite{kabanov.1997} to our setting, which is what we want to do in the following. We define the following set of $\bbf$-optional semimartingales:
\beao
\mathcal{X}:=\{\wt{V}^\psi : \psi\in (\bP)^\Pi\ \mbox{such that\ } (\Pi(\psi),\psi)\ \mbox{is $M$-admissible}\}.
\eeao
This means that for $Y\in\mathcal{X}$, $Y_T$ is the gain of an $M$-admissible strategy in the original model if on $\{X_T>0\}$ terminal stock positions are evaluated by $(S,S')$. We list key properties of the set~$\mathcal{X}$:
\begin{itemize}
\item[(i)] By the estimate~$V^{S,S'}_T(\psi) \le V^{\rm cost}_T(\psi)$ and condition~(\ref{12.6.2023.1}),
the set~$\{Y_T : Y\in\mathcal{X}\}$ is $L^0$-bounded.
\item[(ii)] For any $Y^1,Y^2\in\mathcal{X}$, $\lambda\in[0,1]$, there exists $Y^3\in\mathcal{X}$ with $Y^3\ge \lambda Y^1 + (1-\lambda)Y^2$
(this holds since the wealth process is concave in the strategy). 
\item[(iii)] Let $Y^1,Y^2\in\mathcal{X}$ and $(Y^1)^r,(Y^2)^r$ defined as in Lemma~\ref{14.7.2024.1}. For any $[0,1]$-valued $\wt{\bbf}$-predictable processes~$H^1,H^2$ with $H^1 H^2=0$ and $G_i^1,G_i^2\in L^0(\mathcal{F}_{\tau_1^i};[0,1])$, $i\in\bbn$, with $G_i^1 G_i^2=0$, the process 
\beao
H^1\mal (Y^1)^r +\sum_{i=1}^\infty G_i^1 1_{\{X_{\tau_1^i}=0\}} 1_{\zu \tau_1^i,T\zu}\Delta^+Y^1_{\tau_1^i} 
+ H^2\mal (Y^2)^r +\sum_{i=1}^\infty G_i^2 1_{\{X_{\tau_1^i}=0\}} 1_{\zu \tau_1^i,T\zu}\Delta^+Y^2_{\tau_1^i}
\eeao
is the dormant wealth process of a self-financing strategy (not necessarily $M$-admissible).  
\item[(iv)] In the dormant market, one can consider the 
num\'eraire~$N_t:= (1+S_{\tau_t-}) 1_{D_t} + (1+S_{\tau_t}) 1_{\Omega\setminus D_t}$, $t\in[0,T]$,
that is the sum of $1$ and the time-changed semimartingale price system~$S$. For an $M$-admissible strategy~$(\psi^0,\psi)$ with $\psi\in (\bP)^\Pi$, we have that $\wt{V}^\psi_t/N_t\ge -M$ on $\{\tau_t<T\}$ for all $t\in[0,T)$ (for stopping at $\tau_t-$ cf. Lemma~\ref{29.1.2023.2} combined with ${\rm NA}^{nf}$).
We refer to the term of an ``allowable strategy'' introduced in Yan~\cite[Definition~2.4]{yan.1998}. 
\end{itemize}
With these four properties of $\mathcal{X}$, one obtains a $\wt{d}_{\bbs}$-Cauchy sequence along the lines of Kabanov~\cite[proofs of Lemmas~2.3 to 2.8 and Lemma~3.3]{kabanov.1997}. Let us only describe the minor adjustments that are needed. 
By considering the dormant market, one switches from one strategy to another strategy only at frictionless points. Here, one can express the difference of frictionless wealth as multiple of $N$ (cf. item~(iv)) and control it as in \cite{kabanov.1997}. 
Then, during an excursion of the spread away from zero, the number of stocks of a concatenated strategy coincides with some $\psi$, with $(\psi^0,\psi)$ admissible, and the position in the bank account coincides with $\psi^0$ shifted by the difference of frictionless wealth before the excursion. This means that the difference of frictionless wealth coming from past trades is invested in the bank account, which allows to estimate the constant~$M'$ with which the concatenated strategy is $M'$-admissible. On the technical level, the arguments in the proof of Lemma~\ref{24.4.2023.1} have to be repeated to concatenate strategies, using that $(\tau_1^i)_{\{X_{\tau_1^i}=0\}}$, $i\in\bbn$, are stopping times and the left limit process~$N_-$ is locally bounded. In that proof it can also be seen how the maximality of $(C^0,C)$ is used.

To prove the second assertion, we use the fact that the space of semimartingales is complete with regard to the \'Emery topology (see \'Emery~\cite[Theorem~1]{Emery} for c\`adl\`ag semimartingales). Along the lines of the proof of \cite[Theorem~1]{Emery}, one can show that there exists a limiting $\wt{\bbf}$-optional semimartingale~$\wt{V}$ that shares with $\wt{V}^\psi$ the property that right jumps only take place at
$(\tau_1^i)_{\{X_{\tau_1^i}=0\}}$, $i\in\bbn$. We stress that $\wt{V}$ does in general not lie in $\mathcal{X}$, which is generated by bounded strategies. 
\end{proof}
\begin{remark}
In a more recent article, Cuchiero and Teichmann~\cite[Theorem~3.3(ii)]{cuchiero.teichmann.2015} show that in the case of a frictionless market,
Lemma~\ref{2.6.2023.2} holds without passing to (further) forward convex combinations. We leave it as an open problem if their arguments can also be adapted to our setting. For the proof of Theorem~\ref{7.5.2022.2}, this question is not crucial, and the proofs of \cite{kabanov.1997} can 
be more easily adapted. 
\end{remark}

\begin{proof}[Proof of Theorem~\ref{7.5.2022.2}]
So far, we made well-known results on frictionless markets accessible to our model by considering a fictitious dormant market that ignores 
the problems that occur by positive bid-ask spreads. We now turn to the construction of the limiting strategy when the spread does not vanish.
The random vector $(C^0,C)$ and the approximating sequence~$(\vp^{0,n},\vp^n)_{n\in\bbn}$ are still from (\ref{25.4.2023.1}), and we have to show that $(C^0,C)$ is the terminal position of an $M$-admissible strategy. We can assume that $(\vp^{0,n},\vp^n)_{n\in\bbn}$ satisfies (\ref{31.12.2022.1}).\\

{\em Step 1:} By (\ref{12.6.2023.2}), the second part of the ${\rm NUPBR}^{ps}$ condition,  
we can assume w.l.o.g. that 
\beam\label{13.6.2023.1}
\sup_{t\in[0,T]}\sup_{n\in\bbn}V^{\rm cost}_t(\vp^n)<\infty. 
\eeam
If this does not already hold for the original sequence of cost value processes, we pass to forward convex combinations.
(\ref{13.6.2023.1}) allows us to define the finite 
process~$A_t:=\sup_{s\in[0,t]}\sup_{n\in\bbn}V^{\rm cost}_s(\vp^n)$ that dominates all cost processes. Putting this together with the $M$-admissibility of $\vp^n$, the later means that $\vp^{0,n} + M +(\vp^n+M)^+ \un{X}-(\vp^n+M)^- \ov{X} \ge 0$,  we can control the size of the strategies. Namely, we get
\beam\label{12.2.2023.1}
|\vp^n|(\ov{X}-\un{X}) \le A + M + M\ov{X},\quad \forall n\in\bbn.
\eeam
Estimate (\ref{12.2.2023.1}) is used for the case that positions are built up during a frictionless interval but the spread jumps away from zero.  For the case that the portfolio is rebalanced under a positive spread we need another estimate.
The finite process~$A$ is pre-locally bounded, i.e., there exists a sequence of stopping times~$(T_m)_{m\in\bbn}$ with $P(T_m=\infty)\to 1$ for $m\to\infty$ and $V^{\rm cost}(\vp^n) 1_{\auf 0,T_m\auf}\le m$ for all $n\in\bbn$. 
Let us show that 
\beam\label{9.9.2022.1}
\left(\vp^{0,n}+(\vp^n)^+ \ov{X}_- -(\vp^n)^-\un{X}_-\right)1_{\auf 0,T_m\zu}\le m,\quad \forall n,m\in\bbn.
\eeam 
Intuitively, this means that trade cannot increase the cost value of a portfolio.
We have to prove that $\vp^{0,n}_s+(\vp^n_s)^+ \ov{X}_s -(\vp^n_s)^-\un{X}_s\le m$ for all $s<t$ implies that 
$\vp^{0,n}_t+(\vp^n_t)^+ \ov{X}_{t-} -(\vp^n_t)^-\un{X}_{t-}\le m$.
On the set $\{\ov{X}_{t-}>\un{X}_{t-}\}$, $\vp^n_s$ converges to $\vp^n_{t-}$, and the variation of $\vp^n$ on $[s,t)$ vanishes
as $s\uparrow t$ by Proposition~\ref{10.12.2022}(a). 
This implies that
\beam\label{14.2.2023.1}
\vp^{0,n}_{t-}+(\vp^n_{t-})^+ \ov{X}_{t-} -(\vp^n_{t-})^-\un{X}_{t-} \le m.
\eeam
In addition, the LHS of (\ref{14.2.2023.1}) dominates $\vp^{0,n}_t+(\vp^n_t)^+ \ov{X}_{t-} -(\vp^n_t)^-\un{X}_{t-}$.
On $\{\ov{X}_{t-}=\un{X}_{t-}\}$ we can use that 
$V^{S,S'}_s(\vp^n) \le \vp^{0,n}_s+(\vp^n_s)^+ \ov{X}_s -(\vp^n_s)^-\un{X}_s \le m$ for all $s<t$ and  
$V^{S,S'}_s(\vp^n)\to V^{S,S'}_{t-}(\vp^n)=\vp^{0,n}_t+(\vp^n_t)^+ \ov{X}_{t-} -(\vp^n_t)^-\un{X}_{t-}$. 

We can and do choose the sequence~$(T_m)_{m\in\bbn}$ from above such that one has, in addition, that
\beao
\ov{X}_-1_{\auf 0_{\{T_m>0\}}\zu\cup\zu 0,T_m\zu}\le m,\quad \forall m\in\bbn.
\eeao
From (\ref{9.9.2022.1}) we subtract the inequality 
\beao
\vp^{0,n} + M +(\vp^n+M)^+ {\rm ess inf}_{\mathcal{F}_-}\un{X}-(\vp^n+M)^- {\rm ess sup}_{\mathcal{F}_-}\ov{X} \ge 0
\eeao
that holds by Lemma~\ref{29.1.2023.2} and obtain the estimate
\beao
\max\{(\vp^n)^+(\ov{X}_- -{\rm ess inf}_{\mathcal{F}_-}\un{X}), (\vp^n)^-({\rm ess sup}_{\mathcal{F}_-}\ov{X} -\un{X}_-)\} \le m + M + M m\quad\mbox{on}\ \zu 0,T_m\zu.
\eeao
We arrive at
\beam\label{29.1.2023.01}
& & |\vp^n| \wh{X} \le m + M + M m\quad\mbox{on}\ \zu 0,T_m\zu,\nonumber\\ 
& & \mbox{where\ }\wh{X}:= \min(\ov{X}_- -{\rm ess inf}_{\mathcal{F}_-}\un{X},{\rm ess sup}_{\mathcal{F}_-}\ov{X} -\un{X}_-).
\eeam
Estimate (\ref{29.1.2023.01}) is crucial to control the position in the risky asset during an excursion of the spread away from zero.\\

{\em Step 2:} Let $\tau:=\tau_1^i$ be the starting time of an excursion (cf. (\ref{24.3.2023.1})). The spead at $\tau$ can be zero or positive.
Let us show that the end of an excursion can be rewritten as 
\beam\label{14.5.2023.01}
\Gamma(\tau) & := & \inf\{t>\tau : \ov{X}_t=\un{X}_t\ \mbox{or}\ \ov{X}_{t-}=\un{X}_{t-}\}\\
& = & \inf\{t>\tau : \ov{X}_t=\un{X}_t\ \mbox{or}\ \ov{X}_{t-}=\un{X}_{t-}\ \mbox{or}\ \ov{X}_{t-}=({\rm ess inf}_{\mathcal{F}_-}\un{X})_t\ \mbox{or}\ ({\rm ess sup}_{\mathcal{F}_-}\ov{X})_t=\un{X}_{t-}\}.\nonumber
\eeam
Let $\tau_1$ be a predictable stopping time with $\ov{X}_{\tau_1-}=({\rm ess inf}_{\mathcal{F}_-}\un{X})_{\tau_1}$ on $\{\tau_1<\infty\}$.
Since $\un{X}_{\tau_1} \ge ({\rm ess inf}_{\mathcal{F}_-}\un{X})_{\tau_1}$, 
the ${\rm NA}^{nf}$ condition implies that $\un{X}_{\tau_1} = \ov{X}_{\tau_1-}$. This means that a long stock position built up at time $\tau_1-$ can be liquidated for sure at time $\tau_1$. Consequently, we must have that $\ov{X}_{\tau_1}=\un{X}_{\tau_1}$ on $\{\tau_1<\infty\}$ since otherwise the sequence~$\psi^n:=n1_{\auf\tau_1\zu}$, $n\in\bbn$, would violate the ${\rm NUPBR}^{ps}$ condition. This means that at $t=\tau_1$ the first condition $\ov{X}_t=\un{X}_t$ is satisfied as well. 

To complete the proof, we define the debut~$\tau_2:=\inf\{t>\tau : \ov{X}_{t-}=({\rm ess inf}_{\mathcal{F}_-}\un{X})_t\}$ that is a 
(not necessarily predictable) stopping time
with $\auf (\tau_2)_{\{\ov{X}_{\tau_2-}=({\rm ess inf}_{\mathcal{F}_-}\un{X})_{\tau_2}\}}\zu\in\mathcal{P}$. We have to show that $P(\tau_2\ge \Gamma(\tau))=1$. Assume by contradiction that there exists an $\eps>0$ such that $P(\tau_2+\eps<\Gamma(\tau))>0$.
By a section theorem for predictable sets (see, e.g., \cite[Theorem~4.8]{he.wang.yan.1992}) applied to the predictable set~$\{\ov{X}_-={\rm ess inf}_{\mathcal{F}_-}\un{X}\}\cap\auf\tau_2,(\tau_2+\eps)\wedge\Gamma(\tau)\zu$, there exists a predictable stopping time~$\tau_3$ 
with $P(\tau_3<\infty, \tau_2+\eps<\Gamma(\tau))>0$, $\tau_2\le\tau_3\le \tau_2+\eps$ and
$\ov{X}_{\tau_3-}=({\rm ess inf}_{\mathcal{F}_-}\un{X})_{\tau_3}$ on $\{\tau_3<\infty, \tau_2+\eps<\Gamma(\tau)\}$. 
Above, we have shown that this implies that
$\ov{X}_{\tau_3} = \un{X}_{\tau_3}$ on $\{\tau_3<\infty, \tau_2+\eps<\Gamma(\tau)\}$ -- a contradiction to the definition of $\Gamma(\tau)$.
By exactly the same arguments we get rid of the condition $({\rm ess sup}_{\mathcal{F}_-}\ov{X})_t=\un{X}_{t-}$ in (\ref{14.5.2023.01}).\\

In the following, we construct a double sequence of stopping times~$(\tau^{1,N},\tau^{2,N})_{N\in\bbn}$ with which one can exhaust the excursion while keeping the spread away from zero. 
We set $\tau^{1,N}:=(\tau+ 1/N 1_{\{\ov{X}_\tau=\un{X}_\tau\}})\wedge \sigma$. Since the stopping time $(\Gamma(\tau))_{\{\ov{X}_{\Gamma(\tau)-} = \un{X}_{\Gamma(\tau)-}\}}$  
is predictable, it possesses an announcing sequence.
Thus, there exists a sequence~$(\tau^{2,N})_{N\in\bbn}$ with $\wt{\sigma}\le \tau^{2,N}\le \Gamma(\tau)$,  
$\tau^{2,N}<\Gamma(\tau)$ on $\{\ov{X}_{\Gamma(\tau)-} = \un{X}_{\Gamma(\tau)-}\}$, and $\tau^{2,N}\to\Gamma(\tau)$ a.s. as $N\to\infty$, where $\wt{\sigma}$ comes from Assumption~\ref{11.5.2022.1}. 
For fixed $N\in\bbn$, the event~$\{\tau^{2,N}<\tau^{1,N}\}$ can have positive probability, but we have
\beam\label{26.2.2023.1}
\zu \tau,\Gamma(\tau)\zu\setminus\auf (\Gamma(\tau))_{\{\ov{X}_{\Gamma(\tau)-}=\un{X}_{\Gamma(\tau)-}\}}\zu
= \cup_{N\in\bbn} \zu \tau^{1,N},\tau^{2,N}\zu.
\eeam
 For all $N\in\bbn$ we define 
\beam\label{26.2.2023.2}
\eta^N & := & \inf_{t\in [\tau^{1,N}_{\{\ov{X}_\tau=\un{X}_\tau\}}]\cup(\tau^{1,N},\tau^{2,N})\cup [\tau^{2,N}_{\{\tau^{2,N}<\Gamma(\tau)\}}]}\min\left\{\ov{X}_{t-} -({\rm ess inf}_{\mathcal{F}_-}\un{X})_t,\right.\nonumber\\
& & \qquad\qquad\qquad\qquad\qquad\qquad\qquad\qquad\quad \left. ({\rm ess sup}_{\mathcal{F}_-}\ov{X})_t-\un{X}_{t-}, \ov{X}_{t-} - \un{X}_{t-}\right\}
\eeam
with the convention that $\inf\emptyset := \infty$. Let us show that 
\beam\label{21.5.2023.1}
\eta^N>0\quad\mbox{a.s.}
\eeam
The ${\rm NA}^{nf}$ condition and a section theorem for predictable sets (see, e.g., \cite[Theorem~4.8]{he.wang.yan.1992})
imply that $\eta^N\ge 0$ a.s. Now, fix some $N\in\bbn$ and define
\beao
B_n:=\left\{\inf_{t\in [\tau^{1,N}_{\{\ov{X}_\tau=\un{X}_\tau\}}]\cup(\tau^{1,N},\tau^{2,N}]}(\ov{X}_{t-} -({\rm ess inf}_{\mathcal{F}_-}\un{X})_t)\le 1/n\right\},\quad n\in\bbn.
\eeao
In contrast to (\ref{26.2.2023.2}), $t=\Gamma(\tau)$ is included in the infimum that makes $B_n$ predictable. 
Let $\eps>0$. Again by a section theorem for predictable sets, there exists a sequence of predictable stopping times~$(\sigma_n)_{n\in\bbn}$ such 
that $P(B_n\cap\{\sigma_n<\infty\})\ge P(B_n)-\eps 2^{-n}$, $({\rm ess inf}_{\mathcal{F}_-}\un{X})_{\sigma_n}\ge 
\ov{X}_{\sigma_n-} - 1/n$, $\tau^{1,N}\le \sigma_n\le \tau^{2,N}$ on $\{\sigma_n<\infty\}$, and $\tau^{1,N}<\sigma_n$ 
on $\{\ov{X}_\tau>\un{X}_\tau\}$. In addition, $\sigma_n$ can be chosen such that it does not exceed the debut of $B_n$ by more than $2^{-n}$.
It follows that $\un{X}_{\sigma_n}\ge \ov{X}_{\sigma_n-}-1/n$ on $\{\sigma_n<\infty\}$ which means that the
strategies~$\psi^n:= n1_{\auf\sigma_n\zu}$, $n\in\bbn$, are $1$-admissible. The ${\rm NUPBR}^{ps}$ condition implies that the sequence $(n(\ov{X}_{\sigma_n}-\ov{X}_{\sigma_n-})1_{\{\sigma_n<\infty\}})_{n\in\bbn}$ and thus $(n(\ov{X}_{\sigma_n}-\un{X}_{\sigma_n})1_{\{\sigma_n<\infty\}})_{n\in\bbn}$ is $L^0$-bounded. The latter implies that 
$\ov{X}_{\sigma_n}-\un{X}_{\sigma_n}$ converges to zero in probability on $B:=\cap_{n\in\bbn}(B_n\cap\{\sigma_n<\infty\})$  as $n\to\infty$. Consequently, there exists a (deterministic) subsequence~$(n_k)_{k\in\bbn}$ such that $\ov{X}_{\sigma_{n_k}}-\un{X}_{\sigma_{n_k}}\to 0$ on $B$ a.s. as $k\to\infty$.
First, we observe that on $\{\tau^{2,N}<\Gamma(\tau)\}$ the bid-ask spread is bounded away from zero and thus $B\subseteq\{\tau^{2,N}=\Gamma(\tau)\}$ a.s. On the other hand, on $\{\tau^{2,N}=\Gamma(\tau)\}$ we have that $\ov{X}_{\Gamma(\tau)-}>\un{X}_{\Gamma(\tau)-}$. Putting together, we obtain that $B\subseteq\{\sigma_{n_k}=\Gamma(\tau)\ \mbox{for all but finitely many\ }k\}$ a.s. Since the stopping times are close to the debuts of $B_{n_k}$ and $\eps>0$ was arbitrary, we arrive at (\ref{21.5.2023.1}) by symmetry.

We conclude this step with a remark.
Following Example~\ref{15.1.2023.3}, it can happen that $\ov{X}_{\Gamma(\tau)-}=({\rm ess inf}_{\mathcal{F}_-}\un{X})_{\Gamma(\tau)}$ but
$\ov{X}_{\Gamma(\tau)-}>\un{X}_{\Gamma(\tau)-}$. In this case the point~$\Gamma(\tau)$ is still considered to be part of the regime with friction.
However, on $\{\ov{X}_{\Gamma(\tau)-}=({\rm ess inf}_{\mathcal{F}_-}\un{X})_{\Gamma(\tau)}\}$ we must have anyway that 
$\ov{X}_{\Gamma(\tau)-}=\un{X}_{\Gamma(\tau)}=\ov{X}_{\Gamma(\tau)}$, and there are no investment opportunities between $\Gamma(\tau)-$ 
and $\Gamma(\tau)$.\\  
 
{\em Step 3a:} In the following, we consider an interval~$\mathcal{I}^c_i$, $i\in\bbn$, whose left 
endpoint~$(\tau_1^i)_{\{X_{\tau_1^i}=0\}}$ is the starting time of an excursion of the spread~$X$ away from zero but at which the spread is still zero. This is the most tricky case in the proof since we do not have an upper bound for the number of stocks of $M$-admissible strategies on $\{X_{\tau_1^i}=0\}$. Assumption~\ref{11.5.2022.1} is needed to show that with $n\to\infty$ the cost value processes~$V^{\rm cost}(\vp^n)$ cannot increase significantly ``closer and closer'' to $(\tau_1^i)_{\{X_{\tau_1^i}=0\}}$.

Let $\tau:=(\tau_1^i)_{\{X_{\tau_1^i}=0\}}$ be accompanied by the measures~$Q^A$, $A\in\mathcal{P}$, from Assumption~\ref{11.5.2022.1}.
In addition, we fix some $m\in\bbn$ and $\eps\in(0,1)$. By Proposition~\ref{7.6.2022}(b), there exists $\gamma\in(0,\eps)$ such that for every $A\in\mathcal{P}$ and every $Q^A$-supermartingale~$Y$ on $[0,T]$ with $Y_0=0$ and $Y\ge -m-M$, the following implication holds:
\beam\label{21.8.2022.2}
Q^A(\sup_{t\in[0,T]}|Y_t|>\gamma)\le\gamma\quad\implies\quad d_{\bbs}(Y,0)\le \eps^2,
\eeam
where $d_{\bbs}$ denotes the \'Emery distance under $P$ (the proposition is applied under the measures~$Q^A$,
but by (\ref{16.8.2022.1}), $d_{\bbs}$ is small if the \'Emery distance under $Q^A$ is small).
Next, there exists $\delta\in(0,\gamma/3)$ such that 
for every $A\in\mathcal{P}$ and every $Q^A$-supermartingale~$Y$ with $Y_0=0$ and $Y\ge -m-M$
\beam\label{21.8.2022.1}
P(Y_T<-3\delta)\le 3\delta\quad\implies\quad P(\sup_{t\in[0,T]}|Y_t|>\gamma)\vee Q^A(\sup_{t\in[0,T]}|Y_t|>\gamma)\le \gamma.
\eeam
Indeed, by (\ref{16.8.2022.2})/(\ref{16.8.2022.1}) we can switch between the measures $P$ and $Q^A$, and 
one has that $E_{Q^A}(Y^-_T)\le a + (m+M) Q^A(Y_T<-a)$ for all $a\in\bbr_+$. Since at every stopping time the expected loss of a supermartingale exceeds the expected gain, and the former is maximal at maturity, implication~(\ref{21.8.2022.1}) follows by considering the stopping times~$\inf\{t\ge 0 : Y_t>\gamma\}$ and $\inf\{t\ge 0 : Y_t<-\gamma\}$. 

By Lemma~\ref{24.4.2023.1} and $\wt{V}^n_\tau=V^{S,S'}_\tau(\vp^n)=V^{\rm cost}_\tau(\vp^n)$ on $\{\tau<\infty\}$ for all $n\in\bbn$, we have that  
$P(\tau<\infty,\ |V^{\rm cost}_\tau(\vp^{n_1}) - V^{\rm cost}_{\tau}(\vp^{n_2})|>\delta)\le \delta$ for all $n_1,n_2$ large enough.
Since the converging sequence~$(\vp^{0,n}_T,\vp^n_T)_{n\in\bbn}$ is maximal, there exists $n^\eps\in \bbn$ such that for all $n_1,n_2\ge n^\eps$ there exists 
{\em no} $M$-admissible $(\psi^0,\psi)$ with $(\psi^0_T,\psi_T)\ge (\vp^{0,n_1}_T\wedge\vp^{0,n_2}_T,\vp^{n_1}_T\wedge\vp^{n_2}_T)$ and $P(\psi^0_T\ge \vp^{0,n_1}_T+\delta)\ge\delta$ (cf. the end of the proof of Lemma~\ref{24.4.2023.1}).
By $|\vp^{n^\eps}|\le a_{n^\eps}$ (cf. (\ref{31.12.2022.1})), we find an $N^\eps\in\bbn$ such that $P(\sigma>\tau,\ \inf_{t\in[\tau,\tau^{1,N^\eps}]}V^{\rm liq}_t(\vp^{n^\eps}) - V^{\rm cost}_{\tau}(\vp^{n^\eps})<\delta)\le \delta$ and $P(\tau<\sigma\le \tau^{1,N^\eps})\le \eps$. Let us show that
\beam\label{6.5.2023.1}
P(\inf_{t\in[\tau,\tau^{1,N^\eps}]} (V^{\rm cost}_t(\vp^n) - V^{\rm liq}_t(\vp^{n^\eps}))<-\delta)\le \delta\quad\mbox{for all\ } n\ge n^\eps.
\eeam 
We suppose otherwise. Then, one can switch from $\vp^{n^\eps}$ to $\vp^n$ at a stopping time $\wt{\tau}_n$ with $P(\wt{\tau}_n<\infty)>\delta$ and 
$\tau\le \wt{\tau}_n\le\tau^{1,N^\eps}$, $V^{\rm cost}_{\wt{\tau}_n}(\vp^n)-V^{\rm liq}_{\wt{\tau}_n}(\vp^{n^\eps})<-\delta$ on $\{\wt{\tau}_n<\infty\}$ (such a stopping time exists by a section theorem for optional sets, see, e.g., \cite[Theorem~4.7]{he.wang.yan.1992}). This generates a superior strategy~$(\psi^0,\psi)$ from above that is a contradiction. Putting together we obtain that
\beam\label{20.5.2022.1}
P(\inf_{t\in[\tau,\tau^{1,N^\eps}]} (V^{\rm cost}_t(\vp^n) - V^{\rm cost}_{\tau}(\vp^n))<-3\delta)\le 3\delta\quad\mbox{for all\ } n\ge n^\eps.
\eeam 
But, (\ref {20.5.2022.1}) implies that
\beam\label{13.8.2022.1}
P(\inf_{t\in[\tau,\tau^{1,N^\eps}]}( (\vp^n)^+ 1_{\zu \tau,\tau^{1,N^\eps}\zu}\mal \ov{X}_t - (\vp^n)^- 1_{\zu\tau,\tau^{1,N^\eps}\zu}\mal \un{X}_t) <-3\delta) 
\le 3\delta.
\eeam
The processes~$1_{\{T_m>\tau\}}((\vp^n)^+ 1_{\zu\tau,\tau^{1,N^\eps}\zu}\mal\ov{X} - (\vp^n)^-1_{\zu \tau,\tau^{1,N^\eps}\zu}\mal \un{X})$,\ $n\in\bbn$, are bounded from below by $-m-M$.
By Assumption~\ref{11.5.2022.1}, the $n$-th process is a supermartingale with respect to the measure~$Q^{A_n}$ with $A_n:=\{\vp^n\ge 0\}$.  Thus, we can derive from (\ref{13.8.2022.1}) and (\ref{21.8.2022.1}) that
\beam\label{13.8.2022.2}
Q^{A_n}(1_{\{T_m>\tau\}}\sup_{t\in[\tau,\tau^{1,N^\eps}]}|(\vp^n)^+ 1_{\zu \tau,\tau^{1,N^\eps}\zu}\mal \ov{X}_t - (\vp^n)^- 1_{\zu\tau,\tau^{1,N^\eps}\zu}\mal \un{X}_t| >\gamma)\le \gamma.
\eeam
From (\ref{21.8.2022.2}), it follows that
\beam\label{21.8.2022.3}
d_{\bbs}(1_{\{T_m>\tau\}}((\vp^n)^+1_{\zu\tau,\tau^{1,N^\eps}\zu}\mal \ov{X} - 
(\vp^n)^-1_{\zu\tau,\tau^{1,N^\eps}\zu}\mal \un{X}), 0)\le \eps^2\quad\mbox{for all\ }n\ge n^\eps
\eeam
and $d_{\bbs}$ taken with respect to $P$. By $P(\tau<\sigma\le\tau^{1,N^\eps})\le\eps$, $(\ov{X},\un{X})$ coincides with $(S,S')$ on $\auf\tau,\tau^{1,N^\eps}\zu$ with high probability
(cf. Note~\ref{5.5.2023.1}). Together with (\ref{20.5.2022.1}), $3\delta\le\eps$, and (\ref{21.8.2022.3}), we conclude that
\beam\label{9.6.2022.1}
P(1_{\{T_m>\tau\}}\sup_{t\in[\tau,\tau^{1,N^\eps}]}|V^{S,S'}_t(\ph^n)-V^{S,S'}_\tau(\ph^n)|>\eps) \le 3\eps\quad\mbox{for all\ }n\ge n^\eps.
\eeam

{\em Step 3b:} Let $\Gamma(\tau)$ be the end time of the excursion. The numbers $\eps$ and $n^\eps$ are still given by Step~3a 
Analogously to $\tau^{1,N^\eps}$, again using (\ref{31.12.2022.1}), we find an $N'$ such that $P(\tau^{2,N'}<\Gamma(\tau),\ \sup_{t\in[\tau^{2,N'},\Gamma(\tau))}V^{\rm cost}_t(\vp^{n^\eps}) - V^{\rm liq}_{\Gamma(\tau)}(\vp^{n^\eps})1_{\{X_{\Gamma(\tau)-}>0\}} -
V^{\rm liq}_{\Gamma(\tau)-}(\vp^{n^\eps})1_{\{X_{\Gamma(\tau)-}=0\}}>\delta)\le \delta$. 
Since Step~3a is a fortiori true if we reduce $\tau^{1,N}$, we can assume that $N'=N^\eps$.
By similar arguments as for (\ref{6.5.2023.1}), we get
\beam\label{7.5.2023.01}
P(\sup_{t\in[\tau^{2,N^\eps},\Gamma(\tau))}(V^{\rm liq}_t(\vp^n) - V^{\rm cost}_t(\vp^{n^\eps})) > \delta)\le \delta\quad\mbox{for all\ } n\ge n^\eps.
\eeam
If (\ref{7.5.2023.01}) did not hold, then one could switch from $\vp^n$ to $\vp^{n^\eps}$ and improve the terminal position. 
Economically, this means that it could be anticipated if the sequence of strategies performed too bad at the foreseeable end of the excursion, and one would switch to $\vp^{n^\eps}$ before that happens. By contrast, at the beginning of the excursion one would switch from $\vp^{n^\eps}$ to $\vp^n$ {\em after} a bad performance 
of $\vp^n$, cf. (\ref{6.5.2023.1}). Consequently, towards the end of the excursion, the liquidation values of the sequence of strategies can be controlled, instead of the cost values  
as at the beginning of the excursion. Putting together, we obtain
\beam\label{5.8.2022.1}
P(V^{\rm liq}_{\Gamma(\tau)}(\vp^n)1_{\{X_{\Gamma(\tau)-}>0\}} + 
V^{\rm liq}_{\Gamma(\tau)-}(\vp^n)1_{\{X_{\Gamma(\tau)-}=0\}}- \sup_{t\in[\tau^{2,N^\eps},\Gamma(\tau))}V^{\rm liq}_t(\vp^n) < - 3\delta)\le 3\delta
\eeam 
for all $n\ge n^\eps$. 
Since the cost term~$C$ is nondecreasing, (\ref{5.8.2022.1}) implies that
\beam\label{27.8.2022.1}
& & P((\vp^n)^+ 1_{\zu\tau^{2,N^\eps},\Gamma(\tau)\zu\setminus\auf(\Gamma(\tau))_{\{X_{\Gamma(\tau)-}=0\}}\zu}\mal \un{X}_T\nonumber\\ 
& & \quad - (\vp^n)^- 1_{\zu\tau^{2,N^\eps},\Gamma(\tau)\zu\setminus\auf (\Gamma(\tau))_{\{X_{\Gamma(\tau)-}=0\}}\zu}\mal \ov{X}_T 
< - 3\delta)\le 3\delta\qquad\mbox{for all\ }n\ge n^\eps.
\eeam
The processes
\beao
1_{\{T_m>\tau^{2,N^\eps}\}}((\vp^n)^+ 1_{\zu\tau^{2,N^\eps},\Gamma(\tau)\zu\setminus\auf(\Gamma(\tau))_{\{X_{\Gamma(\tau)-}=0\}}\zu}\mal \un{X} - (\vp^n)^- 1_{\zu\tau^{2,N^\eps},\Gamma(\tau)\zu\setminus\auf (\Gamma(\tau))_{\{X_{\Gamma(\tau)-}=0\}}\zu}\mal \ov{X}),
\eeao
$n\in\bbn$, are bounded from below by $-m-M$,
and the $n$-th process is a supermartingale with respect to the measure~$\wt{Q}^{A_n}$ with $A_n:=\{\vp^n\ge 0\}$.  Thus, we can derive from (\ref{27.8.2022.1}) and (\ref{21.8.2022.1}) that
\beam\label{4.8.2022.2}
& & \wt{Q}^{A_n}(1_{\{T_m>\tau^{2,N^\eps}\}}\sup_{t\in[\tau^{2,N^\eps},\Gamma(\tau)]} 
|(\vp^n)^+ 1_{\zu\tau^{2,N^\eps},\Gamma(\tau)\zu\setminus\auf(\Gamma(\tau))_{\{X_{\Gamma(\tau)-}=0\}}\zu}\mal \un{X}_t \nonumber\\
& & \qquad\qquad\qquad -(\vp^n)^- 1_{\zu\tau^{2,N^\eps},\Gamma(\tau)\zu\setminus\auf(\Gamma(\tau))_{\{X_{\Gamma(\tau)-}=0\}}\zu}\mal \ov{X}_t| 
>  \gamma)\le \gamma\quad\mbox{for all\ } n\ge n^\eps.
\eeam
From  (\ref{4.8.2022.2}) and (\ref{21.8.2022.2}), it follows that           
\beam\label{28.8.2022.1}
& & d_{\bbs}(1_{\{T_m>\tau^{2,N^\eps}\}}((\vp^n)^+1_{\zu\tau^{2,N^\eps},\Gamma(\tau)\zu\setminus\auf(\Gamma(\tau))_{\{X_{\Gamma(\tau)-}=0\}}\zu}\mal \un{X} - \nonumber\\
& & \qquad\qquad\qquad\quad (\vp^n)^-1_{\zu\tau^{2,N^\eps},\Gamma(\tau)\zu\setminus\auf(\Gamma(\tau))_{\{X_{\Gamma(\tau)-}=0\}}\zu}\mal \ov{X}), 0)\le \eps^2\quad\mbox{for all\ } n\ge n^\eps
\eeam
and $d_{\bbs}$ taken with respect to $P$.
Since $\tau^{2,N^\eps}\ge\wt{\sigma}$, we have that $(\un{X},\ov{X})=(S,S')$ on \linebreak $\auf\tau^{2,N^\eps},\Gamma(\tau)\zu\setminus\auf (\Gamma(\tau))_{\{X_{\Gamma(\tau)-}=0\}}\zu$, cf. Note~\ref{5.5.2023.1}. Together with (\ref{5.8.2022.1}) and (\ref{28.8.2022.1}) we arrive at 
\beam\label{5.9.2022.1}
& & P(1_{\{T_m>\tau^{2,N^\eps}\}}\sup_{t\in[\tau^{2,N^\eps},\Gamma(\tau))} |V^{S,S'}_t(\vp^n) - V^{S,S'}_{\tau^{2,N^\eps}}(\vp^n)|
\vee (1_{\{ X_{\Gamma(\tau)-}>0\}}|V^{\rm liq}_{\Gamma(\tau)}(\vp^n) - V^{\rm liq}_{\tau^{2,N^\eps}}(\vp^n)|)\nonumber\\
& &  \qquad > \eps) \le  2\epsilon\quad\mbox{for all $n\ge n^\eps$.}
\eeam

{\em Step 4:} Now, we fix $m,N\in\bbn$ and show how a limiting strategy can be constructed on the stochastic interval~$\zu \tau^{1,N},\tau^{2,N}\wedge T_m\zu$ defined in (\ref{26.2.2023.1}).
By (\ref{29.1.2023.01})/(\ref{26.2.2023.2}), we have 
\beam\label{28.5.2022.3}
|\vp^n|1_{\auf\tau^{1,N}\wedge T_m,\tau^{2,N}\wedge T_m\zu} \le (m+M+Mm)/\eta^N=:Y_{m,N}\quad \forall n\in\bbn,
\eeam
and by (\ref{21.5.2023.1}), $Y_{m,N}$ is a finite random variable (not necessarily bounded). 
%
%
First we observe that by the semimartingale property of $S$ and $S'$ the set
\beam\label{28.5.2022.1}
\left\{\psi^+ 1_{\zu \tau^{1,N},\tau^{2,N}\zu}\mal S_T - \psi^- 1_{\zu\tau^{1,N},\tau^{2,N}\zu}\mal S'_T \ :\ \psi\ \mbox{is a predictable process with $|\psi|\le Y_{m,N}$}\right\}
\eeam
is $L^0$-bounded (namely, for a given probability of error, $Y_{m,N}$ can be estimated by a constant). This implies that the set
\beam\label{28.5.2022.2}
{\rm conv}(C^{S,S'}(\vp^n, [\tau^{1,N},\tau^{2,N}\wedge T_m]); n\in\bbn)\quad\mbox{is also $L^0$-bounded.}
\eeam
Note that since $L^0$ is not locally convex, the  $L^0$-boundedness of $(C^{S,S'}(\vp^n, [\tau^{1,N},\tau^{2,N}]))_{n\in\bbn}$ would be potentially weaker.
However, we can consider a convex combination of strategies $\vp^n,\vp^{n+1},\ldots,\vp^{n+k}$ but executing them through different trading accounts. This means that we do not benefit from the subadditivity of $C^{S,S'}$. Since the resulting strategies are still $M$-admissible,
%
%
(\ref{28.5.2022.2}) follows from (\ref{28.5.2022.3}) and (\ref{28.5.2022.1}).

The processes~$\un{X}$, $\ov{X}$, $S$, and $S'$ are c\`adl\`ag, and the paths of $\ov{X}-\un{X}$ are bounded away from zero on $[\tau^{1,N},\tau^{2,N})$. The latter holds since $\tau^{2,N}<\Gamma(\tau)$ on $\{\ov{X}_{\Gamma(\tau)-}=\un{X}_{\Gamma(\tau)-}\}$. Consequently, we have that
$\ov{X}-S\ge (\ov{X}_{\tau^{1,N}}-\un{X}_{\tau^{1,N}})/3>0$ or $S-\un{X}\ge (\ov{X}_{\tau^{1,N}}-\un{X}_{\tau^{1,N}})/3>0$ on an interval with positive random length, and after finitely many analogous estimates we arrive at $\tau^{2,n}$. The same holds for $\ov{X}-S'$ and $S'-\un{X}$.
By (\ref{28.5.2022.2}) and  \cite[Proposition~3.3]{kuehn.molitor.2022}, this implies that $\vp^n$, $n\in\bbn$, are processes of finite variation on $\auf\tau^{1,N},\tau^{2,N}\zu$. Furthermore, $|(\vp^n)^\uparrow_t - (\vp^n)^\uparrow_s  - ((\vp^n)^\downarrow_t - (\vp^n)^\downarrow_s)| \le 2Y_{m,N}$ for all 
$\tau^{1,N}\le s\le t\le\tau^{2,N}\wedge T_m$ and $n\in\bbn$ by (\ref{28.5.2022.3}), we obtain that
${\rm conv}((\vp^n)^\uparrow_{\tau^{2,N}\wedge T_m} - (\vp^n)^\uparrow_{\tau^{1,N}\wedge T_m}; n\in\bbn)$ and ${\rm conv}((\vp^n)^\downarrow_{\tau^{2,N}\wedge T_m} - (\vp^n)^\downarrow_{\tau^{1,N}\wedge T_m}; n\in\bbn)$ are $L^0$-bounded, too.

Now, we can proceed as in Schachermayer~\cite[proof of Theorem~3.4]{schacher.2014}. It is a stochastic version of Helly's classic theorem that shows the existence of a converging subsequence of monotone functions on the real line. 
We repeat only the results that are needed in the present paper.
%
%
On $\auf\tau^{1,N},\tau^{2,N}\wedge T_m\zu$ there exists a predictable process $\vp$ such that after passing to forward convex combinations, $\vp^n\to\vp$ up to evanescence.
If one applies the same construction 
for a larger pair~$(N,m)$ (and thus on a larger subinterval of the excursion), the strategy coincides with $\vp$ on the smaller subinterval
up to evanescence. By $P(T_m=\infty)\to 1$ as $m\to\infty$ and (\ref{26.2.2023.1}), we arrive at a predictable process~$\vp$ with
\beam\label{28.8.2022.2}
\vp^n\to \vp\quad\mbox{on}\ \zu \tau, \Gamma(\tau)\zu\setminus\auf(\Gamma(\tau))_{\{X_{\Gamma(\tau)-}=0\}}\zu\quad \mbox{up to evanescence},\quad n\to\infty,
\eeam
after passing to joint forward convex combinations using a diagonalization argument.\\

{\em Step 5:} In this step, we want to show that $\vp^+ 1_{\zu \tau, \Gamma(\tau)\zu\setminus\auf(\Gamma(\tau))_{\{X_{\Gamma(\tau)-}=0\}}\zu}\in L(S)$ and \linebreak $\vp^- 1_{\zu \tau, \Gamma(\tau)\zu\setminus\auf(\Gamma(\tau))_{\{X_{\Gamma(\tau)-}=0\}}\zu}\in L(S')$, where the semimartingales~$S$ and $S'$ are defined in Note~\ref{5.5.2023.1}. 
Let $\eps>0$. We choose $m\in\bbn$ such that $P(T_m<\infty)\le \eps^2$. The stopping times~$\tau^{1,N^\eps}$, $\tau^{2,N^\eps}$ are from Step~2.  They actually also depend on $m$. Beforehand, we observe that the $d_{\bbs}$-distance of two semimartingales is bounded from above by the probability  that their paths do not coincide. Then, by (\ref{21.8.2022.3}) and the triangle inequality of the metric~$d_{\bbs}$, we obtain
\beao
d_{\bbs}((\vp^{n_1})^+1_{\zu \tau,\tau^{1,N^\eps}\zu}\mal \ov{X}
- (\vp^{n_1})^-1_{\zu \tau,\tau^{1,N^\eps}\zu}\mal \un{X}, 
(\vp^{n_2})^+1_{\zu \tau,\tau^{1,N^\eps}\zu}\mal \ov{X}
- (\vp^{n_2})^-1_{\zu \tau,\tau^{1,N^\eps}\zu}\mal \un{X})
 \le 3\eps^2
\eeao
for all $n_1,n_2\ge n^\eps$. By (\ref{28.8.2022.1}) we have the analogue estimate at the end of the excursion. 
Since $\vp^n 1_{\auf \tau^{1,N^\eps}\wedge T_m,\tau^{2,N^\eps}\wedge T_m\zu} \le Y_{m,N^\eps}$ 
for all $n\in\bbn$, all strategies are bounded by $y\in\bbr_+$ on $\zu \tau^{1,N^\eps},\tau^{2,N^\eps}\zu$ outside the event
$\{T_m<\infty\}\cup\{Y_{m,N^\eps}>y\}$ that does not depend on $n$ and has smaller probability than $2\eps^2$ for $y$ large enough (with a 
corresponding bound of the effect on $d_{\bbs}$). Consequently, we can argue with 
the dominated convergence theorem for stochastic integrals (cf., e.g. \cite[Theorem~12.4.10]{cohen.elliott.2015}) and the pointwise convergence~(\ref{28.8.2022.2}) to deduce that
\beao
& & d_{\bbs}((\vp^{n_1})^+1_{\zu \tau^{1,N^\eps},\tau^{2,N^\eps}\zu}\mal S - (\vp^{n_1})^-1_{\zu \tau^{1,N^\eps},\tau^{2,N^\eps}\zu}\mal S',\\
& & \quad (\vp^{n_2})^+ 1_{\zu \tau^{1,N^\eps},\tau^{2,N^\eps}\zu}\mal S - (\vp^{n_2})^- 1_{\zu \tau^{1,N^\eps},\tau^{2,N^\eps}\zu}\mal S') \le 3\eps^2
\eeao
for $n_1,n_2$ large enough. Since $\eps>0$ was arbitrary and $(\vp^n)^+(\vp^n)^-=0$, we conclude that the sequences
\beao 
((\vp^n)^+ 1_{\zu \tau,\Gamma(\tau)\zu\setminus\auf(\Gamma(\tau))_{\{X_{\Gamma(\tau)-}=0\}}\zu}\mal S)_{n\in\bbn},\quad ((\vp^n)^- 1_{\zu \tau,\Gamma(\tau)\zu\setminus\auf(\Gamma(\tau))_{\{X_{\Gamma(\tau)-}=0\}}\zu}\mal S')_{n\in\bbn}
 \eeao
and a fortiori the sequences
\beao
(((\vp^n)^+\wedge \vp^+)1_{\zu \tau,\Gamma(\tau)\zu\setminus\auf(\Gamma(\tau))_{\{X_{\Gamma(\tau)-}=0\}}\zu}\mal S)_{n\in\bbn},\ 
(((\vp^n)^-\wedge \vp^-)1_{\zu \tau,\Gamma(\tau)\zu\setminus\auf(\Gamma(\tau))_{\{X_{\Gamma(\tau)-}=0\}}\zu}\mal S')_{n\in\bbn}
\eeao
are $d_{\bbs}$-Cauchy. Therefore, together with (\ref{28.8.2022.2}) and $(\vp^n)^+\wedge\vp^+\le \vp^+$,  $(\vp^n)^-\wedge\vp^-\le \vp^-$, we 
are in the position to apply Chou, Meyer, and Stricker~\cite{Chou} (see also \cite[Note~4.4]{kuehn.molitor.2022}) and arrive at
\beam\label{30.8.2022.1}
\vp^+ 1_{\zu\tau,\Gamma(\tau)\zu\setminus\auf(\Gamma(\tau))_{\{X_{\Gamma(\tau)-}=0\}}\zu}\in L(S)\quad\mbox{and}\quad
\vp^- 1_{\zu\tau,\Gamma(\tau)\zu\setminus\auf(\Gamma(\tau))_{\{X_{\Gamma(\tau)-}=0\}}\zu}\in L(S'). 
\eeam

{\em Step 6:} Let us show that $\vp 1_{\zu \tau,\Gamma(\tau)\zu\setminus\auf(\Gamma(\tau))_{\{X_{\Gamma(\tau)-}=0\}}\zu}\in L(\un{X},\ov{X})$
(this states that $\vp$, which is not yet globally defined, satisfies the ``local'' properties~(\ref{27.9.2022.1}), (\ref{27.9.2022.2}), and 
(\ref{23.8.2020.1}) of Definition~\ref{def:DefinitionL} on the interval~$\mathcal{I}^c_i$). As a candidate for the ``optimal'' sequence~$(\psi^N)_{N\in\bbn}\subseteq \bf{b}\mathcal{P}$ we take 
\beao
\psi^N:={\rm median}(-y_N,\vp,y_N) 1_{\zu\tau^{1,N},\tau^{2,N}\wedge T_{m_N}\zu}\in\bP,
\eeao
where $m_N\in\bbn$ is large enough such that 
$P(T_{m_N}<\infty)\le 1/N$ and $y_N\in\bbr_+$ is large enough such that $P(Y_{m_N,N}>y_N)\le 1/N$. 

Let $\eps>0$ and $m$ large enough such that $P(T_m<\infty)\le \eps$. By (\ref{21.8.2022.3}) and (\ref{9.6.2022.1}) we have 
$P(T_m=\infty,\ C^{S,S'}_{\tau^{1,N^\eps}}(\vp^n) - C^{S,S'}_{\tau^{1,N}}(\vp^n) > 2\eps) \le 4\eps$ for all $n\ge n_\eps$ and $N\ge N^\eps$. By the Fatou-type estimate in 
\cite[Proposition~3.11]{kuehn.molitor.2022}, the cost term of the limiting strategy cannot be higher in the sense that
\beam\label{29.5.2023.1}
P(T_m=\infty,\ Y_{m,N}\le y_N,\ C^{S,S'}_{\tau^{1,N^\eps}}(\vp) - C^{S,S'}_{\tau^{1,N}}(\vp) > 3\eps) \le 4\eps\quad\mbox{for all}\ N\ge N^\eps.
\eeam
By (\ref{30.8.2022.1}) and again by the dominated convergence theorem for stochastic integrals, we find $N'\ge N^\eps$ such that $d_{\bbs}(\vp^+1_{\zu \tau,\tau^{1,N'}\zu}\mal S - \vp^-1_{\zu \tau,\tau^{1,N^\eps}\zu}\mal S',0)\le \eps^2$.
Adding up the increments we get that
\beam\label{9.6.2022.2}
P(T_m=\infty,\ Y_{m,N}\le y_N,\ \sup_{t\in[\tau^{1,N},\tau^{1,N'}]}|V^{S,S'}_t(\ph)-V^{S,S'}_{\tau^{1,N}}(\ph)|>4\eps)\le 5\eps \quad\mbox{for all}\ N\ge N'.
\eeam 
By (\ref{28.8.2022.1}) and (\ref{5.9.2022.1}), the analogue estimate holds for the end of the excursion. Since $S=\ov{X}$, $S'=\un{X}$ on $\zu \tau,\sigma\zu$, the strategy~$\psi^N$ does not produce trading costs at time~$\tau^{1,N}$, and we have that $(V^{S,S'}(\psi^N))_{N\in\bbn}$ is a up-Cauchy sequence. Since trading gains and trading costs converge separately, condition~(\ref{27.9.2022.2}) follows from (\ref{27.9.2022.1}) by the arguments in the proof of Proposition~\ref{15.1.2023.4}.\\

Now let $(\wt{\psi}^N)_{N\in\bbn}\subseteq (\bP)^\Pi$ be a competing sequence with $(\wt{\psi}^N)^+\le \vp^+$, $(\wt{\psi}^N)^-\le \vp^-$, and 
$\wt{\psi}^N\to\vp$ pointwise. Since, by Step~5, $\vp^+$, $\vp^-$ are integrable with respect to the semimartingale price systems, it follows again from the dominated convergence theorem for stochastic integrals that $d_{\bbs}(
(\wt{\psi}^N)^+ \mal S - (\wt{\psi}^N)^- \mal S', \vp^+ \mal S - \vp^- \mal S')\to 0$ as $N\to \infty$ and thus
\beam\label{29.5.2023.2}
d_{\bbs}((\wt{\psi}^N)^+ \mal S - (\wt{\psi}^N)^- \mal S', (\psi^N)^+ \mal S - (\psi^N)^- \mal S')\to 0,\quad N\to\infty.
\eeam
On the other hand, again by \cite[Proposition~3.11]{kuehn.molitor.2022}, the cost term of the competing sequence can be estimated from below by  
\beam\label{29.5.2023.3}
\liminf_{N\to\infty} C^{S,S'}(\wt{\psi}^N) \ge C^{S,S'}(\psi^{N'}) - C^{S,S'}_{\tau^{1,N'}}(\psi^{N'})\quad \mbox{for every}\ N'\in\bbn.  
\eeam
By (\ref{29.5.2023.1}), $C^{S,S'}_{\tau^{1,N'}}(\psi^{N'})\to 0$ in probability as $N'\to\infty$. Thus, putting 
(\ref{29.5.2023.2}) and (\ref{29.5.2023.3}) together yields (\ref{23.8.2020.1}), and we arrive at
\beam\label{26.1.2023.1}
\vp 1_{\zu \tau,\Gamma(\tau)\zu\setminus\auf(\Gamma(\tau))_{\{X_{\Gamma(\tau)-}=0\}}\zu}\in L(\un{X},\ov{X})\quad\mbox{for}\ 
\tau=(\tau_1^i)_{\{X_{\tau_1^i}=0\}},\ i\in\bbn.
\eeam
We define the wealth process starting in zero at time $(\tau_1^i)_{\{X_{\tau_1^i}=0\}}$ as 
\beam\label{8.6.2023.1}
V^{1,i} := \uplim_{N\to\infty}V^{S,S'}(\psi^N)\quad\mbox{on}\ \mathcal{I}^c_i=\zu (\tau_1^i)_{\{X_{\tau_1^i}=0\}},\Gamma(\tau_1^i)\zu\setminus\auf (\Gamma(\tau_1^i))_{\{X_{\Gamma(\tau_1^i)-}=0\}}\zu,\quad i\in\bbn
\eeam
(it corresponds to $1_{\mathcal{I}^c_i}\mal V$, but the global wealth process~$V$ in Definition~\ref{def:DefinitionL} is not yet defined).\\ 

{\em Step 7:} By (\ref{12.2.2023.1}), we 
have that for fixed $i\in\bbn$, $(\vp^n_{\Lambda(\sigma_1^i)}1_{\{X_{\Lambda(\sigma_1^i)}>0\}})_{n\in\bbn}$ is $L^0$-bounded. After passing once 
again to forward convex combinations, it converges a.s. to some $\mathcal{F}_{\Lambda(\sigma_1^i)}$-measurable finite random variable~$\psi^i$. By a diagonalization argument, one finds a joint sequence for all $i\in\bbn$ (cf. the paragraph at the end of the proof). 
With this, one can argue as in Step~4, but starting directly at the beginning of the excursion. Step~3a is not needed, and the arguments in Steps~5 and 6 only become simpler. Consequently, we obtain a predictable process~$\vp$ that satisfies
\beam\label{30.5.2023.3}
(\vp-\psi^i) 1_{\zu \tau,\Gamma(\tau)\zu\setminus\auf(\Gamma(\tau))_{\{X_{\Gamma(\tau)-}=0\}}\zu}\in L(\un{X},\ov{X})\quad\mbox{where\ }\tau
:=(\Lambda(\sigma_1^i))_{\{X_{\Lambda(\sigma_1^i)}>0\}},\ i\in\bbn.  
\eeam
On $\zu \tau,\Gamma(\tau)\zu\setminus\auf(\Gamma(\tau))_{\{X_{\Gamma(\tau)-}=0\}}\zu$ the limiting wealth process~$V^{2,i}$ is defined as in (\ref{8.6.2023.1}). Observe that the convergence of the strategies at the starting time of the excursion on $\{X_\tau>0\}$ follows from an analysis of the frictionless intervals whereas at the end the behavior can be explained by considering only the excursion itself.\\ 

{\em Step 8:} Let $\mathcal{I}^f_i:=\auf (\sigma^i_1)_{\{X_{\sigma^i_1-}=0\}}\zu \cup \zu \sigma_1^i,\Lambda(\sigma_1^i)\zu$ be a frictionless interval. It is $(\mathcal{F}_t)_{t\in[0,T]}$-predictable and thus a fortiori $\wt{\bbf}$-predictable, and we have that 
\beam\label{2.6.2023.1}
& & 1_{\mathcal{I}^f_i}\mal (\wt{V}^n)^r = (\vp^n)^+ 1_{\mathcal{I}^f_i}\mal S - (\vp^n)^- 1_{\mathcal{I}^f_i}\mal S'\nonumber\\
& &  + 1_{(\cdot\ge \Lambda(\sigma_1^i))}1_{\{X_{\Lambda(\sigma_1^i)}>0\}}(\Delta\wt{V}^n_{\Lambda(\sigma_1^i)} 
- (\vp^n_{\Lambda(\sigma_1^i)})^+\Delta S_{\Lambda(\sigma_1^i)} + (\vp^n_{\Lambda(\sigma_1^i)})^-\Delta S'_{\Lambda(\sigma_1^i)}), i\in\bbn,
\eeam
where $\wt{V}^n$ is defined in (\ref{11.8.2022.1}) and $(\wt{V}^n)^r$ is its ``c\`adl\`ag part'' in the sense of Lemma~\ref{14.7.2024.1}.
By Lemma~\ref{2.6.2023.2}, $(1_{\mathcal{I}^f_i}\mal (\wt{V}^n)^r)_{n\in\bbn}$ is $\wt{d}_\bbs$-Cauchy and $(\Delta\wt{V}^n_{\Lambda(\sigma_1^i)})_{n\in\bbn}$ is Cauchy with regard to the convergence in probability. Together with the convergence 
of $(\vp^n_{\Lambda(\sigma_1^i)}1_{\{X_{\Lambda(\sigma_1^i)}>0\}})_{n\in\bbn}$ established in Step~7,  
we conclude with (\ref{2.6.2023.1}) that $((\vp^n)^+1_{I^f_i}\mal S - (\vp^n)^-1_{I^f_i}\mal S' )_{n\in\bbn}$ is $\wt{d}_\bbs$-Cauchy
(and since $S$ and $S'$ are c\`adl\`ag $(\mathcal{F}_t)_{t\in[0,T]}$-semimartingales, the sequence is also $d_\bbs$-Cauchy). This means that although at 
the end of the frictionless interval the spread can be positive, we have established a corresponding Cauchy sequence in a purely frictionless market.
By M\'emin's theorem (see \cite[Theorem~V.4]{memin.1980}), there exists a predictable process $\vp$ such that 
$\vp^+ 1_{\mathcal{I}^f_i}\in L(S)$, $\vp^-1_{\mathcal{I}^f_i}\in L(S')$, and 
\beam\label{24.3.2023.2}
d_{\bbs}((\vp^n)^+1_{\mathcal{I}^f_i}\mal S - (\vp^n)^- 1_{\mathcal{I}^f_i}\mal S',
\vp^+1_{\mathcal{I}^f_i}\mal S - \vp^- 1_{\mathcal{I}^f_i}\mal S')\to 0\quad n\to\infty.
\eeam 
In fact, it only follows the existence of a limiting strategy in $L((S,S'))$, the set of two-dimensional predictable processes which are integrable with respect to $(S,S')$. But, by $(\vp^n)^+(\vp^n)^-=0$, the proof (that argues with $\mathcal{M}^2\otimes\mathcal{A}$) shows the stronger assertion above. 
We can and do assume that 
\beam\label{3.6.2023.1}
\vp=\lim_{n\to\infty}\vp^n\quad\mbox{on the set}\ \{\lim_{n\to\infty}\vp^n\ \mbox{exists in\ }\bbr\}\cap\mathcal{I}^f_i\in\mathcal{P}.
\eeam
To see this, define $\wt{\vp}:=\lim_{n\to\infty}\vp^n$ if this limit exists in $\bbr$ and $\wt{\vp}:=\vp$ otherwise. Then, using that $\{\wt{\vp}\not=\vp\}=\cup_{k\in\bbn}\{|\vp^n|\le k\ \forall n\in\bbn\}\cap\{\wt{\vp}\not=\vp\}$ it follows that
the process $1_{\{\wt{\vp}^+\not=\vp^+\}}\mal S - 1_{\{\wt{\vp}^-\not=\vp^-\}}\mal S'$ must be evanescent by (\ref{24.3.2023.2}) and the dominated convergence theorem for stochastic integrals.\\

{\em Step 9:} Putting together the partial constructions of the previous steps, we have a candidate for a global limiting strategy. Summing up:
\begin{align}
\vp:=\begin{cases}
\mbox{according to (\ref{24.3.2023.2})}  & \mbox{on}\ \{X_-=0\}\\
\mbox{according to (\ref{26.1.2023.1})}& \mbox{on}\ \zu(\tau_1^i)_{\{X_{\tau_1^i}=0\}},\Gamma(\tau_1^i)\zu\setminus\auf(\Gamma(\tau_1^i))_{\{X_{\Gamma(\tau_1^i)-}=0\}}\zu,\ i\in\bbn\\
\mbox{according to (\ref{30.5.2023.3})}& \mbox{on}\ \zu(\Lambda(\sigma_1^i))_{\{X_{\Lambda(\sigma_1^i)}>0\}},\Gamma(\Lambda(\tau_1^i))\zu\setminus\auf(\Gamma(\Lambda(\tau_1^i)))_{\{X_{\Gamma(\Lambda(\tau_1^i))-}=0\}}\zu,\ i\in\bbn.
 \end{cases}
\end{align}
We note that by (\ref{3.6.2023.1}), one has that $\psi^i = \vp_{\Lambda(\sigma_1^i)} 1_{\{X_{\Lambda(\tau_1^i)}>0\}}$ a.s. for all $i\in\bbn$ in (\ref{30.5.2023.3}).

Using the process $\wt{V}$ from Lemma~\ref{2.6.2023.2} that is the semimartingale limit in the coarser $\wt{\bbf}$-model and the processes from
(\ref{8.6.2023.1}), we define the $(\mathcal{F}_t)_{t\in[0,T]}$-adapted process $V$ by
\beao
& & V_t := \wt{V}_{\tau_1^i}+ V^{1,i}_t\quad\mbox{for\ }t\in(\tau_1^i,\Gamma(\tau_1^i)),\ X_{\tau_1^i}=0\\
& & V_t:= \wt{V}_{\Lambda(\sigma_1^i)-} + \vp_{\Lambda(\sigma_1^i)}^+\Delta S_{\Lambda(\sigma_1^i)} - \vp_{\Lambda(\sigma_1^i)}^-\Delta S'_{\Lambda(\sigma_1^i)} + V^{2,i}_t\ \mbox{for\ }t\in[\Lambda(\sigma_1^i),\Gamma(\Lambda(\sigma_1^i))), X_{\Lambda(\sigma_1^i)}>0\\
& & V_t:=\wt{V}_t\quad\mbox{otherwise}.
\eeao
By Steps~6 and 7 we have that $V$ satisfies (\ref{27.9.2022.1}), (\ref{27.9.2022.2}), and the corresponding ``optimal'' sequences satisfy (\ref{23.8.2020.1}). Let $\wt{V}^r$ be the ``c\`adl\`ag part'' of $\wt{V}$ as defined in Lemma~\ref{14.7.2024.1}. Since $1_{\mathcal{I}^c_i}\mal V = \Delta^+ \wt{V}_{\tau_1^i} 1_{\{X_{\tau_1^i}=0\}} 1_{\zu \tau_1^i,T\zu}$ and $1_{\mathcal{I}^{fc}_i}\mal V = 1_{\mathcal{I}^{fc}_i}\mal \wt{V}^r$ both on $\{X=0\}$, and $\mathcal{I}^{fc}_i$ is $\wt{\bbf}$-predictable, it follows from the continuity of the integral with respect to $\wt{V}$ (see Lemma~\ref{14.7.2024.1} and its proof) that $V$ satisfies (\ref{3.6.2023.2}). We arrive at $\vp\in L(\un{X},\ov{X})$.\\

{\em Step 10:} Finally, we observe that by (\ref{28.8.2022.2}) and $\vp^n_{\Lambda(\sigma_1^i)}\to \vp_{\Lambda(\sigma_1^i)}$\ a.s. on $\{X_{\Lambda(\sigma_1^i)}>0\}$, Lemma~\ref{28.2.2023.3} is applicable, and thus $(\Pi(\vp),\vp)$ is $M$-admissible. This completes the proof.\\

To construct $\vp$, it was necessary to pass to forward convex combinations of the sequence~$(\vp^n)_{n\in\bbn}$ introduced in (\ref{25.4.2023.1})
at several places in the proof. Let us look at the whole picture. 

Once we pass to forward convex combinations in Lemma~\ref{2.6.2023.2}.
Once again we do it for the cost value processes. Then, we pass to these
combinations for each excursion (Step~4) and each end time of a frictionless interval (Step~7). For the latter two, one again applies a diagonalization argument to obtain a joint sequence of forward convex combinations. 
\end{proof}

We reformulate the implication~$(i)\implies(ii)$ of \cite[Theorem~4.2]{delbaen1994general} to make it directly applicable to our two-dimensional setting: 
\begin{theorem}\label{9.6.2023.1}
If a cone $\mathcal{C}_0\subseteq L^0(\Omega,\mathcal{F},P;\bbr^2)$ is Fatou closed (in the sense of Theorem~\ref{7.5.2022.2}), then 
$\mathcal{C}:=\mathcal{C}_0\cap L^\infty(\Omega,\mathcal{F},P;\bbr^2)$ is $\sigma(L^\infty,L^1)$-closed with respect to the measure space $\wt{\Omega}:=\Omega\times \{0,1\}$, $\wt{\mathcal{F}}:=\mathcal{F}\otimes 2^{\{0,1\}}$,
$\mu:=P\otimes(\delta_0+\delta_1)$, where $\delta_0$, $\delta_1$ denote Dirac measures.
\end{theorem}
\begin{theorem}[$2$-dimensional version of Kreps-Yan]\label{10.6.2023.1}
Let $\mathcal{C}$ be a $\sigma(L^\infty,L^1)$-closed convex cone in $L^\infty(\Omega,\mathcal{F},P;\bbr^2)$ containing $L^\infty(\Omega,\mathcal{F},P;\bbr^2_-)$ and such that 
$\mathcal{C}\cap L^\infty(\Omega,\mathcal{F},P;\bbr^2_+)=\{0\}$. Then, there exists $(G^0,G)\in L^1(\Omega,\mathcal{F},P;(\bbr_+\setminus\{0\})^2)$ with $E(G^0)=1$ such that
$E(C^0G^0 + C G)\le 0$ for all $(C^0,C)\in \mathcal{C}$. 
\end{theorem}
Of course, Theorem~\ref{10.6.2023.1} is completely standard. One may again consider the measure space $\wt{\Omega}:=\Omega\times \{0,1\}$, $\wt{\mathcal{F}}:=\mathcal{F}\otimes 2^{\{0,1\}}$,
$\mu:=P\otimes(\delta_0+\delta_1)$ and apply the one-dimensional proof of ``$\Rightarrow$'' in \cite[Theorem~5.2.2]{delbaen.schachermayer.2006} (noting that only $G^0$ is normalized). 
%
%
\begin{proof}[Proof of Theorem~\ref{25.7.2022.1}]
(i) By Theorems~\ref{7.5.2022.2}, \ref{9.6.2023.1}, and \ref{10.6.2023.1}, there exists $(G^0,G)\in L^1(\Omega,\mathcal{F},P;(\bbr_+\setminus\{0\})^2)$ with $E(\vp^0_T G^0 + \vp_T G)\le 0$ for all 
$(\vp^0,\vp)\in \mathcal{A}$. For the convenience of the reader, we repeat and adjust the arguments in Schachermayer~\cite[Definition~4.1 and Proposition~4.2]{schacher.2014}. One defines c\`adl\`ag $P$-martingales by $Z^0_t = E(G^0 | \mathcal{F}_t)$, $Z_t = E(G | \mathcal{F}_t)$, and sets $dQ/dP:=Z^0_T>$ a.s., $S:=Z/Z^0$. By construction, $S$ is a (true) $Q$-martingale. Let us show that $\un{X}\le S\le \ov{X}$. Assume by contradiction that
there exists a stopping time~$\tau$ with $P(\tau<\infty)>0$ and $Z_\tau/Z^0_\tau>\ov{X}_\tau$ on $\{\tau<\infty\}$. We consider the strategy
$(\vp^0,\vp):=(-(\ov{X}_\tau\wedge 1), 1\wedge (1/\ov{X}_\tau))1_{\zu \tau, T\zu}$ that is bounded in both components and satisfies
\beao
E(\vp^0_T Z^0_T + \vp_T Z_T) = E((\vp^0_{\tau+} Z^0_\tau + \vp_{\tau+} Z_\tau)1_{\{\tau<\infty\}})>0,
\eeao
a contradiction. Analogously, for $Z_\tau/Z^0_\tau<\un{X}_\tau$, one considers the bounded strategy
$(\vp^0,\vp):=(\un{X}_\tau\wedge 1, -(1\wedge (1/\un{X}_\tau)))1_{\zu \tau, T\zu}$.\\

(ii) Let $Q\sim P$, $S$ be a $Q$-martingale, and $\un{S}=\ov{S}=S$. The martingale property implies that $\un{X}=\ov{X}=S$.
Thus, Assumptions~\ref{20.8.2022} and \ref{11.5.2022.1} are automatically satisfied. It remains to show that the model $(S,S)$ satisfies ${\rm NA}^{nf}$ and ${\rm NUPBR}^{ps}$. The $1$-admissibility condition boils down to $\vp^0+\vp S \ge -(1+S)$.
By Kallsen~\cite[Lemma~3.3 and Proposition~3.1]{kallsen.2004}, $V=\vp\mal S$ has to be a $Q$-supermartingale for any admissible strategy. This means  that for every $(\vp^0,\vp)\in\mathcal{A}$ either $P(\vp^0_T+\vp_T S_T=0)=1$ or $P(\vp^0_T+\vp_T S_T<0)>0$. Together with the condition $P(S_T>0)=1$, cf. (\ref{11.6.2023.1}), we have ${\rm NA}^{nf}$.

For $(\vp^0,\vp)\in\mathcal{A}^1_0$ and $a\in\bbr_+$, we define $\tau_a:=\inf\{t\ge 0 : V_t(\vp)>a \}$. Since $S$ is a $Q$-martingale and $V$ is a $Q$-supermartingale, one has
\beao
a Q(\tau_a<\infty)\le E_Q(V_{\tau_a\wedge T}(\vp)^+)\le E_Q(V_{\tau_a\wedge T}(\vp)^-) \le 1+E_Q(S_{\tau_a\wedge T})= 1+S_0.
\eeao
Since the RHS does not depend on $\vp$, the market satisfies (\ref{12.6.2023.1}). 
Since we have a frictionless market satisfying NFLVR with num\'eraire~$1+S$, we can apply \cite[Proposition 3.2 and Lemma~4.5]{delbaen1994general}   
and any ``maximal'' sequence of wealth processes is $d_{up}$-Cauchy. This implies condition~(\ref{12.6.2023.2}), and we arrive at ${\rm NUPBR}^{ps}$.
\end{proof}	

\begin{appendix}
\section{Appendix}\label{5.6.2023.3}

\begin{proposition}\label{23.7.2022.1}
There exists a c\`adl\`ag process~$\un{X}$ with 
$\un{X}_t={\rm ess inf}_{\mathcal{F}_t}\sup_{u\in[t,T]}\un{S}_u$\ a.s. for all $t\in[0,T]$, where
\beam\label{3.6.2023.3}
{\rm ess inf}_{\mathcal{F}_t}\sup_{u\in[t,T]}\un{S}_u := {\rm ess sup}\{Z:\Omega\to \ov{\bbr}:\ Z\ \mbox{is}\ \mathcal{F}_t\mbox{-measurable}\ \mbox{and}\ Z\le \sup_{u\in[t,T]}\un{S}_u\ \mbox{a.s.}\}.
\eeam
\end{proposition}
\begin{proof}
Let $D$ be a dense countable subset of $[0,T]$ with $T\in D$. Let $\wt{X}_t$,\ $t\in D$, be versions of ${\rm ess inf}_{\mathcal{F}_t}\sup_{u\in[t,T]}\un{S}_u$ (for existence and uniqueness up to null sets cf., e.g., \cite[Definition~1.12 and Theorem~1.13]{he.wang.yan.1992}) and
\beao
A:=\{\wt{X}_{t_1} \le \wt{X}_{t_2} + \sup_{u\in[t_1,t_2]}\un{S}_u - \un{S}_{t_2}\quad \mbox{for all\ }t_1,t_2\in D\ \mbox{with}\ t_1<t_2\}.
\eeao
For every $\mathcal{F}_{t_1}$-measurable random variable~$Z$ with $Z\le \sup_{u\in[t_1,T]}\un{S}_u$, the random variable
$Z':=Z + \un{S}_{t_2} - \sup_{u\in[t_1,t_2]}\un{S}_u$ is an $\mathcal{F}_{t_2}$-measurable lower bound of $\sup_{u\in[t_2,T]}\un{S}_u$. This implies that $P(A)=1$, and 
by the usual conditions we can define the adapted process~$\un{X}$ (up to evanescence) by
\beao
\un{X}_t(\omega):=\left\{
\begin{array}{ll} 
\lim_{q\in D, q>t, q\to t} \wt{X}_q(\omega) & \textrm{for}\ \omega\in A\ \mbox{and}\ t<T\\
\un{S}_t(\omega) & \textrm{otherwise}.
\end{array}\right.
\eeao
Namely, for $\omega\in A$, the above limit exists by the right-continuity of $\un{S}$. 
The existence of finite left limits carries over from $\un{S}$ to $\un{X}$, and the process~$\un{X}$ is c\`adl\`ag. It remains to show that for every $t\in[0,T)$, $\un{X}_t$ is a version of ${\rm ess inf}_{\mathcal{F}_t}\sup_{u\in[t,T]}\un{S}_u=:\wh{X}_t$. The estimate $P(\wh{X}_t\le \un{X}_t)=1$ follows from
the same arguments which lead to $P(A)=1$ combined with the right-continuity of $\un{S}$. For the opposite estimate, we follow an argument from Larsson~\cite[Lemma 2.8(iv)]{larsson}: for a sequence~$(q_n)_{n\in\bbn}\subseteq D$ with $q_n>t$ and $q_n\to t$,
one has $\inf_{n\in\bbn}{\rm ess inf}_{\mathcal{F}_{q_n}}\sup_{u\in[t,T]}\un{S}_u=\wh{X}_t$\ a.s., since $\cap_{n\in\bbn} \mathcal{F}_{q_n} = \mathcal{F}_t$. 
Together with $\wt{X}_{q_n}\le {\rm ess inf}_{\mathcal{F}_{q_n}}\sup_{u\in[t,T]}\un{S}_u$, we arrive at $P(\wh{X}_t = \un{X}_t)=1$.
\end{proof}

\begin{proposition}\label{13.1.2023.1}
Let $X$ be a c\`adl\`ag process. Then, there exists a unique (up to evanescence) predictable process~$Y$ with 
$Y_{\tau}={\rm ess inf}_{\mathcal{F}_{\tau-}}X_\tau$\ a.s. for all predictable stopping times~$\tau$ (where ${\rm ess inf}_{\mathcal{F}_{\tau-}}\ldots$ is defined analogously to (\ref{3.6.2023.3}) with the $\sigma$-algebra $\mathcal{F}_{\tau-}$).
\end{proposition}
\begin{proof}
The set $\{\Delta X\not= 0\}$ is thin, i.e., there exists a sequence of stopping times with $\{\Delta X\not= 0\} = \bigcup_{n\in\bbn}\auf T_n\zu$ (cf., e.g., \cite[Theorem~3.32]{he.wang.yan.1992}). 
For each $n\in\bbn$, let $(\sigma_{n,m})_{m\in\bbn}$ be a maximal sequence of predictable stopping times that access the accessible part of the graph of $T_n$ 
(in the sense of \cite[proof of Theorem~4.20]{he.wang.yan.1992}). 
This means that $\auf T_n\zu\setminus\bigcup_{m\in\bbn} \auf \sigma_{n,m}\zu$ cannot be overlapped by the graph of a predictable stopping time with
positive probability. We pass to a single sequence~$(\sigma_k)_{k\in\bbn}$ and obtain the same property for $\{\Delta X\not= 0\}\setminus\bigcup_{k\in\bbn}\auf \sigma_k\zu$. We can and do choose the sequence such that $P(\sigma_{k_1}=\sigma_{k_2}<\infty)=0$ for all $k_1\not=k_2$. Let us define the process
\beam\label{15.1.2023.2}
Y:=\left\{
\begin{array}{ll} 
X_- & \mbox{on}\ (\Omega\times[0,T])\setminus\bigcup_{k\in\bbn}\auf \sigma_k \zu\\
{\rm ess inf}_{\mathcal{F}_{\sigma_k-}}X_{\sigma_k} & \mbox{on}\ \auf \sigma_k\zu\ \mbox{for some\ }k\in\bbn
\end{array}\right.
\eeam
that is obviously predictable. Now, let $\tau$ be a predictable stopping time. We have that 
\beam\label{15.1.2023.1}
Y_\tau = X_{\tau-} 1_{\{\tau\not=\sigma_k\ \forall k\in\bbn\}} 
+ \sum_{k\in\bbn}{\rm ess inf}_{\mathcal{F}_{\sigma_k-}}X_{\sigma_k}1_{\{\tau=\sigma_k\}}\quad\mbox{a.s.}
\eeam 
In (\ref{15.1.2023.1}), $X_{\tau-}$ can be replaced by $X_\tau$ or ${\rm ess inf}_{\mathcal{F}_{\tau-}}X_\tau$ since the sequence~$(\sigma_k)_{k\in\bbn}$ is maximal which implies that $X$ does not jump on this set with positive probability. On $\{\tau=\sigma_k\}\in\mathcal{P}$ 
we have that ${\rm ess inf}_{\mathcal{F}_{\sigma_k-}}X_{\sigma_k}
= {\rm ess inf}_{\mathcal{F}_{\tau-}}X_\tau$ a.s. (we leave it to the reader to check this ``local property'' of ${\rm ess inf}$). 
Together, we obtain that $Y$ satisfies the properties of the proposition. Almost-sure-uniqueness along predictable stopping times follows in the same way. Then, an application of a section theorem for predictable sets (see, e.g., \cite[Theorem~4.8]{he.wang.yan.1992}) 
shows uniqueness of the predictable process up to evanescence. 
\end{proof}

\begin{proposition}\label{7.6.2022}
(a)\ Let $\eps_1,\eps_2,\eps_3>0$ and $Y$ be a supermartingale starting at zero with $-1\le Y \le \eps_1$
and $P(Y_T < -\eps_2)\le \eps_3$. Then, the Doob-Meyer decomposition $Y=M-A$ (i.e., $M$ is a martingale, and $A$ is a nondecreasing, predictable process with $M_0=A_0=0$) satisfies $E(A^2_T)\le (\eps_1+1)(\eps_2+\eps_3)$ and 
$E(M^2_T)\le \eps_1^2 + \eps_2 + \eps_2^2 + 2\eps_3 + 3\eps_1\eps_2 + 3\eps_1\eps_3$, i.e., the second moments are of order $\max(\eps_1,\eps_2,\eps_3)$.\\
(b) There exists a $C\in\bbr_+$ such that for all $\eps>0$ and all supermartingales $\wt{Y}$ starting at zero with $-1\le \wt{Y}$
and $P(\sup_{t\in[0,T]}|\wt{Y}_t| > \eps)\le \eps$, one has $d_{\bbs}(\wt{Y},0)\le C \sqrt{\eps}$.   
\end{proposition}
\begin{proof}
{\em Step 1:} We show assertion~(a) for the discrete time Doob(-Meyer) decomposition along a finite deterministic grid, w.l.o.g. $0,1,\ldots,T-1,T$ with $T\in\bbn$. 
For the first moment of $A_T$, we have $E(A_T)=E(-Y_T) \le \eps_2 + \eps_3$. Following the proof of Meyer~\cite[Theorem II.45]{meyer.1972}, we use this for an estimate of the second moment: 
\beao
E(A^2_T) & \le & 2 E(\sum_{s=1}^T\sum_{t=s}^T (A_s-A_{s-1})(A_t-A_{t-1})) = 2 \sum_{s=1}^T E((A_s-A_{s-1})(A_T-A_{s-1}))\\
 & = & 2 \sum_{s=1}^T E((A_s-A_{s-1})E(A_T-A_{s-1} \ |\ \mathcal{F}_{s-1}))\\
 & = & \sum_{s=1}^T E((A_s-A_{s-1})E(Y_T-Y_{s-1} \ |\ \mathcal{F}_{s-1}))\\
 & \le & \sum_{s=1}^T E((A_s-A_{s-1}) (\eps_1+1)) = (\eps_1+1)E(A_T)\le (\eps_1+1)(\eps_2+\eps_3),
 \eeao
where for the second equation it is used that $A_s-A_{s-1}$ is $\mathcal{F}_{s-1}$-measurable.
This yields $E(M^2_T)\le E((Y^-_T)^2) + E((A_T+\eps_1)^2) \le \eps_2^2 + \eps_3 + (\eps_1+1)(\eps_2+\eps_3) +2(\eps_2 + \eps_3)\eps_1 + \eps_1^2$.
 
{\em Step 2:} The continuous time extension follows by an inspection of the proof of Beiglb\"ock, Schachermayer, and Veliyev~\cite[Theorem~1.1]{beiglbock.2012} in which the mesh of the grid in Step~1 tends to zero. 
The arguments are easier than in the original proof since we know from the estimate in Step~1 that the sequence of terminal values of the discrete time martingales is $L^2(P)$-bounded.
Thus, it is sufficient to apply the Koml\'os theorem for Hilbert spaces, and one obtains $L^2$-convergence to the terminal value of the continuous time martingale part
(cf. \cite[Equation (6)]{beiglbock.2012}).
This implies $L^2$-convergence of the terminal values of the drift parts, and the estimates from Step~1 also hold for the continuous time Doob-Meyer decomposition. This yields (a).

{\em Step 3:} Let  $\xi:=\inf\{t>0 : \wt{Y}_t > \eps\}$. 
Consider the (pre-)stopped process~$Y_t:= \wt{Y}_t 1_{(t<\xi)} 
+ \wt{Y}_{\xi-} 1_{(t\ge \xi)}$. Since $\wt{Y}$ is a supermartingale and 
$\Delta \wt{Y}_{\xi}\ge 0$, $Y$ has to be a supermartingale as well. 
The process~$Y$ is bounded from above by $\eps$ and we have that $P(Y_T<-\eps)\le P(\sup_{t\in[0,T]}|\wt{Y}_t| > \eps)\le \eps$. 
This allows us to apply part~(a) with $\eps_1=\eps_2=\eps_3=\eps$ to $Y$, and we obtain for its Doob-Meyer decomposition that
$E(M^2_T)\le 3\eps + 8\eps^2$ and $E(A^2_T)\le 2\eps+2\eps^2$. By the corollary to Protter~\cite[Theorem~24 of Chapter~IV]{protter2005stochastic}
we have $\sup_{H\in\bP, ||H||_\infty\le 1}\sqrt{E(\sup_{t\in[0,T]}|H\mal Y_t|^2)} \le 3\sqrt{E(M^2_T)}+3\sqrt{E(A^2_T)}$, which implies that
$d_{\bbs}(Y,0)\le \sqrt{27\eps + 72\eps^2} + \sqrt{18\eps + 18\eps^2}$. 
Since the paths of $\wt{Y}$ and $Y$
coincide at least with probability~$1-\eps$, we have that $d_{\bbs}(\wt{Y},Y)\le\eps$, and (b) follows with the triangle inequality of $d_{\bbs}$.
\end{proof}

\end{appendix}

\end{document}